\documentclass[11pt,a4paper]{article}

\usepackage[utf8]{inputenc}
\usepackage[T1]{fontenc}
\usepackage[margin=1in]{geometry}
\usepackage{amsmath, amsthm, amssymb, amsfonts}
\usepackage{mathtools, bm, hyperref, xcolor}

\usepackage{algorithm}
\usepackage{algpseudocode}

\newtheorem{theorem}{Theorem}[section]
\newtheorem{proposition}[theorem]{Proposition}
\newtheorem{lemma}[theorem]{Lemma}
\newtheorem{corollary}[theorem]{Corollary}

\theoremstyle{definition}

\newtheorem{assumption}[theorem]{Assumption}

\theoremstyle{remark}
\newtheorem{remark}[theorem]{Remark}

\usepackage{authblk}

\usepackage{tikz}
\usepackage{pgfplots}
\pgfplotsset{compat=1.18}
\usepgfplotslibrary{fillbetween}

\usepackage{graphicx}
\usepackage{placeins}
\usepackage{caption}

\title{Projected Dynamic Programming for Sequential Quantum State Discrimination}

\author[1, 2]{Jaehun Jeong\thanks{Email: \href{mailto:abin1125@snu.ac.kr}{\texttt{abin1125@snu.ac.kr}}. ORCID: \href{https://orcid.org/0009-0005-6261-0407}{0009-0005-6261-0407}.}}
\author[1, 2]{Donghwa Ji\thanks{ORCID: \href{https://orcid.org/0009-0008-2798-7251}{0009-0008-2798-7251}.}}
\author[3, 2]{Hyunjun Jang}
\author[4, 2]{Kabgyun Jeong\thanks{Corresponding author. Email: \href{mailto:kgjeong6@snu.ac.kr}{\texttt{kgjeong6@snu.ac.kr}}. ORCID: \href{https://orcid.org/0000-0001-7628-7835}{0000-0001-7628-7835}.}}

\affil[1]{College of Liberal Studies, Seoul National University, Seoul 08826, Korea}
\affil[2]{Team QST, Seoul, Korea}
\affil[3]{Department of Mathematics, College of Science, Korea University, Seoul 02841, Korea}
\affil[4]{Institute of Computer Technology, College of Engineering, Seoul National University, Seoul 08826, Korea}

\date{\today}

\begin{document}

\maketitle

\begin{abstract}
Sequential Quantum State Discrimination (SQSD) can be naturally framed as a sequential decision-making problem: at each time step, an agent must decide whether to perform an additional measurement to gather more information or to conclude with an optimal decision based on the current belief. In this paper, we formally cast SQSD into a static-hidden-state Partially Observable Markov Decision Process (POMDP) framework. We demonstrate that this formulation precisely subsumes the conventional minimum-error discrimination (MED) scheme as a special one-step case. Furthermore, we apply a regular grid-based discretization to the continuous belief simplex and approximate the possibly continuous measurement space using a finite library. Then we provide rigorous mathematical bounds on the resulting errors and analyze the computational complexity for both offline planning and online execution. Our analysis confirms that the inherent trade-off between accuracy and complexity, as well as the curse of dimensionality regarding the number of hypotheses, are also prominently observed in the quantum regime. Finally, we provide a working example of binary state discrimination to derive explicit forms of various functions and present numerical simulations for trine state discrimination to visualize the sequential structure of our POMDP-based SQSD.
\end{abstract}

\tableofcontents
\newpage

\section{Introduction}

Quantum state discrimination (QSD) is one of the most basic decision problems in
quantum information theory \cite{BaeKwek2015}. In its standard minimum-error form, one is given a finite
ensemble of candidate states together with a prior distribution and seeks a measurement
and decision rule that maximizes the average probability of correct identification \cite{BaeKwek2015, Helstrom1976QuantumDetection}. In many
settings of practical and conceptual interest, however, the discrimination process is not
naturally one-shot. Rather, the experimenter may choose measurements adaptively, update
their posterior belief after each outcome, and decide at some intermediate stage that
further measurements are no longer worth their cost \cite{TianEtAl2024, LiTanTomamichel2022SequentialQHT}. Related sequential quantum decision problems have also been studied in the setting of adaptive quantum hypothesis testing, where one decides after each test whether to continue sampling or to terminate with a final decision \cite{LiTanTomamichel2022SequentialQHT}. 

While such sequential quantum-testing works focus primarily on hypothesis-testing performance and adaptive stopping strategies, our goal here is different: we formulate sequential quantum discrimination itself as a static-hidden-state finite-horizon POMDP and analyze projected dynamic programming approximations for the resulting belief-space control problem. More specifically, we consider a finite family of
candidate quantum states, indexed by a hidden hypothesis that is drawn once from a prior
distribution and remains fixed throughout the episode. At each stage, the agent may either
perform a measurement action, thereby obtaining a classical observation according to the
Born rule, or stop and declare one of the candidate hypotheses. Because the hidden
hypothesis is static and only the observation process is stochastic, the resulting model fits
naturally into a specialized finite-horizon partially observable Markov decision process
(POMDP) framework with static latent state. The posterior belief over the hidden
hypothesis then serves as a sufficient statistic for the full action--observation history, so
the sequential-QSD problem may be written as a Bellman recursion on the belief simplex \cite{SmallwoodSondik1973}. This formulation is not only conceptually convenient as a concrete example of dynamic programming, but also geometrically informative: in explicit examples, it allows one to visualize how one-step measurement gain, posterior routing, and finite-horizon continuation regions are organized over the belief simplex. In this sense, the POMDP viewpoint provides both an analytical framework and an interpretable geometric picture of sequential quantum discrimination.

A first conceptual point of the paper is that this POMDP formulation does not replace the
conventional QSD objective with an ultimately different one. Rather, the one-step special case
recovers the usual minimum-error discrimination criterion exactly. When the agent performs
a single measurement and then declares optimally, the resulting objective coincides with
the standard measurement-plus-classical-post-processing form of minimum-error QSD, and
hence also with the familiar guess-labeled POVM formulation. In this sense, the present
framework should be understood as a sequential generalization of conventional QSD rather
than a departure from it.

A second point is computational. Although the belief-state formulation is conceptually
clean, exact finite-horizon dynamic programming is generally posed on a continuous belief
simplex and may also involve a rich, possibly continuous measurement family \cite{Monahan1982}. This makes
the exact Bellman recursion computationally demanding \cite{Hauskrecht2000, KaelblingLittmanCassandra1998}. To address this, we
study a projected dynamic-programming architecture in which the belief simplex is replaced
by a finite grid and the measurement family is approximated by a finite library. The
resulting planner performs backward induction on a projected finite computational model
and yields a reusable sequential decision rule for online execution.

The structure of the paper is as follows. Section~2 formulates the sequential-QSD model as a finite-horizon POMDP with static-hidden-state. Section~3 explains how this model relates to conventional one-step minimum-error QSD and establishes the basic consistency results. Section~4 introduces the projected dynamic-programming architecture to solve the sequential-QSD planning problem. Section~5 develops the approximation theory, including belief-space and action-space discretization bounds, explicit regularity constant relations, and a total approximation theorem. Section~6 studies the computational complexity of the projected planner, including both offline planning and online execution. Section~7 presents a binary working example that connects the general framework with the familiar two-state discrimination setting. Section~8 then develops a numerical trine-state example, including belief-simplex geometry, one-step gain, posterior routing, and a compact finite-horizon Bellman analysis. Section~9 concludes with a summary and directions for future work.

\section{Model Formulation for Sequential QSD as a Finite-Horizon POMDP}

We begin by formulating sequential quantum state discrimination (QSD) as a finite-horizon
partially observable Markov decision process (POMDP) with static latent state. The purpose
of this section is to fix the model completely before turning to consistency questions,
approximation analysis, and computational issues.

\subsection{Hidden-state model, action space, and observation law}

Let the hidden hypothesis be denoted by
\[
h \in \{1,\dots,M\},
\]
drawn once from a prior distribution
\[
p_0 = (p_0(1),\dots,p_0(M)).
\]
The corresponding quantum states are
\[
\rho_1,\dots,\rho_M.
\]
Throughout one episode, the hidden index $h$ remains fixed. Thus the latent state space is
\[
S=\{1,\dots,M\},
\]
and the transition kernel is static:
\begin{equation}
\label{eq:static-transition-kernel}
T(h' \mid h,a)=\mathbf{1}(h'=h).
\end{equation}

The action space is the union of two qualitatively different types of actions. First, there are
\emph{measurement actions}, collected in a set
\[
A_{\mathrm{meas}},
\]
each of which consumes one copy and produces a classical outcome according to a POVM.
Second, there are \emph{declaration actions}
\[
\{\delta_1,\dots,\delta_M\},
\]
where $\delta_i$ means that the agent stops and declares hypothesis $i$. Hence the total
action space is
\begin{equation}
\label{eq:total-action-space}
A = A_{\mathrm{meas}} \cup \{\delta_1,\dots,\delta_M\}.
\end{equation}

Measurement actions may be described at two complementary levels. At the most
abstract level, one simply writes
\[
a \in A_{\mathrm{meas}},
\]
thereby emphasizing only the decision-theoretic role of the action inside the
POMDP. This notation is convenient when the measurement set is regarded as a
primitive collection of admissible experiments.

In many physically relevant implementations, however, the measurement-action set
is not merely an abstract set of labels but arises from a structured,
parameterized family. Concretely, one may introduce a parameter space
\[
\Theta,
\]
and regard each measurement action as specified by a parameter
\[
\theta \in \Theta.
\]
A particularly useful POVM subclass is obtained by fixing a reference POVM
\[
F=\{F_o\}_{o\in O}
\]
and conjugating a parameterized unitary operator
\[
U(\theta)
\]
to the corresponding measurement operator. The resulting POVM is then
\begin{equation}
\label{eq:unitary-induced-povm}
E_o(\theta)
:=
U(\theta)^\dagger F_o U(\theta),
\qquad o\in O.
\end{equation}
Thus a control parameter $\theta$ determines a measurement action through the
induced POVM
\[
E(\theta)=\{E_o(\theta)\}_{o\in O}.
\]
However, this is not the only parameterization of the measurement operator; one may alternatively realize measurements by parameterizing the POVM directly via rotation angles.

The abstract notation $a \in A_{\mathrm{meas}}$ and the structured notation
$\theta \in \Theta$ are therefore compatible. The former is preferable when one
wishes to emphasize the generic POMDP structure, whereas the latter becomes
useful when the geometry of the measurement family matters. In particular, the
parameterized viewpoint will play an important role later when we study
action-space discretization by approximating the continuous family $\Theta$ with
a finite sublibrary $\Theta^h = \{\bar{\theta}_1, \cdots, \bar{\theta}_h\}\subset \Theta$.

For each measurement action, whether written abstractly as $a \in A_{\mathrm{meas}}$
or concretely as $\theta \in \Theta$, let
\[
E^{(a)}=\{E_o^{(a)}\}_{o\in O}
\]
denote the associated POVM. In the parameterized case, this simply means
\[
E_o^{(a)} = E_o(\theta)
\]
for the corresponding parameter value. We continue to work with a common finite
outcome alphabet $O$. This common-outcome convention is not restrictive for the
present analysis: action-dependent outcome sets may be embedded into a single
finite alphabet by padding with zero-probability outcomes when necessary.

The observation kernel is induced by the Born rule. More precisely, for each
measurement action $a \in A_{\mathrm{meas}}$, each outcome $o \in O$, and each
hidden hypothesis $i$,
\begin{equation}
\label{eq:born-observation-kernel}
Z(o \mid h=i,a)=\mathrm{Tr}\!\left(E_o^{(a)}\rho_i\right).
\end{equation}
Equivalently, in the parameterized notation one may write
\begin{equation}
\label{eq:born-observation-kernel-parameterized}
Z(o \mid h=i,\theta)=\mathrm{Tr}\!\left(E_o(\theta)\rho_i\right).
\end{equation}
Thus all stochasticity enters through the observation process, whereas the
hidden state itself does not evolve during the episode.

\begin{proposition}[Static-hidden-state POMDP formulation]
\label{prop:static-hidden-pomdp}
The sequential-QSD model described above is a finite-horizon POMDP with latent state
space $S=\{1,\dots,M\}$, action space \eqref{eq:total-action-space}, static transition
kernel \eqref{eq:static-transition-kernel}, and observation kernel
\eqref{eq:born-observation-kernel}.
\end{proposition}

\begin{proof}
The claim follows immediately from the definitions. The hidden hypothesis set provides the
latent state space, the action space contains both information-gathering and terminal
declaration actions, the latent state remains unchanged under every action, and the
observation kernel is given by the Born rule associated with the chosen measurement.
\end{proof}

\begin{remark}
This model is more specialized than a general POMDP. In a general POMDP, uncertainty may
arise both from state transitions and from noisy observations. Here, by contrast, the hidden
hypothesis is static and only the observation process is stochastic. Nevertheless, the standard
belief-state formalism and Bellman recursion remain applicable.
\end{remark}

\subsection{Belief dynamics and Bayesian update}

Because the latent state is not directly observed, the agent must act on the basis of its
posterior belief over the hidden hypothesis \cite{SmallwoodSondik1973}. Let
\begin{equation}
\label{eq:belief-simplex-definition}
\Delta_M
=
\left\{
b \in \mathbb{R}_{\ge 0}^M :
\sum_{i=1}^M b(i)=1
\right\}
\end{equation}
denote the $(M-1)$-simplex of belief states. The initial belief is just the prior,
\begin{equation}
\label{eq:initial-belief}
b_0(i)=p_0(i), \qquad i=1,\dots,M.
\end{equation}

Suppose the current belief is $b \in \Delta_M$ and the agent chooses a measurement action
$a \in A_{\mathrm{meas}}$. Then the probability of observing outcome $o \in O$ is
\begin{equation}
\label{eq:observation-probability-belief}
\Pr(o \mid b,a)
=
\sum_{j=1}^M b(j)\,\mathrm{Tr}\!\left(E_o^{(a)}\rho_j\right).
\end{equation}
If this probability is positive, Bayes' rule yields the posterior update
\begin{equation}
\label{eq:qsd-bayesian-update}
\tau(b,a,o)(i)
=
\frac{
b(i)\,\mathrm{Tr}\!\left(E_o^{(a)}\rho_i\right)
}{
\sum_{j=1}^M b(j)\,\mathrm{Tr}\!\left(E_o^{(a)}\rho_j\right)
},
\qquad i=1,\dots,M.
\end{equation}

\begin{proposition}[Belief sufficiency]
\label{prop:belief-sufficiency}
The posterior belief state is a sufficient statistic for the full action--observation history.
\end{proposition}

\begin{proof}
Because the latent state is static and all uncertainty enters through the observation kernel,
the posterior after one step is obtained from the previous posterior, the chosen action, and
the realized observation by Bayes' rule. Hence the
belief state contains all information relevant for future decision making and is a sufficient
statistic for the full history.
\end{proof}

\subsection{Reward structure and terminal declarations}

The purpose of the sequential decision process is to balance discrimination accuracy against
measurement cost. We therefore assign rewards as follows.

For every measurement action $a \in A_{\mathrm{meas}}$, the agent pays a fixed measurement
cost:
\begin{equation}
\label{eq:measurement-reward}
r(h,a)=-c_{\mathrm{meas}}.
\end{equation}
For every declaration action $\delta_i$, the reward is the indicator of correct identification:
\begin{equation}
\label{eq:declaration-reward}
r(h,\delta_i)=\mathbf{1}(h=i).
\end{equation}

Given a belief state $b \in \Delta_M$, the expected reward of declaring $\delta_i$ is therefore
\[
R_{\mathrm{stop}}(b,\delta_i)
=
\sum_{h=1}^M b(h)\,r(h,\delta_i)
=
b(i).
\]
Hence the optimal stopping reward at belief $b$ is
\begin{equation}
\label{eq:stop-value-definition}
\mathrm{StopVal}(b):=\max_{i} b(i).
\end{equation}
This is the expected reward obtained by stopping immediately and declaring the most likely
hypothesis under the current posterior.

\subsection{Policies, Bellman recursion, and stopping time}

A time-dependent policy is a function
\begin{equation}
\label{eq:time-dependent-policy}
\mu_t : \Delta_M \to A,
\qquad t=0,\dots,H,
\end{equation}
and a finite-horizon policy is the tuple
\begin{equation}
\label{eq:full-policy}
\pi=(\mu_0,\mu_1,\dots,\mu_H).
\end{equation}
The interpretation is that, at stage $t$, the agent observes its current belief $b_t$ and
applies the decision rule $\mu_t(b_t)$.

At the terminal horizon, only stopping is meaningful, so the terminal value is
\begin{equation}
\label{eq:terminal-value-qsd}
V_H(b)=\mathrm{StopVal}(b)=\max_i b(i).
\end{equation}
For $t=0,\dots,H-1$, the Bellman recursion is
\begin{equation}
\label{eq:qsd-bellman-recursion}
V_t(b)
=
\max\left\{
\mathrm{StopVal}(b),\;
\max_{a\in A_{\mathrm{meas}}}
\left(
-c_{\mathrm{meas}}
+
\sum_{o\in O}
\Pr(o\mid b,a)\,
V_{t+1}(\tau(b,a,o))
\right)
\right\} \cite{SmallwoodSondik1973}.
\end{equation}
Thus, at each stage and belief, the agent compares the reward of stopping immediately
against the expected continuation value of performing one more measurement.

Finally, the stopping time induced by a policy $\pi$ is
\begin{equation}
\label{eq:stopping-time-definition}
\tau_\pi
:=
\min\left\{
t\in\{0,\dots,H\}:
\mu_t(b_t)\in\{\delta_1,\dots,\delta_M\}
\right\}.
\end{equation}
This is the first stage at which the policy chooses a declaration action.

\begin{remark}
The decomposition of the action space into measurement actions and declaration actions is
crucial. It allows stopping decisions and information-gathering decisions to be treated
within one unified Bellman recursion. This structural feature will later reappear in both the
projected dynamic programming algorithm and the stopping-time-based online complexity
analysis.
\end{remark}

\section{Positioning and Consistency with Conventional QSD}

Having fixed the POMDP model, we now clarify how this sequential-QSD formulation relates
to conventional one-shot minimum-error QSD. The main point is that the present framework should not be viewed as introducing
a new discrimination criterion. Rather, it provides a sequential decision-theoretic extension
of the standard discrimination problem, with belief updating and stopping made explicit.

Related adaptive-discrimination works have studied Bayesian local measurement procedures and dynamic-programming-based optimization for discriminating between two arbitrary tensor-product quantum states \cite{BrandsenLianStubbsRengaswamyPfister2020AdaptiveTPQS}. Our focus here is different. Rather than optimizing a specific local-measurement architecture for tensor-product inputs, we formulate sequential quantum state discrimination itself as a static-hidden-state finite-horizon POMDP and analyze finite-grid / finite-library approximation error and complexity for the resulting belief-space planner.

The static-hidden-state structure described in Section~2 shows that the model belongs to a
specialized class of finite-horizon POMDPs. At the same time, because the hidden variable is
a quantum-state label and observations are generated through the Born rule, the model also
belongs naturally to the class of sequential Bayesian quantum state discrimination problems.
Thus the framework should be read simultaneously in three ways:
\begin{enumerate}
    \item as a sequential-QSD model,
    \item as a static-hidden-state finite-horizon POMDP,
    \item as an information-acquisition problem with stopping.
\end{enumerate}

The key consistency check is the one-step case. In that regime, the agent performs one
measurement, receives one observation, updates its belief, and then stops with an optimal
declaration. If the resulting one-step objective agrees with the usual minimum-error QSD 
objective, then the present sequential framework can be understood as a genuine extension
of the conventional theory rather than a different optimization problem.

\begin{proposition}[One-step reduction formula]
\label{prop:one-step-reduction}
Consider the one-step special case in which the agent performs exactly one measurement and
then immediately declares a hypothesis. For a fixed prior belief $b \in \Delta_M$ and a
measurement action $a \in A_{\mathrm{meas}}$, let
\[
\Pr(o \mid b,a)
=
\sum_{j=1}^M b(j)\,\mathrm{Tr}\!\left(E_o^{(a)} \rho_j\right)
\]
be the probability of observation $o$, and let the posterior belief after observing $o$ be
\[
b_o(i)
=
\frac{
b(i)\,\mathrm{Tr}\!\left(E_o^{(a)} \rho_i\right)
}{
\sum_{j=1}^M b(j)\,\mathrm{Tr}\!\left(E_o^{(a)} \rho_j\right)
}.
\]
If the agent stops after the observation and declares optimally, then the resulting expected
stopping reward under measurement action $a$ is
\[
J_1(b,a)
:=
\sum_o \Pr(o \mid b,a)\,\max_i b_o(i).
\]
Moreover,
\[
J_1(b,a)
=
\sum_o \max_i \Bigl[
b(i)\,\mathrm{Tr}\!\left(E_o^{(a)} \rho_i\right)
\Bigr].
\]
Consequently, the one-step optimal value is
\[
V_1(b)
:=
\sup_{a \in A_{\mathrm{meas}}} J_1(b,a)
=
\sup_{a \in A_{\mathrm{meas}}}
\sum_o \max_i \Bigl[
b(i)\,\mathrm{Tr}\!\left(E_o^{(a)} \rho_i\right)
\Bigr].
\]
\end{proposition}

\begin{proof}
After observing outcome $o$, the posterior belief is $b_o$. If the agent stops immediately
and chooses the best declaration action, then its reward is
\[
\max_i R_{\mathrm{stop}}(b_o,\delta_i)=\max_i b_o(i).
\]
Averaging over the observation variable yields
\[
J_1(b,a)=\sum_o \Pr(o \mid b,a)\,\max_i b_o(i).
\]
Substituting the posterior formula gives
\[
\sum_o
\left(
\sum_{j=1}^M b(j)\,\mathrm{Tr}\!\left(E_o^{(a)}\rho_j\right)
\right)
\max_i
\left[
\frac{
b(i)\,\mathrm{Tr}\!\left(E_o^{(a)}\rho_i\right)
}{
\sum_{j=1}^M b(j)\,\mathrm{Tr}\!\left(E_o^{(a)}\rho_j\right)
}
\right].
\]
For each fixed outcome $o$, the denominator does not depend on the maximization index $i$,
so it cancels with the prefactor $\Pr(o\mid b,a)$. Hence
\[
J_1(b,a)
=
\sum_o \max_i \Bigl[
b(i)\,\mathrm{Tr}\!\left(E_o^{(a)} \rho_i\right)
\Bigr].
\]
Taking the supremum over $a \in A_{\mathrm{meas}}$ yields the claimed expression for $V_1(b)$.
\end{proof}

\paragraph{Consistency with measurement plus classical post-processing.}
Proposition~\ref{prop:one-step-reduction} shows that, once a measurement action $a$ has
been fixed, the expected stopping reward obtained by optimal posterior-based declaration is
\[
J_1(b,a)
=
\sum_o \max_i \Bigl[
b(i)\,\mathrm{Tr}\!\left(E_o^{(a)} \rho_i\right)
\Bigr].
\]
This is exactly the objective one obtains by viewing the discrimination task as a two-stage
procedure:
\begin{enumerate}
    \item perform a physical measurement,
    \item apply an optimal classical decision rule to the resulting outcome.
\end{enumerate}

To see this more explicitly, fix a POVM
\[
E=\{E_o\}_{o\in O}
\]
and a classical post-processing rule
\[
g:O\to\{1,\dots,M\},
\]
where $g(o)$ is the declared hypothesis after observing outcome $o$. If the prior is
\[
q=(q_1,\dots,q_M),
\]
then the corresponding average success probability is
\begin{align}
P_{\mathrm{succ}}(E,g)
&=
\sum_i q_i \sum_o \mathrm{Tr}(E_o\rho_i)\,
\mathbf{1}\!\bigl(g(o)=i\bigr) \label{eq:success-postprocessing-1}\\
&=
\sum_o \sum_i q_i\,\mathrm{Tr}(E_o\rho_i)\,
\mathbf{1}\!\bigl(g(o)=i\bigr) \label{eq:success-postprocessing-2}\\
&=
\sum_o q_{g(o)}\,\mathrm{Tr}\!\left(E_o\rho_{g(o)}\right).
\label{eq:success-postprocessing-3}
\end{align}
For each fixed outcome $o$, the best decision is therefore to choose the index maximizing
\[
q_i\,\mathrm{Tr}(E_o\rho_i).
\]
Hence
\begin{equation}
\label{eq:optimal-postprocessing}
\max_g P_{\mathrm{succ}}(E,g)
=
\sum_o \max_i \Bigl[q_i\,\mathrm{Tr}(E_o\rho_i)\Bigr].
\end{equation}
In the present notation, the prior $q$ is simply the belief $b$. Therefore the one-step POMDP
objective coincides exactly with the optimal measurement-plus-post-processing objective.

\begin{remark}
The meaning of the previous calculation is conceptually important. The one-step POMDP
formulation does not alter the discrimination criterion; it merely makes the Bayesian update
and the final classical decision rule explicit. In this sense, the present framework refines the
usual measurement-and-guess picture rather than replacing it.
\end{remark}

\paragraph{Consistency with guess-labeled POVM optimization.}
The conventional one-shot Bayesian QSD problem can equivalently be formulated in terms of guess-labeled POVM (detection) operators \cite{YuenKennedyLax1975}
\begin{equation}
\label{eq:conventional-qsd-objective}
\max_{\{M_i\}} \sum_i q_i\,\mathrm{Tr}(M_i\rho_i),
\end{equation}
where $\{M_i\}_{i=1}^M$ is a POVM whose outcomes are already identified with the final
guesses. This formulation is equivalent to the measurement-plus-post-processing form above.

Indeed, given a POVM $E=\{E_o\}$ and a decision rule $g$, define
\begin{equation}
\label{eq:guess-labeled-povm-from-postprocessing}
M_i := \sum_{o:\,g(o)=i} E_o.
\end{equation}
Then each $M_i$ is positive and
\[
\sum_i M_i = \sum_o E_o = I,
\]
so $\{M_i\}$ is again a valid POVM. Moreover,
\begin{align*}
\sum_i q_i\,\mathrm{Tr}(M_i\rho_i)
&=
\sum_i q_i\,
\mathrm{Tr}\!\left(
\sum_{o:\,g(o)=i} E_o \rho_i
\right) \\
&=
\sum_i \sum_{o:\,g(o)=i}
q_i\,\mathrm{Tr}(E_o\rho_i) \\
&=
\sum_o q_{g(o)}\,\mathrm{Tr}\!\left(E_o\rho_{g(o)}\right),
\end{align*}
which is exactly the success probability \eqref{eq:success-postprocessing-3}.

Conversely, given any guess-labeled POVM $\{M_i\}_{i=1}^M$, one may regard it as a
measurement whose outcomes are already the guess labels themselves. In that case the
post-processing rule is trivial:
\[
g(i)=i.
\]
Thus the two formulations are equivalent in both directions.

\begin{proposition}[One-step consistency with conventional QSD]
\label{prop:one-step-consistency}
In the one-step case, the following three optimization problems are equivalent:
\begin{enumerate}
    \item the posterior-based stopping objective of the present POMDP formulation,
    \item measurement followed by optimal classical post-processing,
    \item conventional guess-labeled POVM optimization.
\end{enumerate}
\end{proposition}

\begin{proof}
Proposition~\ref{prop:one-step-reduction} identifies the one-step POMDP objective with the
measurement-plus-post-processing objective \eqref{eq:optimal-postprocessing}. The
construction \eqref{eq:guess-labeled-povm-from-postprocessing} and its converse establish
equivalence between measurement-plus-post-processing and guess-labeled POVM
optimization. Therefore all three formulations agree.
\end{proof}

\begin{remark}
Proposition~\ref{prop:one-step-consistency} should be read as the basic consistency check
for the whole framework. It shows that the present sequential model is not a departure from
standard minimum-error QSD; rather, it is a sequential generalization in which measurement,
Bayesian updating, classical post-processing, and stopping are all made explicit within one
dynamic decision architecture.
\end{remark}

\begin{remark}
Because the action space contains both information-gathering measurement actions and
terminal declaration actions, the model also belongs naturally to the broader class of
sequential information-acquisition problems with stopping. This perspective will become
important later when we interpret both the Bellman recursion and the stopping-time-based
online execution complexity.
\end{remark}

\section{Projected Dynamic Programming}

At this point, it is useful to clarify how the sequential-QSD problem is actually solved from the POMDP point of view. More broadly, the projected planner studied here fits within the larger literature on approximate POMDP solution methods for continuous belief spaces, especially point-based approaches that approximate value functions over selected belief subsets rather than the full simplex \cite{Kurniawati2022POMDPRobotics, Shani2013SurveyPBVI}. Our focus, however, is on a specialized static-hidden-state setting motivated by sequential quantum state discrimination, together with explicit finite-grid / finite-library approximation and complexity guarantees.

For readers less familiar with finite-horizon POMDP methods, the overall procedure should be understood as consisting of two distinct phases. The first phase is \emph{offline planning}. Before any actual experimental run takes place, one solves a dynamic program on the belief space and computes a finite-horizon policy
\[
\pi = (\mu_0,\mu_1,\dots,\mu_H),
\]
where each decision rule $\mu_t$ maps the current belief state to either a measurement action in $A_{\mathrm{meas}}$ or a declaration action in $\{\delta_1,\dots,\delta_M\}$. The purpose of this offline phase is to determine, for every relevant belief state and every remaining time step, whether it is better to stop immediately or to continue measuring.

The second phase is \emph{online execution}. Once the policy has been computed, the agent no longer solves an optimization problem from scratch during the actual run. Instead, it tracks the current belief state, consults the precomputed decision rule $\mu_t$, performs the recommended action, updates the belief after observing the measurement outcome, and repeats this process until a declaration action is selected. Thus the heavy computation is front-loaded into the offline stage, while the online stage follows only one realized trajectory.

In principle, the Bellman recursion is posed on a continuous belief space, namely the simplex $\Delta_M$, and the admissible measurement family may also be richer than a finite set. For numerical computation, however, we adopt a finite-grid ($B$) approximation of the belief simplex $\Delta_M$ \cite{Lovejoy1991}, together with a finite-library ($\Theta^h = \{\bar{\theta}_1, \cdots, \bar{\theta}_h\}$) approximation of the action space $\Theta$. The resulting planner therefore does not solve the exact continuous problem directly, but instead carries out backward induction on this projected finite computational model \cite{PineauGordonThrun2006}.

Even when the current belief $b$ lies in $B$, the exact posterior $\tau(b,a,o)$ generated after action $a$ and observation $o$ need not belong to $B$. We therefore apply a projection map
\[
\mathrm{Proj}_B : \Delta_M \to B
\]
to return the updated belief to the finite computational state space. By repeatedly combining this projected belief update with backward induction over the finite measurement library, the planner constructs approximate value and policy tables on the grid.

\begin{center}
    \captionsetup{type=algorithm}
    \captionof{algorithm}{Offline planning for projected sequential QSD}
    \label{alg:projected-planner}
\end{center}

\begin{algorithmic}[1]
\State \textbf{Input:} horizon $H$, belief grid $B$, measurement set $A_{\mathrm{meas}}$, cost $c_{\mathrm{meas}}$
\State Define
\[
\mathrm{StopVal}(b)=\max_i b(i), \qquad
\mathrm{ObsProb}(o \mid b,a)=\sum_i b(i)\,\mathrm{Tr}\!\left(E_o^{(a)}\rho_i\right)
\]
\State Define Bayesian update
\[
\tau(b,a,o)(i)=
\frac{b(i)\,\mathrm{Tr}\!\left(E_o^{(a)}\rho_i\right)}
{\sum_j b(j)\,\mathrm{Tr}\!\left(E_o^{(a)}\rho_j\right)}
\]
\For{all $b \in B$}
    \State $V_H(b) \gets \mathrm{StopVal}(b)$
    \State $\mu_H(b) \gets \delta_{\arg\max_i b(i)}$
\EndFor
\For{$t = H-1, H-2, \dots, 0$}
    \For{all $b \in B$}
        \State $v_{\mathrm{stop}} \gets \mathrm{StopVal}(b)$
        \State $a_{\mathrm{stop}} \gets \delta_{\arg\max_i b(i)}$
        \State $v_{\mathrm{meas}}^\star \gets -\infty$
        \State $a_{\mathrm{meas}}^\star \gets \emptyset$
        \For{all $a \in A_{\mathrm{meas}}$}
            \State $q(a) \gets -c_{\mathrm{meas}}$
            \For{all $o \in O$}
                \State $p_o \gets \mathrm{ObsProb}(o \mid b,a)$
                \If{$p_o > 0$}
                    \State $b^+ \gets \tau(b,a,o)$
                    \State $\hat b^+ \gets \mathrm{Proj}_B(b^+)$
                    \State $q(a) \gets q(a) + p_o\,V_{t+1}(\hat b^+)$
                \EndIf
            \EndFor
            \If{$q(a) > v_{\mathrm{meas}}^\star$}
                \State $v_{\mathrm{meas}}^\star \gets q(a)$
                \State $a_{\mathrm{meas}}^\star \gets a$
            \EndIf
        \EndFor
        \If{$v_{\mathrm{stop}} \ge v_{\mathrm{meas}}^\star$}
            \State $V_t(b) \gets v_{\mathrm{stop}}$
            \State $\mu_t(b) \gets a_{\mathrm{stop}}$
        \Else
            \State $V_t(b) \gets v_{\mathrm{meas}}^\star$
            \State $\mu_t(b) \gets a_{\mathrm{meas}}^\star$
        \EndIf
    \EndFor
\EndFor
\end{algorithmic}

Algorithm~\ref{alg:projected-planner} is consistent with the notation introduced earlier. The belief state is denoted by $b \in \Delta_M$, the measurement actions belong to $A_{\mathrm{meas}}$, the Bayesian posterior update is written as $\tau(b,a,o)$, and the planner computes both the value function $V_t(b)$ and the decision rule $\mu_t(b)$ at each stage. In particular, the terminal assignment
\[
V_H(b)=\max_i b(i), \qquad \mu_H(b)=\delta_{\arg\max_i b(i)}
\]
matches the stopping logic of the sequential-QSD formulation: once the horizon is exhausted, the agent must declare the most likely hypothesis.

The backward recursion has a simple interpretation. At each stage $t$ and each grid belief $b$, the planner first evaluates the \emph{stopping branch}, namely the reward obtained by declaring immediately according to the largest posterior coordinate. It then evaluates the \emph{continuation branch} for every measurement action $a \in A_{\mathrm{meas}}$. For a fixed measurement action, the planner averages the next-stage value over all possible outcomes $o \in O$, weighting each outcome by its probability and projecting the resulting posterior belief back to the grid. The quantity $q(a)$ is therefore the projected continuation value of choosing measurement $a$ at belief $b$. After scanning all measurement actions, the planner compares the best measurement value against the stopping value and stores whichever branch is better. In this way, the algorithm constructs both an approximate value table and an approximate policy table over the grid.

Once this offline computation has been completed, online execution is comparatively simple. Starting from the initial prior belief $b_0$, the agent evaluates the current decision rule at the belief $b_t$. If the output is a declaration action $\delta_i$, the process stops and hypothesis $i$ is reported. If the output is a measurement action $a_t \in A_{\mathrm{meas}}$, the agent performs that measurement, receives an observation $o_{t+1}$, and updates the belief according to
\[
b_{t+1} = \tau(b_t,a_t,o_{t+1}).
\]
The process then continues at time $t+1$ with the updated belief.

\section{Approximation Error and Lipschitz Constant Bounds}

We now turn to the approximation-theoretic content of the projected dynamic programming. The previous section introduced a computational planner based on a finite belief grid and a finite measurement library. The natural mathematical question is therefore: how far can the resulting projected value function deviate from the exact Bellman value function defined on the continuous belief simplex? This is a central issue in point-based and belief-set-based POMDP approximations, where approximate planners are constructed from finitely many representative beliefs rather than from the full continuous simplex \cite{PineauGordonThrun2006,SpaanVlassis2005Perseus, Shani2013SurveyPBVI}.

In the following subsections, we first isolate the effects of belief-space
discretization and action-space discretization separately. We then combine the
resulting estimates into a total approximation bound for the fully projected
finite-grid finite-library planner.

\subsection{Belief-Space Discretization Error}

Recall that, once a finite measurement library $A_{\mathrm{meas}}$ has been fixed, the exact finite-library Bellman recursion is
\begin{equation}
\label{eq:exact-finite-library-terminal}
V_H(b) := \mathrm{StopVal}(b),
\end{equation}
and, for $t = 0,\dots,H-1$,
\begin{equation}
\label{eq:exact-finite-library-recursion}
V_t(b)
:=
\max\left\{
\mathrm{StopVal}(b),\;
\max_{a \in A_{\mathrm{meas}}}
\left(
-c_{\mathrm{meas}}
+
\sum_{o \in O}
\mathrm{ObsProb}(o \mid b,a)\,
V_{t+1}(\tau(b,a,o))
\right)
\right\}.
\end{equation}

The projected finite-library Bellman recursion is defined only on the grid $B \subset \Delta_M$:
\begin{equation}
\label{eq:projected-finite-library-terminal}
\hat V_H(b) := \mathrm{StopVal}(b), \qquad b \in B,
\end{equation}
and, for $t = 0,\dots,H-1$,
\begin{equation}
\label{eq:projected-finite-library-recursion}
\hat V_t(b)
:=
\max\left\{
\mathrm{StopVal}(b),\;
\max_{a \in A_{\mathrm{meas}}}
\left(
-c_{\mathrm{meas}}
+
\sum_{o \in O}
\mathrm{ObsProb}(o \mid b,a)\,
\hat V_{t+1}(\mathrm{Proj}_B(\tau(b,a,o)))
\right)
\right\},
\,\, b \in B.
\end{equation}

To quantify the resolution of the belief grid, define the projection radius
\begin{equation}
\label{eq:grid-radius}
\delta_B
:=
\sup_{x \in \Delta_M}
\|x-\mathrm{Proj}_B(x)\|_\infty .
\end{equation}

\begin{theorem}[Belief-space discretization error bound on belief grid]
\label{thm:belief-discretization}
Assume that the exact finite-library value functions $V_t$ defined by \eqref{eq:exact-finite-library-terminal}--\eqref{eq:exact-finite-library-recursion} are Lipschitz on $\Delta_M$ with respect to $\|\cdot\|_\infty$, namely,
\begin{equation}
\label{eq:belief-lipschitz-assumption}
|V_t(b)-V_t(b')|
\le
L_t\|b-b'\|_\infty,
\qquad
\forall b,b' \in \Delta_M,
\end{equation}
for some constants $L_t \ge 0$. Then, for every $t=0,\dots,H$,
\begin{equation}
\label{eq:belief-discretization-bound}
\|V_t-\hat V_t\|_{\infty,B}
:=
\sup_{b \in B}|V_t(b)-\hat V_t(b)|
\le
\delta_B \sum_{s=t+1}^{H} L_s.
\end{equation}
\end{theorem}

\begin{proof}
Define the error sequence
\[
e_t := \|V_t-\hat V_t\|_{\infty,B}
= \sup_{b \in B}|V_t(b)-\hat V_t(b)|.
\]
At terminal time $t=H$, both value functions agree on the grid:
\[
V_H(b)=\mathrm{StopVal}(b)=\hat V_H(b), \qquad b \in B.
\]
Hence
\[
e_H = 0.
\]

Fix $t \in \{0,\dots,H-1\}$ and let $b \in B$. Introduce the exact and projected continuation values
\[
Q_t(b,a)
:=
-c_{\mathrm{meas}}
+
\sum_{o \in O}
\mathrm{ObsProb}(o \mid b,a)\,
V_{t+1}(\tau(b,a,o)),
\]
and
\[
\hat Q_t(b,a)
:=
-c_{\mathrm{meas}}
+
\sum_{o \in O}
\mathrm{ObsProb}(o \mid b,a)\,
\hat V_{t+1}(\mathrm{Proj}_B(\tau(b,a,o))).
\]
Then
\[
V_t(b)=\max\left\{\mathrm{StopVal}(b),\max_{a \in A_{\mathrm{meas}}}Q_t(b,a)\right\},
\]
\[
\hat V_t(b)=\max\left\{\mathrm{StopVal}(b),\max_{a \in A_{\mathrm{meas}}}\hat Q_t(b,a)\right\}.
\]
Since both expressions share the same stopping term, we may use the elementary inequality
\[
|\max\{\alpha,x\}-\max\{\alpha,y\}| \le |x-y|,
\]
which yields
\[
|V_t(b)-\hat V_t(b)|
\le
\left|
\max_{a \in A_{\mathrm{meas}}}Q_t(b,a)
-
\max_{a \in A_{\mathrm{meas}}}\hat Q_t(b,a)
\right|.
\]
Applying the inequality
\[
\left|
\max_a u_a-\max_a v_a
\right|
\le
\max_a |u_a-v_a|,
\]
we obtain
\begin{equation}
\label{eq:max-step}
|V_t(b)-\hat V_t(b)|
\le
\max_{a \in A_{\mathrm{meas}}}|Q_t(b,a)-\hat Q_t(b,a)|.
\end{equation}

Now fix $a \in A_{\mathrm{meas}}$. By adding and subtracting
\[
V_{t+1}(\mathrm{Proj}_B(\tau(b,a,o)))
\]
inside each summand, we obtain
\begin{align*}
|Q_t(b,a)-\hat Q_t(b,a)|
&\le
\sum_{o \in O}
\mathrm{ObsProb}(o \mid b,a)
\Bigl(
|V_{t+1}(\tau(b,a,o))-V_{t+1}(\mathrm{Proj}_B(\tau(b,a,o)))| \\
&\hspace{4.2cm}
+
|V_{t+1}(\mathrm{Proj}_B(\tau(b,a,o))) - \hat V_{t+1}(\mathrm{Proj}_B(\tau(b,a,o)))|
\Bigr).
\end{align*}
By the Lipschitz assumption \eqref{eq:belief-lipschitz-assumption} and the definition of $\delta_B$,
\[
|V_{t+1}(\tau(b,a,o))-V_{t+1}(\mathrm{Proj}_B(\tau(b,a,o)))|
\le
L_{t+1}\delta_B.
\]
Moreover, since $\mathrm{Proj}_B(\tau(b,a,o)) \in B$,
\[
|V_{t+1}(\mathrm{Proj}_B(\tau(b,a,o))) - \hat V_{t+1}(\mathrm{Proj}_B(\tau(b,a,o)))|
\le
e_{t+1}.
\]
Therefore,
\[
|Q_t(b,a)-\hat Q_t(b,a)|
\le
\sum_{o \in O}
\mathrm{ObsProb}(o \mid b,a)\,
\bigl(L_{t+1}\delta_B + e_{t+1}\bigr)
=
L_{t+1}\delta_B + e_{t+1}.
\]
Combining this with \eqref{eq:max-step}, we conclude that
\[
|V_t(b)-\hat V_t(b)|
\le
L_{t+1}\delta_B + e_{t+1}.
\]
Taking the supremum over $b \in B$ gives the recursion
\[
e_t \le L_{t+1}\delta_B + e_{t+1}.
\]
Since $e_H=0$, backward iteration yields
\[
e_t \le \delta_B \sum_{s=t+1}^{H} L_s,
\]
which is exactly \eqref{eq:belief-discretization-bound}.
\end{proof}

\begin{corollary}[Uniform Lipschitz case]
\label{cor:belief-discretization-uniform}
If, in addition, $L_s \le L$ for all $s=t+1,\dots,H$, then
\begin{equation}
\label{eq:belief-discretization-bound-uniform}
\|V_t-\hat V_t\|_{\infty,B}
\le
(H-t)L\,\delta_B.
\end{equation}
\end{corollary}

\begin{proof}
Apply Theorem~\ref{thm:belief-discretization} and bound the sum by
\[
\sum_{s=t+1}^{H} L_s \le (H-t)L.
\]
\end{proof}

\begin{corollary}[Belief-space discretization error bound for arbitrary belief state]
\label{cor:belief-discretization-final-bound}
    For an arbitrary belief $b \in \Delta_M$, let us define the projected estimator
    \begin{equation}
        \tilde{V}_t(b) := \hat{V}_t(\operatorname{Proj}_B(b)).
    \end{equation}
    Then,
    \begin{equation}
        |V_t(b) - \tilde{V}_t(b)| \leq \delta_B \sum_{s=t}^H L_s.
    \end{equation}
\end{corollary}

\begin{proof}
    Due to the triangle inequality,
    \begin{equation}
        |V_t(b) - \tilde{V}_t(b)| \leq |V_t(b) - V_t(\operatorname{Proj}_B(b))| + |V_t(\operatorname{Proj}_B(b)) - \hat{V}_t(\operatorname{Proj}_B(b))|.
    \end{equation}
    Then each term reduces to 
    \begin{equation}
        |V_t(b) - \tilde{V}_t(b)| \leq L_t\delta_B + \delta_B \sum_{s = t + 1}^H L_s,
    \end{equation}
    which proves the claim.
\end{proof}

Theorem~\ref{thm:belief-discretization} has a simple interpretation. Each time the projected planner replaces an exact posterior belief by a nearby grid point, it incurs a local geometric error of size at most $\delta_B$. The impact of that perturbation on future return is then amplified by the sensitivity of the next-stage value function, namely by the Lipschitz constant $L_{t+1}$. The total error is therefore a horizon-wise accumulation of \emph{resolution times sensitivity}. In this sense, the theorem says that belief projection is harmless whenever the value landscape varies smoothly across the simplex.

Theorem~\ref{thm:belief-discretization} reduces the approximation question to a regularity question: how large can the belief-Lipschitz constants $L_t$ be? The stopping branch itself is benign, but the continuation branch contains the posterior map
\[
\tau(b,a,o),
\]
which is a rational function of the belief variable and may become sensitive if its denominator behaves poorly. We therefore next derive sufficient conditions under which the exact finite-library value functions are Lipschitz on $\Delta_M$ with respect to $\|\cdot\|_\infty$.

\subsection{Lipschitz Continuity of the Exact Value Function}

To control the constants $L_t$ recursively, we assume that both the observation probability and the posterior map are regular in the belief variable.

\begin{assumption}[Belief-regularity assumptions]
\label{ass:belief-regularity}
Assume that there exist constants $C_P^b \ge 0$ and $C_\tau^b \ge 0$ such that, for every measurement action $a \in A_{\mathrm{meas}}$, every outcome $o \in O$, and all beliefs $b,b' \in \Delta_M$,
\begin{equation}
\label{eq:obsprob-belief-lipschitz}
|\mathrm{ObsProb}(o \mid b,a)-\mathrm{ObsProb}(o \mid b',a)|
\le
C_P^b \|b-b'\|_\infty,
\end{equation}
and
\begin{equation}
\label{eq:posterior-belief-lipschitz}
\|\tau(b,a,o)-\tau(b',a,o)\|_\infty
\le
C_\tau^b \|b-b'\|_\infty,
\end{equation}
whenever both posteriors are well-defined.
\end{assumption}

\begin{remark}
Assumption~\ref{ass:belief-regularity} separates the easy part from the delicate part. The observation probability is linear in the belief variable, so \eqref{eq:obsprob-belief-lipschitz} is typically immediate. The substantive requirement is \eqref{eq:posterior-belief-lipschitz}, because the Bayesian posterior is a normalized rational map. In concrete models, this posterior regularity is often obtained from a nondegeneracy condition preventing observation probabilities from becoming arbitrarily small. 
\end{remark}

\begin{lemma}[Lipschitz continuity of the continuation value in the belief variable]
\label{lem:continuation-belief-lipschitz}
Assume Assumption~\ref{ass:belief-regularity}. Let $t \in \{0,\dots,H-1\}$, and suppose that $V_{t+1}$ is bounded on $\Delta_M$ and $L_{t+1}$-Lipschitz with respect to $\|\cdot\|_\infty$. Then, for every fixed $a \in A_{\mathrm{meas}}$, the continuation value
\[
Q_t(\cdot,a):\Delta_M \to \mathbb{R},
\qquad
Q_t(b,a)
=
-c_{\mathrm{meas}}
+
\sum_{o \in O}
\mathrm{ObsProb}(o \mid b,a)\,
V_{t+1}(\tau(b,a,o))
\]
is Lipschitz on $\Delta_M$. More precisely,
\begin{equation}
\label{eq:continuation-belief-lipschitz}
|Q_t(b,a)-Q_t(b',a)|
\le
\widetilde L_t(a)\|b-b'\|_\infty,
\qquad
b,b' \in \Delta_M,
\end{equation}
where one may take
\begin{equation}
\label{eq:continuation-belief-lipschitz-constant}
\widetilde L_t(a)
:=
\sum_{o \in O}
\left(
C_P^b \|V_{t+1}\|_\infty
+
L_{t+1} C_\tau^b
\right).
\end{equation}
In particular,
\begin{equation}
\label{eq:continuation-belief-lipschitz-constant-coarse}
\widetilde L_t(a)
\le
|O|
\left(
C_P^b \|V_{t+1}\|_\infty
+
L_{t+1} C_\tau^b
\right).
\end{equation}
\end{lemma}

\begin{proof}
Fix $a \in A_{\mathrm{meas}}$, and for each outcome $o \in O$ define
\[
f_o(b)
:=
\mathrm{ObsProb}(o \mid b,a)\,
V_{t+1}(\tau(b,a,o)).
\]
Then
\[
Q_t(b,a) = -c_{\mathrm{meas}} + \sum_{o \in O} f_o(b),
\]
so it suffices to bound $|f_o(b)-f_o(b')|$ for a fixed outcome $o$.

Set
\[
p := \mathrm{ObsProb}(o \mid b,a), \qquad
p' := \mathrm{ObsProb}(o \mid b',a),
\]
and
\[
x := V_{t+1}(\tau(b,a,o)), \qquad
x' := V_{t+1}(\tau(b',a,o)).
\]
Then
\[
f_o(b)-f_o(b') = px - p'x'.
\]
Using the identity
\[
px-p'x' = (p-p')x + p'(x-x'),
\]
we obtain
\[
|f_o(b)-f_o(b')|
\le
|p-p'|\,|x| + p'|x-x'|.
\]
Since $|x| \le \|V_{t+1}\|_\infty$ and $0 \le p' \le 1$, Assumption~\ref{ass:belief-regularity} and the Lipschitz property of $V_{t+1}$ yield
\begin{align*}
|f_o(b)-f_o(b')|
&\le
|\mathrm{ObsProb}(o \mid b,a)-\mathrm{ObsProb}(o \mid b',a)|\,\|V_{t+1}\|_\infty \\
&\quad
+
|V_{t+1}(\tau(b,a,o))-V_{t+1}(\tau(b',a,o))| \\
&\le
C_P^b \|V_{t+1}\|_\infty \|b-b'\|_\infty
+
L_{t+1}\|\tau(b,a,o)-\tau(b',a,o)\|_\infty \\
&\le
\left(
C_P^b \|V_{t+1}\|_\infty
+
L_{t+1}C_\tau^b
\right)
\|b-b'\|_\infty.
\end{align*}
Summing over all $o \in O$ proves \eqref{eq:continuation-belief-lipschitz}--\eqref{eq:continuation-belief-lipschitz-constant-coarse}.
\end{proof}

Lemma~\ref{lem:continuation-belief-lipschitz} is useful because it isolates the two distinct sources of sensitivity in the continuation branch. The first contribution,
\[
C_P^b \|V_{t+1}\|_\infty,
\]
comes from the fact that changing the belief also changes the outcome probabilities themselves. The second contribution,
\[
L_{t+1}C_\tau^b,
\]
comes from the movement of the posterior state inside the future value function. Thus the continuation branch is sensitive not only because the weights in the expectation move, but also because the evaluation points of the next-stage value function move.

\begin{proposition}[Recursive belief Lipschitz bound for the exact value function]
\label{prop:value-belief-lipschitz}
Assume Assumption~\ref{ass:belief-regularity}. Suppose that $V_{t+1}$ is bounded and $L_{t+1}$-Lipschitz on $\Delta_M$ with respect to $\|\cdot\|_\infty$. Then $V_t$ is also Lipschitz on $\Delta_M$ with respect to $\|\cdot\|_\infty$, and one may take
\begin{equation}
\label{eq:value-belief-lipschitz-recursive}
L_t
\le
\max\left\{
1,\;
\max_{a \in A_{\mathrm{meas}}} \widetilde L_t(a)
\right\},
\end{equation}
where $\widetilde L_t(a)$ is given by \eqref{eq:continuation-belief-lipschitz-constant}. In particular, using \eqref{eq:continuation-belief-lipschitz-constant-coarse},
\begin{equation}
\label{eq:value-belief-lipschitz-recursive-coarse}
L_t
\le
\max\left\{
1,\;
|O|
\left(
C_P^b \|V_{t+1}\|_\infty
+
L_{t+1}C_\tau^b
\right)
\right\}.
\end{equation}
\end{proposition}

\begin{proof}
Define
\[
f_0(b):=\mathrm{StopVal}(b),
\qquad
f_a(b):=Q_t(b,a), \quad a \in A_{\mathrm{meas}}.
\]
Then
\[
V_t(b)
=
\max\left\{
f_0(b),\;
\max_{a \in A_{\mathrm{meas}}}f_a(b)
\right\}.
\]

First, $f_0=\mathrm{StopVal}$ is $1$-Lipschitz with respect to $\|\cdot\|_\infty$, since
\[
|\max_i b(i)-\max_i b'(i)|
\le
\|b-b'\|_\infty.
\]
Second, by Lemma~\ref{lem:continuation-belief-lipschitz}, each $f_a$ is $\widetilde L_t(a)$-Lipschitz.

We now use the elementary fact that the pointwise maximum of finitely many Lipschitz functions is again Lipschitz, with constant bounded by the maximum of the individual constants:
\begin{equation}
    \left|\max_j g_j(x) - \max_j g_j(y)\right| \leq \max_j |g_j(x) - g_j(y)| \leq \max_j L_j ||x - y||_\infty
\end{equation}
Applying this to the finite family
\[
\{f_0\}\cup\{f_a : a \in A_{\mathrm{meas}}\}
\]
yields \eqref{eq:value-belief-lipschitz-recursive}. The coarser estimate \eqref{eq:value-belief-lipschitz-recursive-coarse} follows immediately from \eqref{eq:continuation-belief-lipschitz-constant-coarse}.
\end{proof}

\begin{corollary}[Recursive sufficient condition for belief Lipschitz regularity]
\label{cor:recursive-belief-lipschitz}
Under Assumption~\ref{ass:belief-regularity}, the exact finite-library value functions $(V_t)_{t=0}^H$ are recursively Lipschitz on $\Delta_M$ with respect to $\|\cdot\|_\infty$. More precisely,
\[
V_H = \mathrm{StopVal}
\]
is $1$-Lipschitz, and if $V_{t+1}$ is Lipschitz, then $V_t$ is Lipschitz with constant satisfying \eqref{eq:value-belief-lipschitz-recursive}.
\end{corollary}

\begin{proof}
The terminal claim follows immediately from
\[
V_H(b)=\mathrm{StopVal}(b)=\max_i b(i),
\]
which is $1$-Lipschitz in $\|\cdot\|_\infty$. The recursive step is exactly Proposition~\ref{prop:value-belief-lipschitz}.
\end{proof}

Proposition~\ref{prop:value-belief-lipschitz} and Corollary~\ref{cor:recursive-belief-lipschitz} show that the dynamic programming preserves regularity backward in time. The stopping branch is geometrically simple, and the continuation branch inherits regularity from the next-stage value function provided the observation probabilities and posterior map are themselves regular in belief. The recursive formula
\[
L_t
\le
\max\left\{
1,\;
\max_{a \in A_{\mathrm{meas}}}\widetilde L_t(a)
\right\}
\]
should therefore be read as a stability statement for the Bellman recursion: as long as belief perturbations do not produce wild changes in the posterior update, the value function remains controlled throughout the horizon.

\subsection{Action-Space Discretization Error}

We now turn to the second approximation mechanism: discretization of the measurement action space. Up to this point, measurement actions have been denoted abstractly by $a \in A_{\mathrm{meas}}$. That level of abstraction was sufficient for the POMDP formulation, for the projected finite-library planner, and for the belief-space discretization analysis carried out above. It is also the natural notation in settings where the measurement set is taken to be primitive and finite from the outset. In the present subsection, however, the main object of analysis is precisely the approximation of a richer measurement family by a finite library. For this reason, we now make the parameterization of the measurement action explicit and write $\theta \in \Theta$, with a finite sublibrary $\Theta^h \subset \Theta$. This is not a change of model, but a refinement of notation needed to speak meaningfully about action-space distance, covering radius, and Lipschitz continuity in the action variable.

In this subsection we isolate only the effect of action discretization. Belief projection is not yet imposed, and both the exact and approximate Bellman recursions are evaluated on the full belief simplex $\Delta_M$. Thus the sole approximation mechanism under consideration here is the restriction of the continuous measurement family $\Theta$ to the finite library $\Theta^h$.

Let $(\Theta,d_\Theta)$ be a metric space of measurement parameters. For a bounded function $W:\Delta_M \to \mathbb{R}$, define the continuation functional
\begin{equation}
\label{eq:continuation-functional}
G_t(b,\theta;W)
:=
-c_{\mathrm{meas}}
+
\sum_{o \in O}
\mathrm{ObsProb}(o \mid b,\theta)\,
W(\tau(b,\theta,o)).
\end{equation}
The exact continuous-action Bellman recursion is
\begin{equation}
\label{eq:continuous-action-terminal}
V_H^\ast(b):=\mathrm{StopVal}(b),
\end{equation}
and, for $t=0,\dots,H-1$,
\begin{equation}
\label{eq:continuous-action-recursion}
V_t^\ast(b)
:=
\max\left\{
\mathrm{StopVal}(b),\;
\sup_{\theta \in \Theta} G_t(b,\theta;V_{t+1}^\ast)
\right\}.
\end{equation}

Let $\Theta^h \subset \Theta$ be a finite measurement library. The corresponding action-discretized value function is
\begin{equation}
\label{eq:finite-library-terminal}
V_H^h(b):=\mathrm{StopVal}(b),
\end{equation}
and, for $t=0,\dots,H-1$,
\begin{equation}
\label{eq:finite-library-recursion}
V_t^h(b)
:=
\max\left\{
\mathrm{StopVal}(b),\;
\max_{\bar\theta \in \Theta^h} G_t(b,\bar\theta;V_{t+1}^h)
\right\}.
\end{equation}

It is convenient to introduce the one-step Bellman operators
\begin{equation}
\label{eq:continuous-action-operator}
(T_t^\Theta W)(b)
:=
\max\left\{
\mathrm{StopVal}(b),\;
\sup_{\theta \in \Theta} G_t(b,\theta;W)
\right\},
\end{equation}
and
\begin{equation}
\label{eq:finite-library-operator}
(T_t^{\Theta^h} W)(b)
:=
\max\left\{
\mathrm{StopVal}(b),\;
\max_{\bar\theta \in \Theta^h} G_t(b,\bar\theta;W)
\right\}.
\end{equation}
The action covering radius of the library $\Theta^h$ is defined by
\begin{equation}
\label{eq:action-covering-radius}
\delta_A
:=
\sup_{\theta \in \Theta}
\min_{\bar\theta \in \Theta^h}
d_\Theta(\theta,\bar\theta).
\end{equation}
Thus $\delta_A$ measures how well the finite library $\Theta^h$ covers the continuous measurement family $\Theta$: the smaller $\delta_A$ is, the more closely every admissible measurement can be approximated by a nearby library element.

\begin{lemma}[One-step action discretization error]
\label{lem:one-step-action-discretization}
Let $W:\Delta_M \to \mathbb{R}$ be bounded. Assume that there exists a constant $K_t(W)\ge 0$ such that
\begin{equation}
\label{eq:continuation-action-lipschitz}
|G_t(b,\theta;W)-G_t(b,\theta';W)|
\le
K_t(W)\,d_\Theta(\theta,\theta'),
\qquad
\forall b \in \Delta_M,\ \forall \theta,\theta' \in \Theta.
\end{equation}
Then, for every $b \in \Delta_M$,
\begin{equation}
\label{eq:one-step-action-error}
0
\le
(T_t^\Theta W)(b)-(T_t^{\Theta^h}W)(b)
\le
K_t(W)\,\delta_A.
\end{equation}
\end{lemma}

\begin{proof}
Because $\Theta^h \subset \Theta$, we immediately have
\[
\max_{\bar\theta \in \Theta^h} G_t(b,\bar\theta;W)
\le
\sup_{\theta \in \Theta} G_t(b,\theta;W),
\]
which implies
\[
(T_t^{\Theta^h}W)(b) \le (T_t^\Theta W)(b).
\]
This proves the left inequality in \eqref{eq:one-step-action-error}.

To prove the right inequality, fix $b \in \Delta_M$ and let $\varepsilon>0$. By the definition of the supremum, there exists $\theta_\varepsilon \in \Theta$ such that
\[
\sup_{\theta \in \Theta} G_t(b,\theta;W)
\le
G_t(b,\theta_\varepsilon;W)+\varepsilon.
\]
By the definition of $\delta_A$, there exists $\bar\theta_\varepsilon \in \Theta^h$ such that
\[
d_\Theta(\theta_\varepsilon,\bar\theta_\varepsilon)\le \delta_A.
\]
Using the Lipschitz assumption \eqref{eq:continuation-action-lipschitz}, we obtain
\[
G_t(b,\theta_\varepsilon;W)
\le
G_t(b,\bar\theta_\varepsilon;W)+K_t(W)\delta_A.
\]
Therefore,
\[
\sup_{\theta \in \Theta} G_t(b,\theta;W)
\le
\max_{\bar\theta \in \Theta^h} G_t(b,\bar\theta;W)
+
K_t(W)\delta_A
+
\varepsilon.
\]
Since both Bellman operators share the same stopping term, it follows that
\[
(T_t^\Theta W)(b)
\le
(T_t^{\Theta^h}W)(b)
+
K_t(W)\delta_A
+
\varepsilon.
\]
Letting $\varepsilon \downarrow 0$ proves the right inequality in \eqref{eq:one-step-action-error}.
\end{proof}

Lemma~\ref{lem:one-step-action-discretization} is the action-space analogue of the belief-projection estimate proved earlier. The quantity $\delta_A$ is purely geometric: it measures how closely the continuous measurement family can be covered by the finite library. The constant $K_t(W)$ is analytic: it measures how sharply the continuation functional changes when the measurement parameter moves. The one-step loss is therefore controlled by the product of a geometric discretization scale and an action-sensitivity constant.

We now extend this one-step estimate to the full finite horizon.

\begin{theorem}[Finite-horizon action discretization error bound]
\label{thm:action-discretization}
Assume that, for each $t=0,\dots,H-1$, there exists a constant $K_t\ge 0$ such that
\begin{equation}
\label{eq:action-discretization-hypothesis}
|G_t(b,\theta;V_{t+1}^\ast)-G_t(b,\theta';V_{t+1}^\ast)|
\le
K_t\,d_\Theta(\theta,\theta'),
\qquad
\forall b \in \Delta_M,\ \forall \theta,\theta' \in \Theta.
\end{equation}
Then, for every $t=0,\dots,H$,
\begin{equation}
\label{eq:action-discretization-bound}
\|V_t^\ast - V_t^h\|_\infty
\le
\delta_A\,\sum_{s=t}^{H-1} K_s.
\end{equation}
\end{theorem}

\begin{proof}
Define
\[
e_t^A := \|V_t^\ast - V_t^h\|_\infty.
\]
Since
\[
V_H^\ast = V_H^h = \mathrm{StopVal},
\]
we have
\[
e_H^A = 0.
\]

For $t<H$, use the operator notation
\[
V_t^\ast = T_t^\Theta V_{t+1}^\ast,
\qquad
V_t^h = T_t^{\Theta^h} V_{t+1}^h.
\]
Hence
\begin{align*}
e_t^A
&=
\|T_t^\Theta V_{t+1}^\ast - T_t^{\Theta^h}V_{t+1}^h\|_\infty \\
&\le
\|T_t^\Theta V_{t+1}^\ast - T_t^{\Theta^h}V_{t+1}^\ast\|_\infty
+
\|T_t^{\Theta^h}V_{t+1}^\ast - T_t^{\Theta^h}V_{t+1}^h\|_\infty.
\end{align*}
By Lemma~\ref{lem:one-step-action-discretization} and the hypothesis \eqref{eq:action-discretization-hypothesis}, the first term is bounded by
\[
K_t\,\delta_A.
\]

For the second term, fix $b \in \Delta_M$ and $\bar\theta \in \Theta^h$. Then
\begin{align*}
|G_t(b,\bar\theta;W)-G_t(b,\bar\theta;W')|
&=
\left|
\sum_{o \in O}
\mathrm{ObsProb}(o \mid b,\bar\theta)
\bigl(
W(\tau(b,\bar\theta,o))-W'(\tau(b,\bar\theta,o))
\bigr)
\right| \\
&\le
\sum_{o \in O}
\mathrm{ObsProb}(o \mid b,\bar\theta)\,\|W-W'\|_\infty \\
&=
\|W-W'\|_\infty.
\end{align*}
Taking the maximum over $\bar\theta \in \Theta^h$ and then the outer maximum with the common stopping term shows that
\[
\|T_t^{\Theta^h}W - T_t^{\Theta^h}W'\|_\infty
\le
\|W-W'\|_\infty.
\]
Thus
\[
\|T_t^{\Theta^h}V_{t+1}^\ast - T_t^{\Theta^h}V_{t+1}^h\|_\infty
\le
e_{t+1}^A.
\]
Combining the two estimates gives
\[
e_t^A \le K_t\,\delta_A + e_{t+1}^A.
\]
Since $e_H^A=0$, backward iteration yields
\[
e_t^A \le \delta_A\,\sum_{s=t}^{H-1} K_s,
\]
which proves \eqref{eq:action-discretization-bound}.
\end{proof}

\begin{corollary}[Uniform action Lipschitz case]
\label{cor:uniform-action-discretization}
If, in addition, $K_s \le K$ for all $s=t,\dots,H-1$, then
\begin{equation}
\label{eq:uniform-action-discretization-bound}
\|V_t^\ast - V_t^h\|_\infty
\le
(H-t)\,K\,\delta_A.
\end{equation}
\end{corollary}

\begin{proof}
Apply Theorem~\ref{thm:action-discretization} and bound the sum by
\[
\sum_{s=t}^{H-1} K_s \le (H-t)K.
\]
\end{proof}

Theorem~\ref{thm:action-discretization} shows that the finite-horizon error caused by restricting the measurement family to a finite library accumulates one stage at a time, with a per-stage loss no larger than $K_t\delta_A$. This has the same general structure as the belief-projection estimate obtained earlier: a discretization radius is multiplied by a regularity constant, and the resulting local losses are propagated through the horizon. The crucial difference is that the geometric quantity is now the action-space covering radius $\delta_A$, while the analytic quantity is the action-sensitivity constant $K_t$ of the continuation functional.

\subsection{Regularity in the Action Variable and the Constant $K_t$}

Theorem~\ref{thm:action-discretization} shows that the finite-horizon action-discretization error is controlled by the constants $K_t$. The remaining task is therefore to verify that the continuation functional
\[
G_t(b,\theta;V_{t+1}^\ast)
\]
is Lipschitz in the action variable and to obtain an explicit upper bound for an admissible choice of $K_t$. Since the stopping branch does not depend on the measurement parameter, the entire action-side sensitivity is concentrated in the continuation branch. Accordingly, the problem reduces to understanding how observation probabilities and posterior beliefs vary when the measurement parameter $\theta$ changes.

For each hypothesis $i \in \{1,\dots,M\}$, outcome $o \in O$, and measurement parameter $\theta \in \Theta$, define the likelihood
\begin{equation}
\label{eq:action-likelihood}
\ell_i(\theta,o)
:=
\mathrm{Tr}\!\left(E_o(\theta)\rho_i\right).
\end{equation}
Then the observation probability may be written as
\begin{equation}
\label{eq:action-obsprob-likelihood}
\mathrm{ObsProb}(o \mid b,\theta)
=
\sum_{i=1}^M b(i)\,\ell_i(\theta,o),
\end{equation}
and the posterior update becomes
\begin{equation}
\label{eq:action-posterior-likelihood}
\tau(b,\theta,o)(i)
=
\frac{b(i)\,\ell_i(\theta,o)}
{\sum_{j=1}^M b(j)\,\ell_j(\theta,o)}.
\end{equation}

We now impose regularity assumptions directly on the parameterized measurement family.

\begin{assumption}[Action-regularity assumptions]
\label{ass:action-regularity}
Assume the following.
\begin{enumerate}
    \item[(A1)] There exists a constant $L_\ell \ge 0$ such that, for all $i \in \{1,\dots,M\}$, all $o \in O$, and all $\theta,\theta' \in \Theta$,
    \begin{equation}
    \label{eq:likelihood-action-lipschitz}
    |\ell_i(\theta,o)-\ell_i(\theta',o)|
    \le
    L_\ell\, d_\Theta(\theta,\theta').
    \end{equation}
    \item[(A2)] There exists $\eta>0$ such that, throughout the region of interest,
    \begin{equation}
    \label{eq:action-nondegeneracy}
    \mathrm{ObsProb}(o \mid b,\theta)\ge \eta
    \end{equation}
    whenever the posterior $\tau(b,\theta,o)$ is considered.
\end{enumerate}
\end{assumption}

\begin{remark}
Assumption~\ref{ass:action-regularity}(A2) is a nondegeneracy condition. Without a lower bound on the observation probability, the posterior update may become arbitrarily sensitive because it is a rational map whose denominator can become very small. In practice, one typically either works on a region where such a lower bound is available, or treats extremely low-probability outcomes separately. Thus the role of $\eta$ is not merely technical convenience, but stability of Bayesian normalization itself.
\end{remark}

We first control the observation probability in the action variable. This is the easy part of the analysis, because the observation probability is an average of the likelihoods with respect to the belief weights.

\begin{lemma}[Lipschitz continuity of the observation probability in the action variable]
\label{lem:obsprob-action-lipschitz}
Under Assumption~\ref{ass:action-regularity}(A1), for every $b \in \Delta_M$ and every $o \in O$,
\begin{equation}
\label{eq:obsprob-action-lipschitz}
|\mathrm{ObsProb}(o \mid b,\theta)-\mathrm{ObsProb}(o \mid b,\theta')|
\le
L_\ell\, d_\Theta(\theta,\theta'),
\qquad
\forall \theta,\theta' \in \Theta.
\end{equation}
\end{lemma}

\begin{proof}
Using \eqref{eq:action-obsprob-likelihood},
\begin{align*}
|\mathrm{ObsProb}(o \mid b,\theta)-\mathrm{ObsProb}(o \mid b,\theta')|
&=
\left|
\sum_{i=1}^M
b(i)\bigl(\ell_i(\theta,o)-\ell_i(\theta',o)\bigr)
\right| \\
&\le
\sum_{i=1}^M
b(i)\,|\ell_i(\theta,o)-\ell_i(\theta',o)| \\
&\le
\sum_{i=1}^M
b(i)\,L_\ell\, d_\Theta(\theta,\theta') \\
&=
L_\ell\, d_\Theta(\theta,\theta'),
\end{align*}
since $\sum_i b(i)=1$.
\end{proof}

The posterior map is more delicate, because unlike the observation probability it involves normalization by the denominator $\mathrm{ObsProb}(o \mid b,\theta)$.

\begin{lemma}[Lipschitz continuity of the posterior map in the action variable]
\label{lem:posterior-action-lipschitz}
Under Assumption~\ref{ass:action-regularity}, for every fixed $b \in \Delta_M$ and $o \in O$, the posterior map $\theta \mapsto \tau(b,\theta,o)$ is Lipschitz with respect to $\|\cdot\|_\infty$. More precisely, for each coordinate $i$,
\begin{equation}
\label{eq:posterior-action-coordinate-lipschitz}
|\tau(b,\theta,o)(i)-\tau(b,\theta',o)(i)|
\le
\left(
\frac{L_\ell}{\eta}
+
\frac{L_\ell}{\eta^2}
\right)
d_\Theta(\theta,\theta'),
\end{equation}
and hence
\begin{equation}
\label{eq:posterior-action-lipschitz}
\|\tau(b,\theta,o)-\tau(b,\theta',o)\|_\infty
\le
\left(
\frac{L_\ell}{\eta}
+
\frac{L_\ell}{\eta^2}
\right)
d_\Theta(\theta,\theta').
\end{equation}
\end{lemma}

\begin{proof}
Fix $b \in \Delta_M$ and $o \in O$. For each coordinate $i$, define
\[
n_i(\theta):=b(i)\,\ell_i(\theta,o),
\qquad
D(\theta):=\sum_{j=1}^M b(j)\,\ell_j(\theta,o)=\mathrm{ObsProb}(o \mid b,\theta).
\]
Then
\[
\tau(b,\theta,o)(i)=\frac{n_i(\theta)}{D(\theta)}.
\]
Therefore,
\begin{align*}
|\tau(b,\theta,o)(i)-\tau(b,\theta',o)(i)|
&=
\left|
\frac{n_i(\theta)}{D(\theta)}-\frac{n_i(\theta')}{D(\theta')}
\right| = \left|\frac{n_i(\theta) - n_i(\theta')}{D(\theta)} + n_i(\theta') \left(\frac{1}{D(\theta)} - \frac{1}{D(\theta')}\right)\right|\\
&\le
\frac{|n_i(\theta)-n_i(\theta')|}{D(\theta)}
+
|n_i(\theta')|
\left|
\frac{1}{D(\theta)}-\frac{1}{D(\theta')}
\right|.
\end{align*}

By Assumption~\ref{ass:action-regularity}(A2), both denominators are at least $\eta$. Moreover,
\[
|n_i(\theta)-n_i(\theta')|
=
b(i)\,|\ell_i(\theta,o)-\ell_i(\theta',o)|
\le
L_\ell\, d_\Theta(\theta,\theta'),
\]
so the first term is bounded by
\[
\frac{L_\ell}{\eta}\, d_\Theta(\theta,\theta').
\]

For the second term, note first that $|n_i(\theta')|\le 1$. Also,
\[
\left|
\frac{1}{D(\theta)}-\frac{1}{D(\theta')}
\right|
=
\frac{|D(\theta)-D(\theta')|}{D(\theta)D(\theta')}
\le
\frac{|D(\theta)-D(\theta')|}{\eta^2}.
\]
By Lemma~\ref{lem:obsprob-action-lipschitz},
\[
|D(\theta)-D(\theta')|
=
|\mathrm{ObsProb}(o \mid b,\theta)-\mathrm{ObsProb}(o \mid b,\theta')|
\le
L_\ell\, d_\Theta(\theta,\theta').
\]
Hence the second term is bounded by
\[
\frac{L_\ell}{\eta^2}\, d_\Theta(\theta,\theta').
\]

Combining the two estimates gives \eqref{eq:posterior-action-coordinate-lipschitz}. Taking the maximum over coordinates yields \eqref{eq:posterior-action-lipschitz}.
\end{proof}

Lemma~\ref{lem:posterior-action-lipschitz} is the technical heart of the action-regularity analysis. The term $L_\ell/\eta$ reflects direct variation of the posterior numerator under a stable denominator, while the term $L_\ell/\eta^2$ reflects movement of the denominator itself. Thus the posterior becomes difficult to control precisely when the measurement family changes rapidly with $\theta$ or when the observation probability is allowed to become too small.

We now combine the previous two lemmas with the Lipschitz regularity of the next-stage value function.

\begin{proposition}[An explicit action Lipschitz bound for the continuation functional]
\label{prop:explicit-Kt-bound}
Assume Assumption~\ref{ass:action-regularity}. Suppose further that $V_{t+1}^\ast$ is bounded and $L_{t+1}$-Lipschitz on $\Delta_M$ with respect to $\|\cdot\|_\infty$. Then the continuation functional is uniformly Lipschitz in the action variable: for every $b \in \Delta_M$ and all $\theta,\theta' \in \Theta$,
\begin{equation}
\label{eq:continuation-action-explicit-lipschitz}
|G_t(b,\theta;V_{t+1}^\ast)-G_t(b,\theta';V_{t+1}^\ast)|
\le
K_t\, d_\Theta(\theta,\theta'),
\end{equation}
where one admissible choice is
\begin{equation}
\label{eq:explicit-Kt}
K_t
=
|O|
\left[
L_\ell \|V_{t+1}^\ast\|_\infty
+
L_{t+1}
\left(
\frac{L_\ell}{\eta}
+
\frac{L_\ell}{\eta^2}
\right)
\right].
\end{equation}
\end{proposition}

\begin{proof}
Fix $b \in \Delta_M$, $\theta,\theta' \in \Theta$, and an outcome $o \in O$. Define
\[
p_\theta := \mathrm{ObsProb}(o \mid b,\theta),
\qquad
p_{\theta'} := \mathrm{ObsProb}(o \mid b,\theta'),
\]
and
\[
x_\theta := V_{t+1}^\ast(\tau(b,\theta,o)),
\qquad
x_{\theta'} := V_{t+1}^\ast(\tau(b,\theta',o)).
\]
Then
\[
p_\theta x_\theta - p_{\theta'}x_{\theta'}
=
(p_\theta-p_{\theta'})x_\theta + p_{\theta'}(x_\theta-x_{\theta'}).
\]
Taking absolute values gives
\begin{align*}
&\bigl|
\mathrm{ObsProb}(o \mid b,\theta)\,V_{t+1}^\ast(\tau(b,\theta,o))
-
\mathrm{ObsProb}(o \mid b,\theta')\,V_{t+1}^\ast(\tau(b,\theta',o))
\bigr| \\
&\qquad\le
|\mathrm{ObsProb}(o \mid b,\theta)-\mathrm{ObsProb}(o \mid b,\theta')|\,\|V_{t+1}^\ast\|_\infty \\
&\qquad\quad
+
\mathrm{ObsProb}(o \mid b,\theta')\,
|V_{t+1}^\ast(\tau(b,\theta,o))-V_{t+1}^\ast(\tau(b,\theta',o))|.
\end{align*}
By Lemma~\ref{lem:obsprob-action-lipschitz}, the first term is bounded by
\[
L_\ell \|V_{t+1}^\ast\|_\infty\, d_\Theta(\theta,\theta').
\]
By the $L_{t+1}$-Lipschitz continuity of $V_{t+1}^\ast$, Lemma~\ref{lem:posterior-action-lipschitz}, and the fact that $0 \le \mathrm{ObsProb}(o \mid b,\theta') \le 1$, the second term is bounded by
\[
L_{t+1}
\left(
\frac{L_\ell}{\eta}
+
\frac{L_\ell}{\eta^2}
\right)
d_\Theta(\theta,\theta').
\]
Therefore, for each outcome $o$,
\begin{align*}
&\bigl|
\mathrm{ObsProb}(o \mid b,\theta)\,V_{t+1}^\ast(\tau(b,\theta,o))
-
\mathrm{ObsProb}(o \mid b,\theta')\,V_{t+1}^\ast(\tau(b,\theta',o))
\bigr| \\
&\qquad\le
\left[
L_\ell \|V_{t+1}^\ast\|_\infty
+
L_{t+1}
\left(
\frac{L_\ell}{\eta}
+
\frac{L_\ell}{\eta^2}
\right)
\right]
d_\Theta(\theta,\theta').
\end{align*}
Summing over all $o \in O$ proves \eqref{eq:continuation-action-explicit-lipschitz} with the choice \eqref{eq:explicit-Kt}.
\end{proof}

Proposition~\ref{prop:explicit-Kt-bound} makes the structure of the constant $K_t$ transparent. The factor $L_\ell$ measures how rapidly the parameterized POVM family changes with the action variable. The factor $\|V_{t+1}^\ast\|_\infty$ measures the scale of future rewards, while $L_{t+1}$ measures how sensitive the next-stage value function is to posterior perturbations. Finally, the parameter $\eta$ encodes the stability of Bayesian normalization: when $\eta$ is small, the posterior update becomes more fragile, and the admissible Lipschitz constant $K_t$ grows accordingly. Thus action discretization becomes difficult precisely when the measurement family varies rapidly, the future value landscape is steep, or the relevant observation probabilities are close to zero.

\subsection{Total Approximation Error Bound}

We now combine the two approximation techniques analyzed above.
The previous subsections treated belief-space discretization and action-space
discretization separately. The natural final question is therefore the following:
how far can the fully projected finite-grid finite-library planner deviate from
the exact continuous-action Bellman value function?

To answer this, we introduce one additional layer of notation. Recall that
$V_t^\ast$ denotes the exact value function for the continuous measurement family
$\Theta$, whereas $V_t^h$ denotes the value function obtained after restricting
the action space to the finite library $\Theta^h \subset \Theta$, without yet
projecting the belief state.

We now define the \emph{projected finite-library value function} $\hat V_t^h$
on the grid $B \subset \Delta_M$ by
\begin{equation}
\label{eq:projected-finite-library-terminal-h}
\hat V_H^h(b) := \mathrm{StopVal}(b), \qquad b \in B,
\end{equation}
and, for $t=0,\dots,H-1$,
\begin{equation}
\label{eq:projected-finite-library-recursion-h}
\hat V_t^h(b)
:=
\max\left\{
\mathrm{StopVal}(b),\;
\max_{\bar\theta \in \Theta^h}
\left(
-c_{\mathrm{meas}}
+
\sum_{o \in O}
\mathrm{ObsProb}(o \mid b,\bar\theta)\,
\hat V_{t+1}^h\!\bigl(\mathrm{Proj}_B(\tau(b,\bar\theta,o))\bigr)
\right)
\right\},
\,\, b \in B.
\end{equation}
Thus $\hat V_t^h$ is the value function produced by the fully projected planner:
the action space has been discretized to $\Theta^h$, and posterior beliefs are
projected back to the finite grid $B$.

To control the belief-projection component for this finite-library model, we
invoke the finite-library belief regularity established in
Proposition~\ref{prop:value-belief-lipschitz} and
Corollary~\ref{cor:recursive-belief-lipschitz}. More precisely, that analysis
was carried out for an arbitrary fixed finite measurement library, and therefore
applies in particular to the library $\Theta^h$. Accordingly, we let $L_t^h$
denote Lipschitz constants for the finite-library value functions $V_t^h$, so that the following theorem holds.

\begin{theorem}[Total approximation error bound]
\label{thm:total-approximation-error}
Assume the hypotheses of Theorem~\ref{thm:action-discretization}. Let  $\{L_t^h\}_{t=0}^H$
denote belief-Lipschitz constants for the action-discretized value functions
$(V_t^h)_{t=0}^H$, furnished by
Proposition~\ref{prop:value-belief-lipschitz} and
Corollary~\ref{cor:recursive-belief-lipschitz} when those results are applied
to the finite library $\Theta^h$.
Then, for every $t=0,\dots,H$,
\begin{equation}
\label{eq:total-approximation-error-grid}
\|V_t^\ast-\hat V_t^h\|_{\infty,B}
:=
\sup_{b \in B}|V_t^\ast(b)-\hat V_t^h(b)|
\le
\delta_A \sum_{s=t}^{H-1} K_s
+
\delta_B \sum_{s=t+1}^{H} L_s^h.
\end{equation}
\end{theorem}

\begin{proof}
Fix $t \in \{0,\dots,H\}$. For every $b \in B$, apply the triangle inequality:
\[
|V_t^\ast(b)-\hat V_t^h(b)|
\le
|V_t^\ast(b)-V_t^h(b)|
+
|V_t^h(b)-\hat V_t^h(b)|.
\]
Taking the supremum over $b \in B$ gives
\begin{equation}
\label{eq:total-error-triangle}
\|V_t^\ast-\hat V_t^h\|_{\infty,B}
\le
\|V_t^\ast-V_t^h\|_{\infty,B}
+
\|V_t^h-\hat V_t^h\|_{\infty,B}.
\end{equation}

We now bound the two terms separately.

For the first term, since $B \subset \Delta_M$,
\[
\|V_t^\ast-V_t^h\|_{\infty,B}
\le
\|V_t^\ast-V_t^h\|_\infty.
\]
Therefore, by Theorem~\ref{thm:action-discretization},
\begin{equation}
\label{eq:total-error-action-part}
\|V_t^\ast-V_t^h\|_{\infty,B}
\le
\delta_A \sum_{s=t}^{H-1} K_s.
\end{equation}

For the second term, applying Theorem~\ref{thm:belief-discretization} to the pair
$(V_t^h,\hat V_t^h)$ yields
\begin{equation}
\label{eq:total-error-belief-part}
\|V_t^h-\hat V_t^h\|_{\infty,B}
\le
\delta_B \sum_{s=t+1}^{H} L_s^h.
\end{equation}

Combining \eqref{eq:total-error-triangle},
\eqref{eq:total-error-action-part}, and
\eqref{eq:total-error-belief-part} proves
\eqref{eq:total-approximation-error-grid}.
\end{proof}

\begin{corollary}[Total approximation error bound for arbitrary belief state]
\label{cor:total-approximation-error-arbitrary-belief}
For an arbitrary belief $b \in \Delta_M$, define the projected estimator
\begin{equation}
\label{eq:projected-estimator-h}
\tilde V_t^h(b)
:=
\hat V_t^h(\mathrm{Proj}_B(b)).
\end{equation}
Under the assumptions of Theorem~\ref{thm:total-approximation-error},
\begin{equation}
\label{eq:total-approximation-error-arbitrary-belief}
|V_t^\ast(b)-\tilde V_t^h(b)|
\le
\delta_A \sum_{s=t}^{H-1} K_s
+
\delta_B \sum_{s=t}^{H} L_s^h.
\end{equation}
\end{corollary}

\begin{proof}
The argument is the same as in
Corollary~\ref{cor:belief-discretization-final-bound}, except that one first
decomposes the error into an action-discretization part and a belief-projection
part, and then applies the finite-library belief-projection estimate to the pair
$(V_t^h,\hat V_t^h)$. Indeed,
\begin{align*}
|V_t^\ast(b)-\tilde V_t^h(b)|
&\le
|V_t^\ast(b)-V_t^h(b)|
+
|V_t^h(b)-\hat V_t^h(\mathrm{Proj}_B(b))| \\
&\le
\|V_t^\ast-V_t^h\|_\infty
+
|V_t^h(b)-V_t^h(\mathrm{Proj}_B(b))| \\
&\quad
+
|V_t^h(\mathrm{Proj}_B(b))-\hat V_t^h(\mathrm{Proj}_B(b))|.
\end{align*}
The first term is bounded by Theorem~\ref{thm:action-discretization}, while the
second and third terms are controlled exactly as in
Corollary~\ref{cor:belief-discretization-final-bound}, with the finite-library
Lipschitz constants $L_t^h$. This yields
\[
|V_t^\ast(b)-\tilde V_t^h(b)|
\le
\delta_A \sum_{s=t}^{H-1} K_s
+
\delta_B \sum_{s=t}^{H} L_s^h.
\]
\end{proof}

Theorem~\ref{thm:total-approximation-error} and
Corollary~\ref{cor:total-approximation-error-arbitrary-belief} give the final
error decomposition for the projected approximation architecture. The total loss
is the sum of two contributions. The first is the \emph{action-side} loss,
controlled by the covering radius $\delta_A$ of the finite measurement library
and by the action-sensitivity constants $K_t$. The second is the
\emph{belief-side} loss, controlled by the projection radius $\delta_B$ of the
belief grid and by the belief-Lipschitz constants $L_t^h$ of the finite-library
value functions.

Thus the projected planner is accurate whenever both discretizations are fine and
the corresponding Bellman objects are sufficiently regular. In this sense, the
total approximation error retains the same structural form as the separate bounds
derived earlier: a geometric discretization scale is multiplied by an analytic
sensitivity constant, and the resulting local losses accumulate through the
finite horizon.

\section{Complexity and the Curse of Dimensionality}

We now complement the approximation-theoretic analysis with a computational one.
The previous sections established that the projected dynamic-programming architecture
is mathematically well posed and that its approximation error can be controlled in terms
of belief-space and action-space discretization radii. A different but equally natural
question is therefore the following: what are the offline and online computational costs of the projected planner, and how do they scale with the horizon and discretization sizes? Related complexity issues for approximate and point-based POMDP planning have been studied extensively in the literature \cite{Shani2013SurveyPBVI, Smith2005PointBased}.

The answer has two layers. First, at the level of the projected finite-grid
finite-library algorithm, one may derive an explicit symbolic cost formula by
inspecting the Bellman recursion step by step. Second, once this symbolic
formula is specialized to a simple projection implementation, one obtains an
asymptotic complexity law that makes the computational bottleneck transparent.
In particular, Theorem~\ref{thm:total-approximation-error} and
Corollary~\ref{cor:total-approximation-error-arbitrary-belief} show that
planner accuracy is jointly governed by the action-library covering radius
$\delta_A$ and the belief-grid projection radius $\delta_B$. When the belief
grid is refined in order to reduce approximation error, the resulting planner
therefore exhibits a clear curse of dimensionality through the geometry of the
belief simplex.

\subsection{Offline Planning Complexity of the Projected Finite-Grid Planner}

We begin with the offline phase, namely the backward-induction computation of the
projected value and policy tables on the finite belief grid $B$. Throughout this subsection,
we work with Algorithm~1 and with the projected Bellman recursion introduced in
Section~2. The outcome alphabet is the common finite set $O$, and the measurement
library is the finite set $A_{\mathrm{meas}}$.

To express the complexity transparently, it is convenient to isolate the elementary
operations that occur inside one Bellman update.

\paragraph{Atomic Costs}
Fix a stage $t \in \{0,\dots,H-1\}$ and a grid belief $b \in B$. We introduce the following
primitive costs:
\[
C_{\mathrm{stop}}(M)
\]
for evaluating
\[
\mathrm{StopVal}(b)=\max_i b(i),
\]
\[
C_{\mathrm{obs}}(a,o;M)
\]
for evaluating the observation probability
\[
\mathrm{ObsProb}(o\mid b,a),
\]
\[
C_{\tau}(a,o;M)
\]
for constructing the posterior
\[
\tau(b,a,o),
\]
\[
C_{\mathrm{proj}}(|B|,M)
\]
for projecting a posterior belief to the grid,
\[
C_{\mathrm{lookup}}
\]
for retrieving the next-stage value at the projected belief,
\[
C_{\mathrm{agg}}
\]
for incorporating one outcome contribution into the continuation expectation,
\[
C_{\mathrm{init}}
\]
for per-action initialization of the continuation branch, and
\[
C_{\mathrm{actmax}}(|A_{\mathrm{meas}}|)
\]
for selecting the best measurement action after all continuation branches have been evaluated.

For each measurement action $a \in A_{\mathrm{meas}}$, define the continuation-evaluation cost
\[
C_{\mathrm{cont}}(a;B,M)
:=
C_{\mathrm{init}}
+
\sum_{o\in O}
\Bigl[
C_{\mathrm{obs}}(a,o;M)
+
C_{\tau}(a,o;M)
+
C_{\mathrm{proj}}(|B|,M)
+
C_{\mathrm{lookup}}
+
C_{\mathrm{agg}}
\Bigr].
\]

\begin{proposition}[Whole-planner symbolic cost formula]
Let $C_{\mathrm{term\mbox{-}node}}(M)$ denote the terminal-stage cost per grid node.
Then the total offline planning cost of the projected finite-grid finite-library algorithm is
\[
C_{\mathrm{off}}
=
|B|\, C_{\mathrm{term\mbox{-}node}}(M)
+
\sum_{t=0}^{H-1} \sum_{b\in B} C_{\mathrm{Bell}}(t,b).
\]
If the same grid $B$, action library $A_{\mathrm{meas}}$, and outcome alphabet $O$ are used
at every nonterminal stage, then
\[
C_{\mathrm{off}}
=
|B|\, C_{\mathrm{term\mbox{-}node}}(M)
+
H |B|\, C_{\mathrm{Bell\mbox{-}node}},
\]
where
\[
C_{\mathrm{Bell\mbox{-}node}}
=
C_{\mathrm{stop}}(M)
+
|A_{\mathrm{meas}}|\, C_{\mathrm{init}}
+
\sum_{a\in A_{\mathrm{meas}}}\sum_{o\in O}
\Bigl[
C_{\mathrm{obs}}(a,o;M)
+
C_{\tau}(a,o;M)
\Bigr]
\]
\[
\qquad\qquad
+
|A_{\mathrm{meas}}||O|
\Bigl[
C_{\mathrm{proj}}(|B|,M)
+
C_{\mathrm{lookup}}
+
C_{\mathrm{agg}}
\Bigr]
+
C_{\mathrm{actmax}}(|A_{\mathrm{meas}}|).
\]
\end{proposition}

\begin{proof}
At the terminal stage $t=H$, Algorithm~1 performs no measurement loop. For each grid
belief $b\in B$, it computes the terminal stopping value and stores the corresponding
declaration action. Thus the terminal pass contributes
\[
|B|\, C_{\mathrm{term\mbox{-}node}}(M).
\]

Now fix a nonterminal stage $t<H$ and a grid node $b\in B$. At this node, the algorithm
first evaluates the stopping branch once, at cost $C_{\mathrm{stop}}(M)$. It then loops over
all measurement actions $a\in A_{\mathrm{meas}}$. For each such action, it initializes the
continuation value, then for every outcome $o\in O$ it computes the observation probability,
constructs the posterior, projects that posterior back to the grid, retrieves the next-stage
value, and aggregates the resulting contribution into the continuation expectation. By
definition, the total cost of this continuation evaluation is $C_{\mathrm{cont}}(a;B,M)$.
After all actions have been processed, the algorithm chooses the best measurement branch,
at cost $C_{\mathrm{actmax}}(|A_{\mathrm{meas}}|)$.

Therefore the Bellman-update cost at one nonterminal node $(t,b)$ is exactly
\[
C_{\mathrm{Bell}}(t,b)
=
C_{\mathrm{stop}}(M)
+
\sum_{a\in A_{\mathrm{meas}}} C_{\mathrm{cont}}(a;B,M)
+
C_{\mathrm{actmax}}(|A_{\mathrm{meas}}|).
\]
Summing this over all nonterminal stages and all grid nodes proves the first formula.

If the same finite sets $B$, $A_{\mathrm{meas}}$, and $O$ are used at every nonterminal stage,
then the symbolic Bellman-update cost is stage-independent and node-independent, so it may
be written as a common quantity $C_{\mathrm{Bell\mbox{-}node}}$. Expanding the definition of
$C_{\mathrm{cont}}$ and grouping identical terms yields the displayed expression for
$C_{\mathrm{Bell\mbox{-}node}}$.
\end{proof}

The preceding proposition is purely structural. It does not yet impose any particular
implementation model for projection, posterior computation, or lookup. Its purpose is
instead to make explicit which parts of the planner contribute to the cost and how those
parts are nested inside the Bellman recursion. In particular, the proposition shows that
the projection step enters once per outcome and per measurement action, which will become
crucial after asymptotic specialization.

\subsubsection{Linear-Scan Projection and Asymptotic Specialization}

We now specialize to the simplest implementation of the projection operator.

\begin{proposition}[Linear-scan nearest-neighbor projection cost]
Suppose $\mathrm{Proj}_B$ is implemented by nearest-neighbor search over the entire grid $B$
using linear scan and deterministic tie-breaking rules, and suppose distances are computed from the full $M$-coordinate belief
representation. Then
\[
C_{\mathrm{proj}}(|B|,M)=\Theta(|B|M).
\]
\end{proposition}

\begin{proof}
To project a posterior belief $x$ onto the grid $B$, a linear-scan nearest-neighbor
implementation compares $x$ with every candidate grid point $y\in B$. Thus all $|B|$
candidates must be inspected.

For each candidate $y$, a standard norm-based distance computation, such as
$\ell_1$, $\ell_2$, or $\ell_\infty$, requires examining all $M$ coordinates of the two
belief vectors. Hence one candidate comparison costs $\Theta(M)$. Since the algorithm
scans all $|B|$ candidates and tie-breaking adds only lower-order bookkeeping, the total
projection cost is
\[
C_{\mathrm{proj}}(|B|,M)=\Theta(|B|M).
\]
\end{proof}

Moreover, let us impose the simplest asymptotic model for the remaining primitives:
\[
C_{\mathrm{stop}}(M)=\Theta(M), \qquad
C_{\mathrm{obs}}(a,o;M)=\Theta(M), \qquad
C_{\tau}(a,o;M)=\Theta(M),
\]
\[
C_{\mathrm{lookup}}=\Theta(1), \qquad
C_{\mathrm{agg}}=\Theta(1), \qquad
C_{\mathrm{init}}=\Theta(1), \qquad
C_{\mathrm{actmax}}(|A_{\mathrm{meas}}|)=\Theta(|A_{\mathrm{meas}}|).
\]
The interpretation is straightforward. Evaluating the stopping value scans an
$M$-dimensional belief vector. Observation probabilities are weighted sums over the
$M$ hypotheses. Posterior normalization is likewise linear in $M$. Lookup, one-term
aggregation, and branch initialization are constant-size tasks, while action maximization
scans the finite action library once.

\begin{theorem}[Offline planning complexity under linear-scan projection]
Under the asymptotic assumptions above and under linear-scan nearest-neighbor projection,
\[
C_{\mathrm{Bell\mbox{-}node}}
=
\Theta(|A_{\mathrm{meas}}||O||B|M),
\]
and therefore
\[
C_{\mathrm{off}}
=
\Theta(H|A_{\mathrm{meas}}||O|M|B|^2).
\]
\end{theorem}

\begin{proof}
Substitute the asymptotic primitive costs into the Bellman-node formula of the previous
proposition. The terms involving observation probabilities and posterior computations
contribute
\[
\Theta(|A_{\mathrm{meas}}||O|M).
\]
By the previous proposition,
\[
C_{\mathrm{proj}}(|B|,M)=\Theta(|B|M),
\]
so the total contribution of the projection terms is
\[
|A_{\mathrm{meas}}||O|\, \Theta(|B|M)
=
\Theta(|A_{\mathrm{meas}}||O||B|M).
\]
Since $|B|\ge 1$, this projection term dominates the lower-order contributions
$\Theta(M)$, $\Theta(|A_{\mathrm{meas}}|)$, and $\Theta(|A_{\mathrm{meas}}||O|M)$.
Hence
\[
C_{\mathrm{Bell\mbox{-}node}}
=
\Theta(|A_{\mathrm{meas}}||O||B|M).
\]

Now return to the whole-planner formula:
\[
C_{\mathrm{off}}
=
|B|\, C_{\mathrm{term\mbox{-}node}}(M)
+
H|B|\, C_{\mathrm{Bell\mbox{-}node}}.
\]
The terminal contribution is $\Theta(|B|M)$, whereas the nonterminal contribution is
\[
H|B| \cdot \Theta(|A_{\mathrm{meas}}||O||B|M)
=
\Theta(H|A_{\mathrm{meas}}||O|M|B|^2).
\]
The latter dominates the former, and the claimed asymptotic law follows.
\end{proof}

The theorem has a simple interpretation. Under the present implementation assumptions,
the projected planner is linear in the horizon $H$, linear in the number of measurement
actions $|A_{\mathrm{meas}}|$, linear in the number of outcomes $|O|$, linear in the
hypothesis count $M$, and quadratic in the belief-grid cardinality $|B|$. The quadratic
dependence on $|B|$ is not an artifact of notation but a direct consequence of combining
a grid-based dynamic program with linear-scan belief projection.

\subsubsection{Curse of Dimensionality and Accuracy--Complexity Interpretation}

The asymptotic law obtained above is informative but not yet fully geometric, because
the grid cardinality $|B|$ is not itself a primitive feature of the problem. Rather, $|B|$
arises from discretizing the belief simplex
\[
\Delta_M = \bigl\{ b\in \mathbb{R}_{\ge 0}^M : \sum_{i=1}^M b(i)=1 \bigr\},
\]
whose geometric dimension is $M-1$.

To make this dependence explicit, we impose the standard regular-grid scaling law.

\begin{assumption}[Regular-grid belief scaling]
Suppose the finite belief grid $B \subset \Delta_M$ is constructed as a regular
discretization with covering radius $\delta_B$, and suppose its cardinality obeys
\[
|B| \asymp \delta_B^{-(M-1)}.
\]
\end{assumption}

\begin{proposition}[Curse-of-dimensionality interpretation]
Under the regular-grid scaling assumption and under the hypotheses of the previous theorem,
\[
C_{\mathrm{off}}
\asymp
H|A_{\mathrm{meas}}||O|M\, \delta_B^{-2(M-1)}.
\]
\end{proposition}

\begin{proof}
By the previous theorem,
\[
C_{\mathrm{off}}=\Theta(H|A_{\mathrm{meas}}||O|M|B|^2).
\]
Substituting
\[
|B| \asymp \delta_B^{-(M-1)}
\]
gives
\[
|B|^2 \asymp \bigl(\delta_B^{-(M-1)}\bigr)^2 = \delta_B^{-2(M-1)}.
\]
Hence
\[
C_{\mathrm{off}}
\asymp
H|A_{\mathrm{meas}}||O|M\, \delta_B^{-2(M-1)}.
\]
\end{proof}

This proposition makes the computational meaning of belief-space refinement explicit.
If $M$ is fixed, then the planner scales polynomially in $\delta_B^{-1}$. However, the
exponent of that polynomial is $2(M-1)$, which itself grows linearly with the number of
hypotheses. In this sense, the projected architecture exhibits a clear curse of dimensionality:
as the hypothesis dimension increases, even modest improvements in the belief-grid resolution
become rapidly expensive.

The preceding complexity law becomes still more informative when combined with the
belief-space approximation error bound from Section~5.1. There, the projected value-function
error is controlled linearly in the projection radius $\delta_B$, up to stage-dependent
Lipschitz constants. Thus one may directly convert a target approximation accuracy into a
required computational budget.

\begin{corollary}[Accuracy--complexity law]
Assume the hypotheses of Theorem~\ref{thm:belief-discretization}, and suppose one targets approximation accuracy
$\varepsilon>0$ by choosing the belief-grid resolution so that $\delta_B = O(\varepsilon)$
at the level of scaling. Then
\[
C_{\mathrm{off}}
\asymp
H|A_{\mathrm{meas}}||O|M\, \varepsilon^{-2(M-1)}.
\]
\end{corollary}

\begin{proof}
By the previous proposition,
\[
C_{\mathrm{off}}
\asymp
H|A_{\mathrm{meas}}||O|M\, \delta_B^{-2(M-1)}.
\]
On the other hand, Theorem~\ref{thm:belief-discretization} shows that, for fixed horizon and relevant Lipschitz
constants, the belief-space discretization error scales linearly in $\delta_B$. Therefore,
to achieve target approximation accuracy $\varepsilon$, one may choose the grid resolution
so that $\delta_B = O(\varepsilon)$ at the level of scaling. Substituting
\[
\delta_B \asymp \varepsilon
\]
into the preceding complexity law gives
\[
C_{\mathrm{off}}
\asymp
H|A_{\mathrm{meas}}||O|M\, \varepsilon^{-2(M-1)}.
\]
\end{proof}

The corollary is perhaps the clearest computational summary of the projected approximation
architecture. The approximation error decreases only linearly with the belief-grid radius,
but the corresponding offline planning cost grows polynomially in $\varepsilon^{-1}$ with
an exponent that depends linearly on the belief-space dimension. Thus the projected method
is mathematically controlled, but improving its accuracy by brute-force grid refinement
rapidly becomes expensive in high dimension. Interestingly, this interpretation is also consistent with the broader complexity perspective in point-based POMDP planning, where both the curse of dimensionality and the curse of history play a central role in complexity analysis \cite{Smith2005PointBased}.

\subsection{Online Execution Complexity and Stopping Time}

We finally turn to online execution. The offline phase computes a policy table over the grid.
Once this table is available, the agent no longer sweeps over all beliefs or all branches of
the decision tree. Instead, it follows one realized trajectory: it evaluates the current policy,
performs the prescribed action, receives an observation, updates the belief, and repeats until
a declaration action is chosen. The natural complexity descriptor is therefore not the horizon
alone, but the stopping time of the induced trajectory.

Fix a deterministic policy
\[
\pi=(\mu_0,\dots,\mu_H)
\]
and an initial prior belief $b_0$. Let a sample path be denoted by
\[
\omega=(h,o_1,\dots,o_H),
\]
where $h$ is the hidden hypothesis drawn once from the prior and the outcomes $o_t$ are
generated under the observation law induced by the policy. The corresponding action and
belief trajectories are defined recursively by
\[
a_t(\omega)=\mu_t(b_t(\omega)),
\qquad
b_{t+1}(\omega)=\tau(b_t(\omega),a_t(\omega),o_{t+1}).
\]
The stopping time of the policy is
\[
\tau_\pi(\omega)
:=
\min\Bigl\{
t\in\{0,\dots,H\} :
\mu_t(b_t(\omega)) \in \{\delta_1,\dots,\delta_M\}
\Bigr\}.
\]

Let
\[
C_{\mathrm{pol}}(|B|,M)
\]
denote the policy-evaluation cost per step,
\[
C_{\mathrm{obs\mbox{-}recv}}
\]
the observation-handling cost per step,
\[
C_{\mathrm{upd}}(M)
\]
the posterior-update cost per step, and
\[
C_{\mathrm{term}}^{\mathrm{on}}
\]
the terminal declaration cost. Define the pathwise online execution cost by
\[
C_{\mathrm{on}}(\pi;\omega)
:=
\sum_{t=0}^{\tau_\pi(\omega)-1}
\Bigl[
C_{\mathrm{pol}}(|B|,M)
+
C_{\mathrm{obs\mbox{-}recv}}
+
C_{\mathrm{upd}}(M)
\Bigr]
+
C_{\mathrm{term}}^{\mathrm{on}}.
\]

\begin{proposition}[Expected online execution cost]
With the notation above,
\[
\mathbb{E}_\pi[C_{\mathrm{on}}]
=
\mathbb{E}_\pi[\tau_\pi]
\Bigl[
C_{\mathrm{pol}}(|B|,M)
+
C_{\mathrm{obs\mbox{-}recv}}
+
C_{\mathrm{upd}}(M)
\Bigr]
+
O(1).
\]
\end{proposition}

\begin{proof}
By definition,
\[
C_{\mathrm{on}}(\pi;\omega)
=
\sum_{t=0}^{\tau_\pi(\omega)-1}
\Bigl[
C_{\mathrm{pol}}(|B|,M)
+
C_{\mathrm{obs\mbox{-}recv}}
+
C_{\mathrm{upd}}(M)
\Bigr]
+
C_{\mathrm{term}}^{\mathrm{on}}.
\]
Taking expectation with respect to the probability law induced by the prior and the
observation kernels under the fixed policy $\pi$, and using linearity of expectation, yields
\[
\mathbb{E}_\pi[C_{\mathrm{on}}]
=
\sum_{t\ge 0}
\mathbb{P}_\pi(t<\tau_\pi)
\Bigl[
C_{\mathrm{pol}}(|B|,M)
+
C_{\mathrm{obs\mbox{-}recv}}
+
C_{\mathrm{upd}}(M)
\Bigr]
+
\mathbb{E}_\pi[C_{\mathrm{term}}^{\mathrm{on}}].
\]
Since
\[
\sum_{t\ge 0}\mathbb{P}_\pi(t<\tau_\pi)=\mathbb{E}_\pi[\tau_\pi],
\]
and the terminal declaration overhead is constant-size, we obtain
\[
\mathbb{E}_\pi[C_{\mathrm{on}}]
=
\mathbb{E}_\pi[\tau_\pi]
\Bigl[
C_{\mathrm{pol}}(|B|,M)
+
C_{\mathrm{obs\mbox{-}recv}}
+
C_{\mathrm{upd}}(M)
\Bigr]
+
O(1).
\]
\end{proof}

The proposition shows that online execution is governed by a fundamentally different
complexity principle from offline planning. Offline computation sweeps over all stages and
all grid beliefs, whereas online computation processes only the realized pre-stopping segment
of one trajectory. Thus the natural online cost scale is the expected stopping time
$\mathbb{E}_\pi[\tau_\pi]$, multiplied by the per-step belief-processing cost.

\begin{corollary}[Linear-scan online execution law]
Suppose policy evaluation uses linear-scan projection over the grid, so that
\[
C_{\mathrm{pol}}(|B|,M)=\Theta(|B|M),
\]
and suppose posterior updating obeys
\[
C_{\mathrm{upd}}(M)=\Theta(M).
\]
Then the per-step online cost is dominated by policy evaluation, and
\[
\mathbb{E}_\pi[C_{\mathrm{on}}]
=
\mathbb{E}_\pi[\tau_\pi]\, \Theta(|B|M) + O(1).
\]
\end{corollary}

\begin{proof}
Under the stated assumptions,
\[
C_{\mathrm{pol}}(|B|,M)
+
C_{\mathrm{obs\mbox{-}recv}}
+
C_{\mathrm{upd}}(M)
=
\Theta(|B|M),
\]
because the linear-scan policy-evaluation term dominates the lower-order $\Theta(M)$
posterior-update term and the constant-size observation-handling overhead. Substituting
this into the previous proposition yields the claim.
\end{proof}

The offline and online complexity laws together clarify the computational role of the
projected architecture. The expensive part of the method is the offline construction of the
policy table, whose cost reflects both Bellman recursion over the belief grid and repeated
projection back to that grid. Online execution, by contrast, follows only one realized path
through the precomputed policy and is therefore controlled by the stopping time rather than
by the full horizon-wide branching structure.

\section{Working Example: Binary Quantum State Discrimination Case}
So far, we have established that sequential and adaptive quantum state discrimination (QSD) can be reframed as a Partially Observable Markov Decision Process (POMDP) within the context of classical decision and control theory. We have also conducted a detailed analysis of the approximation errors induced by naive algorithms--specifically projected dynamic programming using regular grids and finite action libraries--alongside the computational complexity of both offline planning and online execution. However, these discussions have focused primarily on problem redefinition and theoretical computational analysis. To address the lack of concrete application details, this section aims to demonstrate how the specific functions of the POMDP framework are constructed in practice and how the tradeoff between keep measuring and stopping arises in the realistic situations. 

\subsection{Problem Statements}
While the methodology presented thus far remains within a highly generalized framework, practical application necessitates further specification and a deliberate narrowing of the problem's scope. In this subsection, we provide a detailed explanation of how to address a concrete binary state discrimination task and how the associated measurements are parameterized. 

Specifically, let us consider the problem of discriminating between the following two states:
\begin{align*}
    & \left|\psi_1\right> = \left|0\right> \\
    & \left|\psi_2\right> = \cos \theta \left|0\right> + \sin \theta \left|1\right> \quad \text{where} \,\, \theta \in \left(0, \frac{\pi}{2}\right).
\end{align*}
Thus those two states are naturally non-orthogonal. Moreover, define the corresponding density matrices as follows:
\begin{equation*}
    \rho_1 = \left|\psi_1\right>\left<\psi_1\right| \quad \rho_2 = \left|\psi_2\right>\left<\psi_2\right|.
\end{equation*}
Since this is basically binary state discrimination case, the (prior) belief would be simplified into probability parameter $p$ (which denotes probability for state-1 from now on) and given in a explicit form of $b = (p, 1 - p)$ where $p \in \left[0, 1\right]$.

Consider the following two states, which would constitute the projective measurement operators parameterized by $\phi$ respectively.
\begin{align*}
    & \left|e_0(\phi)\right> = \cos \phi \left|0\right> + \sin \phi \left|1\right> \\
    & \left|e_1(\phi)\right> = -\sin \phi \left|0\right> + \cos \phi \left|1\right>.
\end{align*}
It is easily verified that those two states are actually orthonormal basis. Think of corresponding measurement operators
\begin{align*}
    E_0(\phi) & = \left|e_0(\phi)\right>\left<e_0(\phi)\right| = \left(\begin{matrix}\sin^2 \phi & -\cos \phi \sin \phi \\ -\cos \phi \sin \phi & \cos^2 \phi\end{matrix}\right) \\
    E_1(\phi) & = \left|e_1(\phi)\right>\left<e_1(\phi)\right| = \left(\begin{matrix}\cos^2 \phi & \cos \phi \sin \phi \\ \cos \phi \sin \phi & \sin^2 \phi\end{matrix}\right),
\end{align*}
and it follows directly that those two operators are valid POVM operators. While we formerly described a method for creating parameterized POVMs via parameterized unitaries, we will focus here on simple parameterized projective measurements to clearly demonstrate the core insights of this framework. We emphasize that this is a choice made for illustrative simplicity, and the framework remains compatible with more complex parameterization schemes. Furthermore, introducing projective measurements based on the states $\left|e_0\right>$ and $\left|e_1\right>$--and assigning the notations $E_0$ and $E_1$ accordingly--implies that we label the outcome as $o = 0$ when the system is projected onto $\left|e_0\right>$, and $o = 1$ when projected onto $\left|e_1\right>$. It is important to underscore, yet, that these labels are purely arbitrary for this specific case study only. They serve as intuitive indices and do not signify that the physical $\left|0\right>$ or $\left|1\right>$ states have been measured.

\subsection{Bayesian Belief Updating and Belief Simplex Geometry}
As defined earlier, let us denote the observation probability as follows.
\begin{equation*}
    \ell_i(\phi, o) = \operatorname{Tr}(E_o(\phi) \rho_i) \quad \text{for} \,\, i = 1, 2 \,\, \text{and} \,\, o = 0, 1
\end{equation*}

If we denote by $p'_o$ the posterior probability of state 1 after observing $o$, then it would be calculated as
\begin{equation*}
    p'_o = \frac{p \ell_1(\phi, o)}{p \ell_1(\phi, o) + (1 - p) \ell_2(\phi, o)}.
\end{equation*}
Naturally, the posterior probability of state 2 is $1 - p'_o$. If we write the observation probabilities explicitly as expectation values, then
\begin{itemize}
    \item For the case of $o = 0$,
    \begin{equation*}
        p'_0 = \frac{p\cos^2 \phi}{p\cos^2 \phi + (1 - p)\cos^2(\theta - \phi)}.
    \end{equation*}
    \item For the case of $o = 1$,
    \begin{equation*}
        p'_1 = \frac{p\sin^2 \phi}{p\sin^2 \phi + (1 - p)\sin^2(\theta - \phi)}.
    \end{equation*}
\end{itemize}

In the conventional odds form, 
\begin{equation*}
    \frac{p'_o}{1 - p'_o} = \frac{p}{1 - p} \cdot \frac{\ell_1(\phi, o)}{\ell_2(\phi, o)}.
\end{equation*}
Thus just by following the so-called \emph{posterior jump} of the crucial parameter $p$, we might reconstruct the full current belief update which could also be readily visualized as a point moving within the interval $\left[0, 1\right]$ on the $x$-axis.

\begin{figure}[t]
\centering
\begin{tikzpicture}[x=10cm,y=1cm,>=stealth]

    % baseline
    \draw[thick] (0,0) -- (1,0);

    % endpoints
    \fill (0,0) circle (1.2pt);
    \fill (1,0) circle (1.2pt);
    \node[below] at (0,0) {$0$};
    \node[below] at (1,0) {$1$};

    % prior and posteriors (arbitrary example positions)
    \fill (0.42,0) circle (1.6pt);
    \node[below=4pt] at (0.42,0) {$p$};

    \fill (0.72,0) circle (1.6pt);
    \node[above=4pt] at (0.72,0) {$p'_0$};

    \fill (0.18,0) circle (1.6pt);
    \node[above=4pt] at (0.18,0) {$p'_1$};

    % arrows
    \draw[->,thick] (0.44,0.08) to[bend left=20] (0.70,0.08);
    \node[above] at (0.575,0.25) {$o=0$};

    \draw[->,thick] (0.40,0.08) to[bend right=20] (0.20,0.08);
    \node[above] at (0.3,0.19) {$o=1$};

\end{tikzpicture}
\caption{Illustration of posterior jumps on the belief line. Starting from the prior belief $p \in [0,1]$, observing outcome $o=0$ moves the belief to $p'_0$, while observing outcome $o=1$ moves it to $p'_1$.}
\label{fig:belief-jump-binary}
\end{figure}
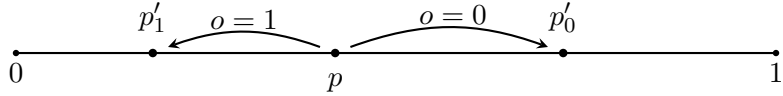

In a more mathematical way, the updating rule would be described as follows.
\begin{align*}
    \frac{\ell_1(\phi, o)}{\ell_2(\phi, o)} > 1 & \Leftrightarrow \frac{p'_o}{1 - p'_o} > \frac{p}{1 - p} \\
    & \Rightarrow p'_o > p.
\end{align*}
Speaking concretely, 
\begin{itemize}
    \item When $o = 0$, we only perform such an update when $\cos^2 \phi > \cos^2 (\theta - \phi)$.
    \item When $o = 1$, we only conduct such an update when $\sin^2 \phi > \sin^2 (\theta - \phi)$.
\end{itemize}
The boundary case $\phi = \frac{\theta}{2}$ corresponds to the situation in which the measurement outcome $o$ provides no additional information for the decision problem.

\subsection{Reward and Optimal One-Step Payoff} \label{seq:Reward and Optimal One-step Payoff}
By following the previous discussions, stop-and-declare gives us
\begin{equation*}
    R_{\text{stop}}(p, \delta_1) = p \quad R_{\text{stop}}(p, \delta_2) = 1 - p.
\end{equation*}
Thus the total $\text{StopVal}(p)$ is 
\begin{align*}
    \text{StopVal}(p) & = \max \{p, 1 - p\} \\
    & = \begin{cases}
        1 - p & \text{if } 0 \leq p < \frac{1}{2} \\
        p & \text{if } \frac{1}{2} \leq p \leq 1.
    \end{cases}
\end{align*}
The interpretation is quite intuitive. If we are less confident about the current hidden state (namely, located near the point $p = 1/2$), the immediate stopping gives us less payoff and vice versa (near the point $p = 0$ or $p = 1$).

Let us consider the case in which we perform one measurement and then declare optimally afterward. In that case, the total expected success probability would be given as
\begin{equation*}
    J_1(p, \phi) = \Pr(0 | p, \phi) \max \{p'_0, 1 - p'_0\} + \Pr(1 | p, \phi) \max \{p'_1, 1 - p'_1\}.
\end{equation*}
Since all of those components have already been discussed in detail in the above paragraph, the equation reduces to 
\begin{equation}
    J_1(p, \phi) = \max \{p\cos^2 \phi, (1 - p)\cos^2 (\theta - \phi)\} + \max \{p\sin^2 \phi, (1 - p)\sin^2 (\theta - \phi)\},
\end{equation}
and this is a concrete formulation of the prior one-step reduction formula.

Observe that
\begin{align*}
    & \max \{p\cos^2 \phi, (1 - p)\cos^2 (\theta - \phi)\} \geq p\cos^2 \phi \\
    & \max \{p\sin^2 \phi, (1 - p)\sin^2 (\theta - \phi)\} \geq p\sin^2 \phi.
\end{align*}
Adding both terms in the left-hand-side and right-hand-side respectively,
\begin{equation*}
    J_1(p, \phi) \geq p(\cos^2 \phi + \sin^2 \phi) = p.
\end{equation*}
Following the exact same logic, we also get
\begin{equation*}
    J_1(p, \phi) \geq 1 - p.
\end{equation*}
Thus, the following relation holds trivially.
\begin{equation*}
    J_1(p, \phi) \geq \max \{p, 1 - p\} = \text{StopVal}(p),
\end{equation*}
with the equality holds if $\phi = \frac{\theta}{2}$, where the measurement is non-informative. While the measurement cost $c_{\text{meas}}$ is excluded here, the inequality provides an intuitive insight: in the absence of costs, taking a measurement is always no worse than making an immediate decision. This aligns with the basic information-theoretic fact that measurements are never detrimental to the decision process; they either gain information or stay non-informative, but they never subtract from what is already known.

\subsection{Consistency Check with the Conventional Helstrom Bound}
Recall to Eq.~(76) and if we substitute prior $p = \frac{1}{2}$ (uniform prior) and perform one-step-further calculation of $\max$ operator, then we get
\begin{align*}
    J_1\left(\frac{1}{2}, \phi\right) & = \frac{1}{2} \cdot \frac{\cos^2 \phi + \cos^2 (\theta - \phi) + |\cos^2 \phi - \cos^2 (\theta - \phi)|}{2} \\ & + \frac{1}{2} \cdot \frac{\sin^2 \phi + \sin^2 (\theta - \phi) + |\sin^2 \phi - \sin^2 (\theta - \phi)|}{2} \\
    & = \frac{1}{2} + \frac{1}{4} \left(|\cos^2 \phi - \cos^2 (\theta - \phi)| + |\sin^2 \phi - \sin^2 (\theta - \phi)|\right).
\end{align*}
Observe
\begin{equation*}
    \sin^2 x - \sin^2 y = -(\cos^2 x - \cos^2 y)
\end{equation*}
thus in the perspective of absolute value, the both term is equal. Moreover, consider 
\begin{equation*}
    \cos^2 u - \cos^2 v = -\sin (u + v) \sin (u - v).
\end{equation*}
Then
\begin{align*}
    J_1\left(\frac{1}{2}, \phi\right) & = \frac{1}{2} + \frac{1}{4} \left(|\cos^2 \phi - \cos^2 (\theta - \phi)| + |\sin^2 \phi - \sin^2 (\theta - \phi)|\right) \\
    & = \frac{1}{2} + \frac{1}{2} |\cos^2 \phi - \cos^2 (\theta - \phi)| \\
    & = \frac{1}{2} + \frac{1}{2} \sin \theta |\sin (2\phi - \theta)|.
\end{align*}

Naturally, the immediate stopping value at belief $p = 1/2$ would be given as $\operatorname{StopVal}(1/2) = 1/2$. Given the previous $J_1$ denotes the expected payoff the agent could achieve when it performs one measurement and then declare optimally, the term 
\begin{equation*}
    \frac{1}{2} \sin \theta |\sin (2\phi - \theta)|
\end{equation*}
might be viewed as the possible gain from one measurement. Now we optimize measurement operator, namely through parameter $\phi$, as follows:
\begin{equation*}
    \phi^* := \operatorname{argmax}_{\phi \in \left[0, \pi\right)} |\sin(2\phi - \theta)|.
\end{equation*}
Simply, it would be given as
\begin{align*}
    2\phi^* - \theta & = \frac{\pi}{2} + k\pi \quad (k\text{ is integer.}) \\
    \phi^* & = \frac{\theta}{2} + \frac{\pi}{4} + \frac{k}{2} \pi \\
    & = \frac{\theta}{2} + \frac{\pi}{4} \quad \left(\operatorname{mod} \,\, \frac{1}{2}\pi\right).
\end{align*}
In that situation, the previous `measurement gain' term gets simplified as
\begin{equation*}
    \frac{1}{2} \sin \theta.
\end{equation*}
With the accompanying 1/2 term, $J^*_1$, which can be also interpreted as optimal success probability, would be written as
\begin{equation*}
    \frac{1}{2} + \frac{1}{2} \sin \theta.
\end{equation*}
Following the previous definitions of $\left|\psi_1\right>$ and $\left|\psi_2\right>$, 
\begin{equation*}
    |\left<\psi_1|\psi_1\right>|^2 = \cos^2 \theta
\end{equation*}
directly follows thus
\begin{equation*}
    \frac{1}{2} + \frac{1}{2} \sin \theta = \frac{1}{2} \left(1 + \sqrt{1 - |\left<\psi_1|\psi_2\right>|^2}\right)
\end{equation*}
which the right-hand-side term is actually well-known Helstrom Bound \cite{Helstrom1976QuantumDetection, BaeKwek2015} for discriminating two pure state with the uniform prior. This is a simple consistency check for the binary state case.

\subsection{Measurement Gain Function and Bellman Equations}
So far we have considered the case where measurement action causes no actual cost, which is quite unrealistic considering the expensiveness of current quantum devices \cite{TianEtAl2024}. Thus, from now on we will fix it by $c_{\text{meas}}$ and determine under which condition the agent might choose keep measuring as an optimum and how the actual Bellman equations are written in the finite horizon case.

Our claim of determining between stopping and keeping measuring is given in a following mathematical proposition.
\begin{proposition}
    In the binary state discrimination scenarios considered so far, the agent performs a measurement if and only if the following condition is met:
    \begin{equation*}
        \sup_\phi J_1(p, \phi) - \operatorname{StopVal}(p) > c_{\text{meas}}.
    \end{equation*}
\end{proposition}
\begin{proof}
    It is trivial considering the definition of $J^*_1$, $\operatorname{StopVal}$, and $c_{\text{meas}}$.
\end{proof}
We are going to define the left-hand-side as $G(p)$, which would denote the one-step information gain. This function has following three notable properties.
\begin{proposition}
    The function $G(p)$ defined above has the following three properties.
    \begin{itemize}
        \item $G(p) \geq 0$, hence it is non-negative.
        \item $G(p) = G(1 - p)$, hence it is symmetric with $p = 1/2$.
        \item $\lim_{p \rightarrow 0+} G(p) = \lim_{p \rightarrow 1-} G(p) = 0$, hence the additional gain of measurement at the very end of the one-dimensional simplex $\left[0, 1\right]$ diminishes to zero.
    \end{itemize}
\end{proposition}
\begin{proof}
    For the first property, refer to
    \begin{equation*}
        \sup_\phi J_1(p, \phi) \geq \operatorname{StopVal}(p),
    \end{equation*}
    which was established in the Section~\ref{seq:Reward and Optimal One-step Payoff}.

    For the second one, consider assigning $p$ to $\psi_1$ was purely arbitrary, bounded by our choice. Thus even if we give $p'$ to $\psi_2$ and $1 - p'$ to $\psi_1$, it does not matter at all. Afterwards, if one follows the detailed expansions $J_1$ and $\operatorname{StopVal}$ which were covered in the sections above, one might readily verify this claim.

    For the third one, consider the case where $p$ goes to 1, for instance. It operationally signifies that the agent believes the hidden state is actually state 1 with almost maximum confidence, making $\operatorname{StopVal}(p) = p \rightarrow 1$. Also note that due to the very definition of $J_1(p, \phi)$, it can never go beyond the bound 1. Then following the previous definition and viewing the asymptotic situation $p \rightarrow 1$, it is easily seen that $J_1(p, \phi) \rightarrow 1$.
\end{proof}
Also, remark that at $p = 1/2$, $G(p) = J^*_1(1/2) - 1/2 = 1/2 \sin \theta > 0$. \\

All of these properties constitute the only characteristics of $G(p)$ that can be rigorous proven or shown at this stage. Ultimately, this confirms our earlier intuition: the gain from measurement is maximized near $p = 1/2$, where belief uncertainty is at its peak. Furthermore, the symmetry between favoring State 1 and State 2 is clearly maintained, while the additional one-step payoff vanishes at the boundaries ($p \rightarrow 0, 1$) where the state becomes nearly certain. Rigorous derivation of further properties would serve as an excellent direction for future research. 

Generally, for functions exhibiting such behavior, one might expect a concave-like structure; in such cases, the decision boundary between measurement and termination would likely manifest as a single interval determined by $c_{\text{meas}}$. Yet, the gain function may take a more intricate form, in which case the measurement region would be represented as a union of multiple disjoint intervals.

Finally, let us explicitly write the Bellman recursive equations for $H = 2$ case. For simplicity, introduce much simpler notation $S(p) := \operatorname{StopVal}(p)$. Then
\begin{align*}
    & V_2(p) = S(p), \\
    & V_1(p) = \max \{S(p), \sup_\phi\left[-c_{\text{meas}} + \Pr(0 | p, \phi) S(p'_0) + \Pr(1 | p, \phi) S(p'_1)\right]\}, \\
    & V_0(p) = \max \{S(p), \sup_\phi \left[-c_{\text{meas}} + \Pr(0 | p, \phi) V_1(p'_0) + \Pr(1 | p, \phi) V_1(p'_1)\right]\}.
\end{align*}

In $V_1(p)$, as only a single measurement opportunity remains, the agent evaluates the trade-off between immediate termination with an optimal decision and performing one final measurement prior to termination. In contrast, $V_0(p)$ embodies a fundamental control-theoretic perspective aligned with the philosophy of \emph{sequential posterior movement}. Here, the agent considers how the current measurement should steer the posterior distribution to maximize the future expected value, namely $V_1$.

\section{Numerical Simulation: Trine State Discrimination Case}
Having already examined the binary quantum state discrimnation case as a working example of the general framework, we now turn to a concrete trine state discrimination example in order to illustrate how the same POMDP-based sequential QSD formulation operates in a genuinely non-binary setting. The binary case is useful for displaying, in the simplest possible geometry, the basic mechancis of belief updating, stopping, and consistency with conventional one-step discrimination. The trine case, yet, leads naturally to a richer belief space geometry: because the hidden hypothesis now ranges over three candidate states, the posterior evolves on a two-dimensional simplex rather than on a one-dimensional interval. As a result, the value of additional measurements, the routing of posterior beliefs, and the resulting stop/continue structure becomes substantially more informative.

The purpose of the present section is not to present the entire numerical pipeline in exhaustive detail, but rather to isolate the aspects of the trine example that are most directly relevant to the theory developed earlier. More specifically, we first describe the trine ensemble and the associated belief-simplex geometry, then examine the one-step measurement gain structure, and next discuss representative posterior-routing behavior that clarifies the sequential interpretation of the model. Finally, we present a compact finite-horizon Bellman analysis showing how nontrivial continuation and stopping regions arise in this setting, together with a brief robustness check under discretization refinement. In this way, the trine case serves as a concise validation example showing that the static-hidden-state POMDP formulation and the projected planner remain meaningful in a genuinely richer multi-state discrimination problem.

\subsection{Trine Ensemble and Belief-Simplex Geometry}

We consider a trine ensemble consisting of three symmetrically arranged candidate states, residing on the xy-plane of the Bloch sphere,
indexed by $i \in \{1,2,3\}$, with phase labels
\[
\phi_1 = 0, \qquad \phi_2 = \frac{2\pi}{3}, \qquad \phi_3 = \frac{4\pi}{3}.
\]
In contrast with the binary case, the hidden hypothesis now takes three possible values,
so the posterior belief is no longer described by a one-dimensional interval.
Rather, the belief state is
\[
b = (b_1,b_2,b_3) \in \Delta_3
:= \Bigl\{ b \in \mathbb{R}_{\ge 0}^3 : b_1+b_2+b_3 = 1 \Bigr\},
\]
that is, a point on the two-dimensional probability simplex.
Accordingly, sequential belief updating in the trine case is naturally interpreted as motion
on a triangular region rather than on a line segment.

For numerical computation, we represent this simplex by a finite lattice of barycentric belief
states. More precisely, for a grid resolution parameter $N \in \mathbb{N}$, we use the set
\[
B_N
=
\left\{
\left(\frac{i}{N},\frac{j}{N},\frac{k}{N}\right)
:
i,j,k \in \mathbb{Z}_{\ge 0},\;
i+j+k=N
\right\}.
\]
Thus each grid point corresponds to a discrete barycentric triple, and the entire simplex is
approximated by a triangular mesh.
To visualize the belief geometry in the plane, we use the standard embedding
\[
x = b_2 + \frac{1}{2} b_3,
\qquad
y = \frac{\sqrt{3}}{2} b_3,
\]
which maps $\Delta_3$ bijectively onto an equilateral triangle.
Under this identification, the vertices correspond to certainty beliefs
$(1,0,0)$, $(0,1,0)$, and $(0,0,1)$, while the center
$\bigl(\frac13,\frac13,\frac13\bigr)$ represents the maximally symmetric prior.

Measurement actions are parameterized by a single angle
\[
\alpha \in \left[0,\frac{2\pi}{3}\right),
\]
which is discretized in computation by a finite measurement library.
For each outcome $o \in \{1,2,3\}$ and hypothesis $i \in \{1,2,3\}$,
the corresponding trine likelihood is written as
\[
\ell_i(\alpha,o)
=
\frac{1}{3}
\Bigl(
1+\cos\bigl(\phi_i-\alpha-\phi_o\bigr)
\Bigr).
\]
The posterior belief after observing outcome $o$ is then obtained by the corresponding Bayesian
normalization,
\[
\tau(b,\alpha,o)(i)
=
\frac{b_i \, \ell_i(\alpha,o)}
{\sum_{j=1}^3 b_j \, \ell_j(\alpha,o)}.
\]
All of these parameterization preserves the cyclic symmetry of the trine ensemble while yielding a
compact measurement family suitable for projected dynamic programming.

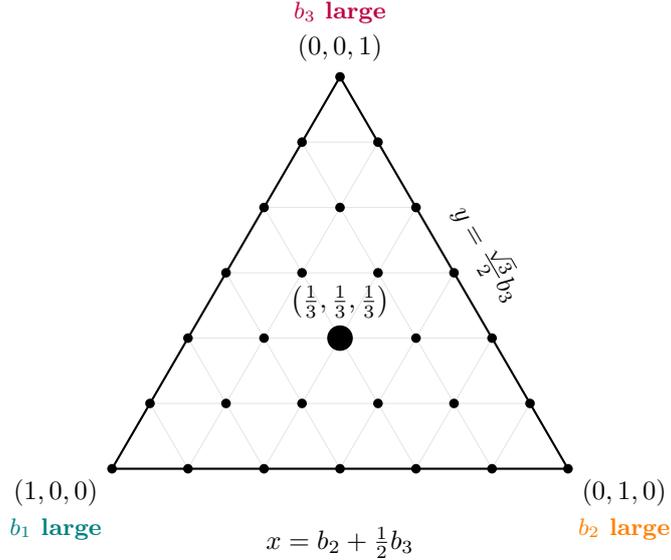
\begin{figure}[t]
\centering
\begin{tikzpicture}[scale=6, line join=round, line cap=round]

  % 1. 꼭짓점 좌표 정의
  \coordinate (v1) at (0,0);
  \coordinate (v2) at (1,0);
  \coordinate (v3) at (0.5,0.866);

  % 2. 내부 격자선 (매우 연한 회색으로 배경화)
  \foreach \i in {1,...,5} {
    \draw[gray!20, thin] ({0.5*\i/6},{0.866*\i/6}) -- ({1-0.5*\i/6},{0.866*\i/6});
    \draw[gray!20, thin] ({\i/6},0) -- ({0.5+0.5*\i/6},{0.866*(1-\i/6)});
    \draw[gray!20, thin] ({1-\i/6},0) -- ({0.5-0.5*\i/6},{0.866*(1-\i/6)});
  }

  % 3. 외곽 삼각형 (격자선보다 위로 오도록 나중에 그림)
  \draw[thick] (v1) -- (v2) -- (v3) -- cycle;

  % 4. 격자점 (크기를 줄여 깔끔하게 유지)
  \foreach \i in {0,...,6} {
    \foreach \j in {0,...,6} {
      \pgfmathtruncatemacro{\k}{6-\i-\j}
      \ifnum\k>-1
        \pgfmathsetmacro{\px}{\j/6 + 0.5*\k/6}
        \pgfmathsetmacro{\py}{0.866*\k/6}
        \fill[black] (\px,\py) circle (0.3pt);
      \fi
    }
  }

  % 5. 중앙 점 및 라벨 (겹치지 않게 처리)
  \fill[black] (0.5, 0.2887) circle (0.8pt);
  \node[above, yshift=4pt, font=\small] at (0.5, 0.2887) {$\left(\frac13, \frac13, \frac13\right)$};

  % 6. 꼭짓점 좌표 및 'Large' 라벨 (밖으로 배치)
  % b1 (왼쪽 아래)
  \node[below left, align=center] at (v1) {
    \small $(1,0,0)$ \\ \footnotesize \textcolor{teal}{\textbf{$b_1$ large}}
  };
  % b2 (오른쪽 아래)
  \node[below right, align=center] at (v2) {
    \small $(0,1,0)$ \\ \footnotesize \textcolor{orange}{\textbf{$b_2$ large}}
  };
  % b3 (위쪽)
  \node[above, align=center, yshift=2pt] at (v3) {
    \footnotesize \textcolor{purple}{\textbf{$b_3$ large}} \\ \small $(0,0,1)$
  };

  % 7. x축 좌표 설명 (아래쪽 중앙)
  \node[below, yshift=-18pt] at (0.5,0) {\small $x = b_2 + \frac12 b_3$};

  % 8. y축 좌표 설명 (빗변과 평행하게 바깥쪽 배치)
  % (0.5, 0.866)에서 (1, 0)으로 이어지는 빗변의 바깥쪽 0.08 만큼 떨어진 지점
  \node[rotate=-60, anchor=south, yshift=4pt] at (0.75, 0.433) {\small $y = \frac{\sqrt{3}}{2} b_3$};

\end{tikzpicture}
\caption{Belief-simplex geometry. Each component of the 3-tuple denotes the belief in a specific state, with the center representing maximal uncertainty (indifference) and the vertices representing certainty. The original 3D probability simplex is mapped onto a 2D equilateral triangle via the transformation shown in the figure.}
\label{fig:trine_simplex_grid_refined}
\end{figure}

\subsection{One-Step Measurement Gain on the Trine Simplex}
We now examine how the one-step optimal measurement value varies over the trine belief simplex introduced in Section~8.1. For each belief state \(b \in \Delta_3\), we optimize over the measurement orientation \(\alpha \in [0,2\pi/3)\) and define the corresponding one-step value by
\[
J_1^*(b)
:=
\sup_{\alpha \in [0,2\pi/3)}
\sum_{o\in\{0,1,2\}}
\max_{i\in\{1,2,3\}}
\bigl[
b_i\,\ell_i(\alpha,o)
\bigr].
\]
We compare this quantity with the immediate stopping value
\[
S(b):=\max\{b_1,b_2,b_3\},
\]
and define the one-step measurement gain by
\[
G(b):=J_1^*(b)-S(b).
\]
Thus $J^*_1(b)$ is the best success probability achievable after one additional measurement, while $G(b) = J^*_1(b) - S(b)$ measures how much that extra measurement improves the value upon immediate stopping.

\begin{figure}[t]
    \centering
    \includegraphics[width=0.82\textwidth]{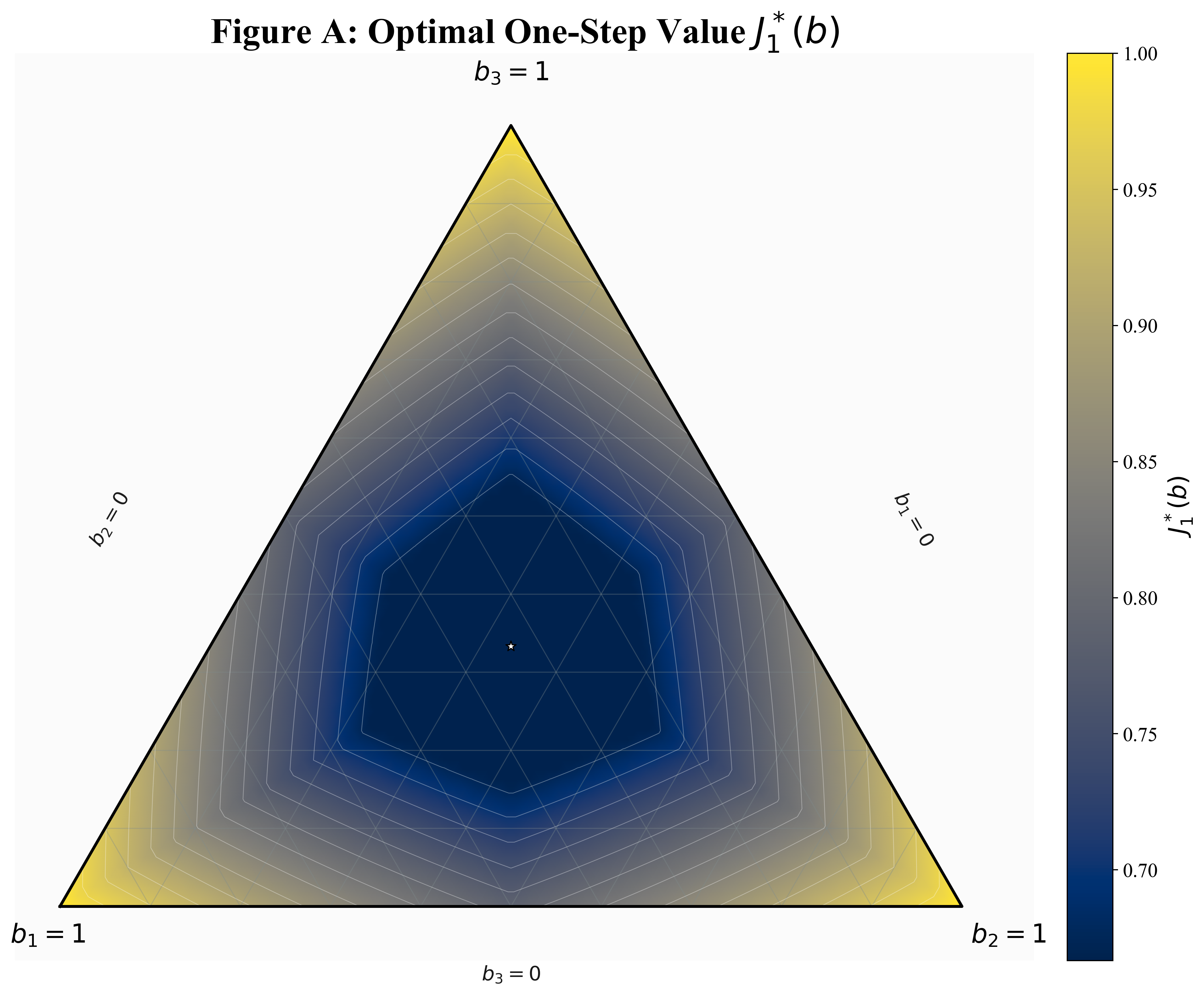}
    \caption{
    Optimal one-step value \(J_1^*(b)\) on the trine belief simplex. The value is highest near the vertices, where one hypothesis already dominates the others, and lowest near the symmetric center \(b=(1/3,1/3,1/3)\), where the three hypotheses are uniformly balanced.
    }
    \label{fig:trine-j1star}
\end{figure}

Figure~\ref{fig:trine-j1star} shows the spatial structure of the optimal one-step value \(J_1^*(b)\) on the simplex. The map is largest near the vertices and smallest near the symmetric center \(b=(1/3,1/3,1/3)\). This pattern is natural. Near a vertex, one hypothesis already dominates the belief state, so the overall one-step success probability is high. By contrast, near the center, the discrimination task is intrinsically harder because the three candidate states are nearly balanced, and the resulting one-step success probability is correspondingly smaller. The map also displays the expected threefold rotational symmetry inherited from the trine ensemble.

However, the direct one-step value \(J_1^*(b)\) does not by itself indicate where measurement is most useful. For that purpose, we compare it with the stopping benchmark \(S(b)\).

\begin{figure}[t]
    \centering
    \includegraphics[width=0.82\textwidth]{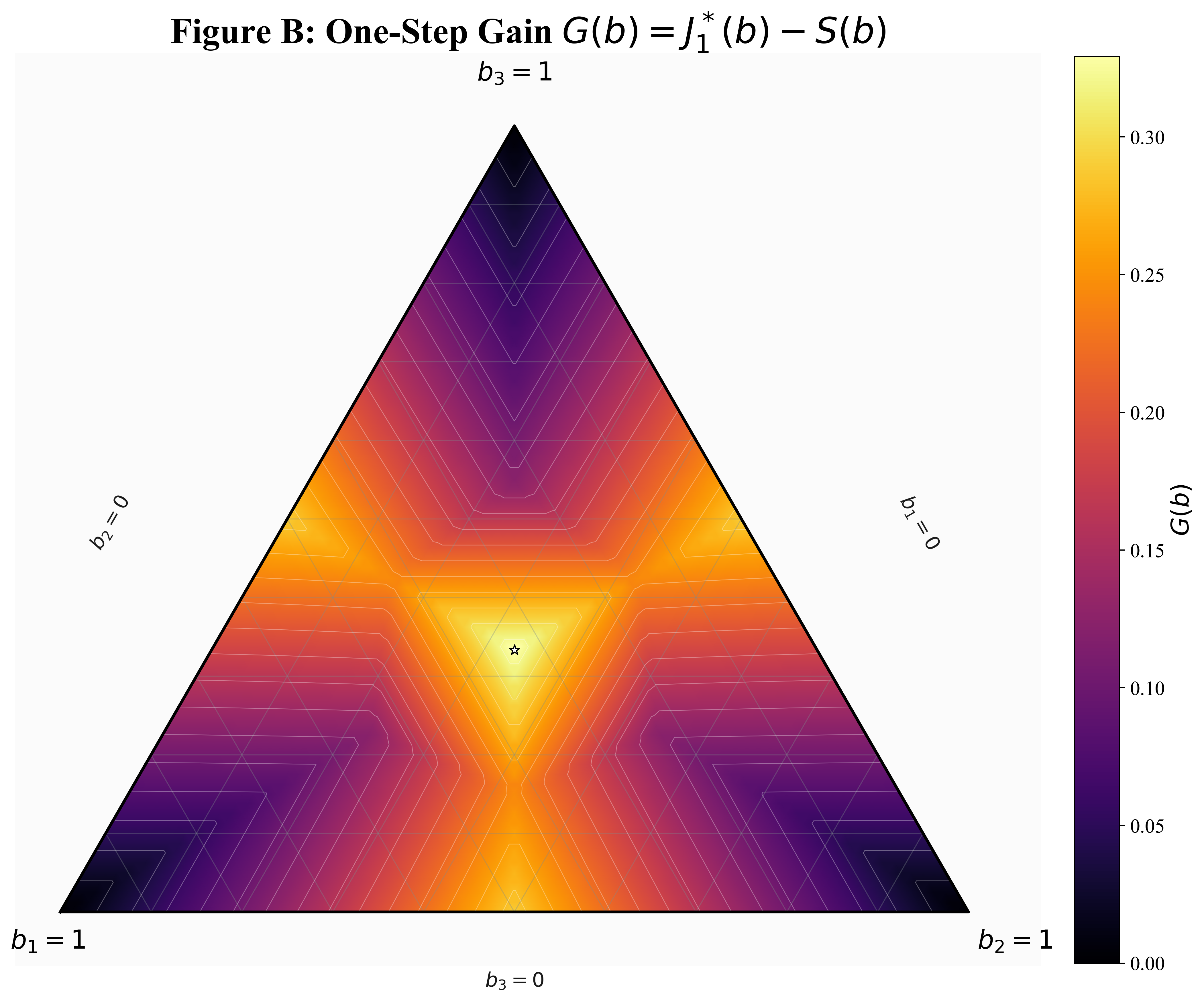}
    \caption{
    One-step gain \(G(b)=J_1^*(b)-S(b)\) on the trine simplex. The gain is largest in the central high-uncertainty region and decreases toward the vertices, where immediate stopping is already nearly optimal.
    }
    \label{fig:trine-gain}
\end{figure}

Figure~\ref{fig:trine-gain} displays the one-step gain map \(G(b)=J_1^*(b)-S(b)\), which measures the advantage of making one additional optimized measurement rather than stopping immediately. In sharp contrast with Figure~\ref{fig:trine-j1star}, the gain is largest in the central high-uncertainty region and decreases toward the vertices, where it approaches zero. This means that one further measurement is most valuable precisely when the current posterior is most ambiguous. Conversely, when the belief is already concentrated near a single hypothesis, further measurement offers little improvement over simply declaring that hypothesis immediately.

This contrast between \(J_1^*(b)\) and \(G(b)\) is conceptually important. The quantity \(J_1^*(b)\) is an absolute success-probability landscape, whereas \(G(b)\) is a marginal value-of-information landscape. On the trine simplex, these two maps exhibit qualitatively different spatial behavior: regions with high absolute one-step success need not be the regions in which measurement is especially worthwhile. In particular, the center of the simplex is not where one-step discrimination is easiest, but where one additional measurement is most worth taking. For sequential control, the gain map is therefore the more operationally relevant object, since it directly indicates where continuation can outperform immediate stopping.

\begin{figure}[t]
    \centering
    \includegraphics[width=0.82\textwidth]{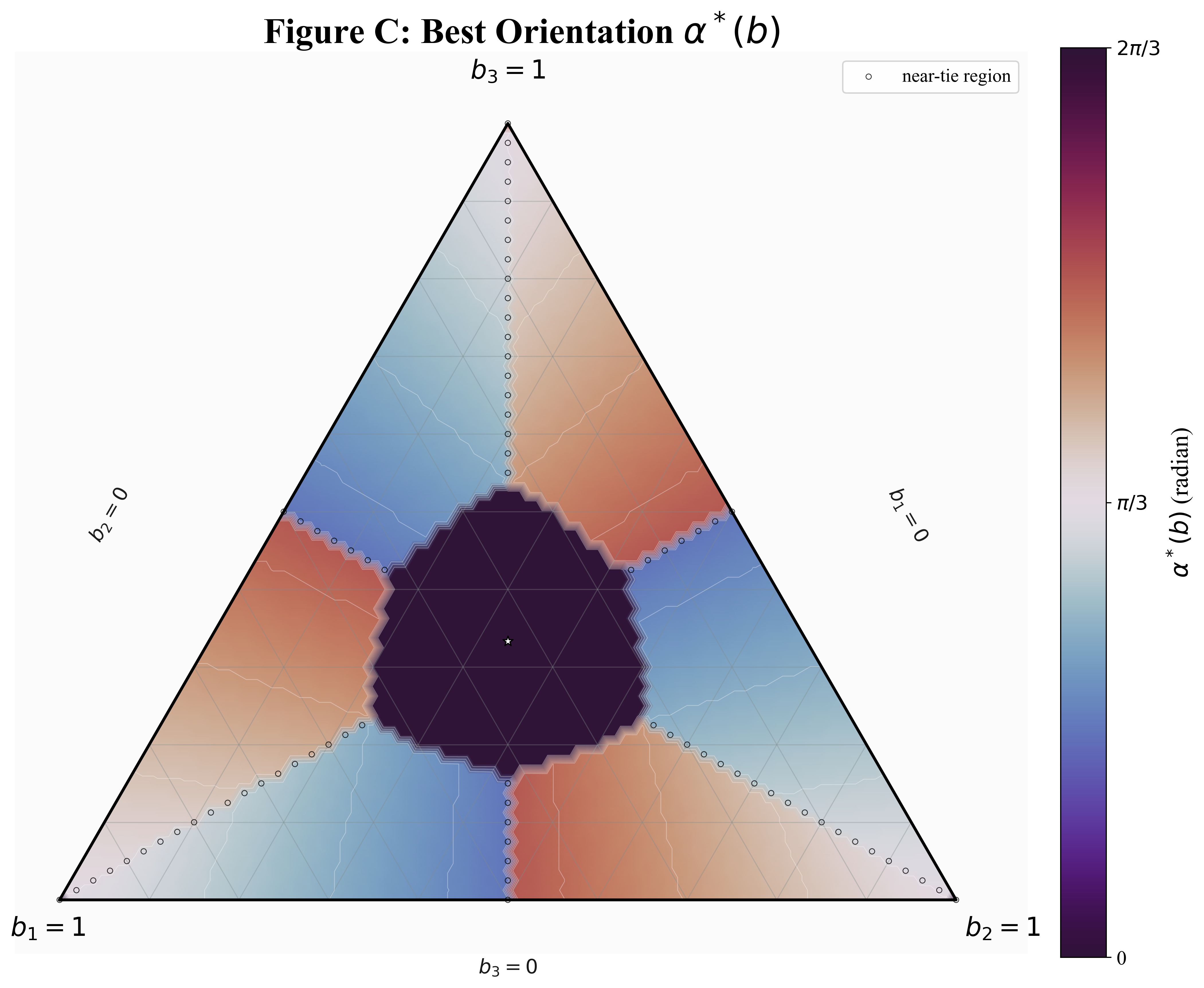}
    \caption{
    Best orientation \(\alpha^*(b)\) on the trine simplex. The maximizing orientation is arranged into symmetry-related sectors, with near-tie boundaries separating regions in which different measurement orientations are almost equally optimal.
    }
    \label{fig:trine-alpha-star}
\end{figure}

Figure~\ref{fig:trine-alpha-star} further shows the maximizing measurement orientation
\[
\alpha^*(b)
\in
\arg\max_{\alpha\in[0,2\pi/3)}
\sum_{o\in\{0,1,2\}}
\max_{i\in\{1,2,3\}}
\bigl[
b_i\,\ell_i(\alpha,o)
\bigr].
\]
The optimal orientation is organized into symmetry-related sectors across the simplex, rather than varying in an unstructured way. Near the center, one finds a broad region in which a common orientation remains optimal, while away from the center the preferred orientation changes according to which hypothesis is relatively favored by the current belief. The boundaries between these sectors are near-tie regions in which multiple orientations yield nearly identical one-step values. Thus the one-step optimization induces not only a scalar performance landscape but also a structured measurement-selection rule reflecting the geometry of the trine ensemble.

Taken together, Figures~\ref{fig:trine-j1star}--\ref{fig:trine-alpha-star} provide a compact one-step picture of the sequential trine discrimination problem. Figure~\ref{fig:trine-j1star} identifies where one-step performance is intrinsically high or low, Figure~\ref{fig:trine-gain} shows where an additional measurement is actually worth taking, and Figure~\ref{fig:trine-alpha-star} indicates which measurement orientation realizes that improvement. This prepares the ground for the next subsection, where we move beyond static maps and examine how representative posterior states are routed across the simplex after a single optimized measurement.

\subsection{Representative Posterior Routing and Sequential Interpretation}
We now move from the static one-step maps of Section~8.2 to the actual posterior routing induced by a single optimized measurement. For each representative belief state \(b\), we take the maximizing orientation \(\alpha^*(b)\) obtained in the previous subsection and compute the outcome-conditioned posterior states
\[
\tau\bigl(b,\alpha^*(b),o\bigr), \qquad o\in\{0,1,2\}.
\]
This produces a branching picture on the trine simplex: from a single starting belief, one optimized measurement generates three possible posterior destinations according to three possible outcomes, each weighted by its Born probability. The resulting routing geometry gives a more concrete interpretation of the innate `sequential' structure, because it shows how the posterior mass is redistributed across the simplex after the observation is revealed.

To summarize this structure, we examine five representative initial beliefs chosen from qualitatively different regions of the simplex: a nearly symmetric central belief, a quasi-binary edge regime, a near-certainty regime, a generic asymmetric interior regime, and an off-center interior point close to a switching boundary. These cases are not introduced as isolated examples, but rather as a compact qualitative atlas of the one-step posterior dynamics on the trine simplex. Enlarged case-by-case views and numerical diagnostics are deferred to the Appendix~\ref{sec:appendix-routing}.

\begin{figure}[t]
    \centering
    \includegraphics[width=\textwidth]{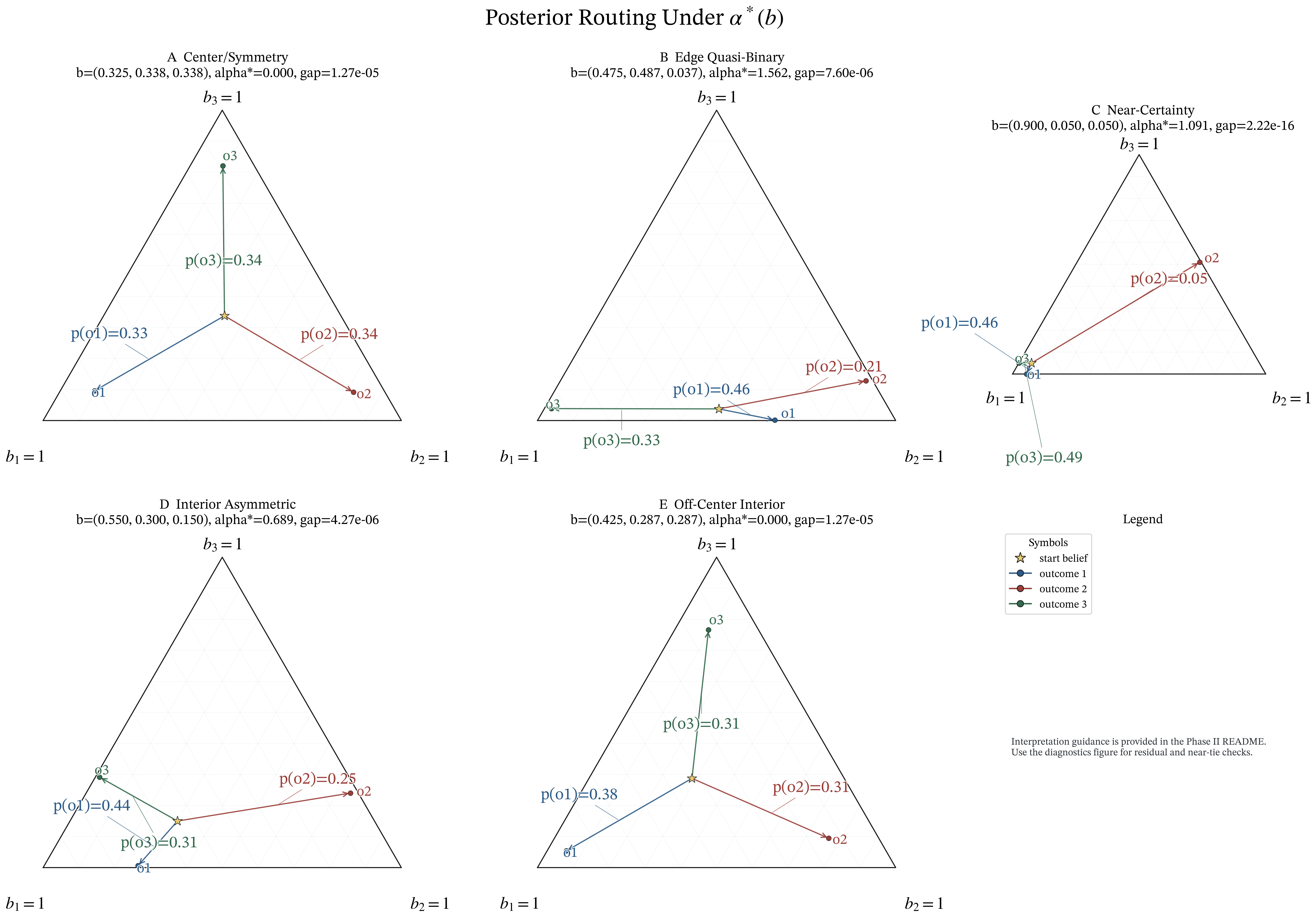}
    \caption{
    Posterior routing on the trine simplex for five representative belief states. In each panel, the star marks the starting belief, the arrows indicate the outcome-conditioned Bayesian updates under the optimized one-step measurement, and the terminal markers show the posterior landing points. The displayed branch probabilities indicate how likely each route is under the selected measurement orientation.
    }
    \label{fig:trine-routing-overview}
\end{figure}

Figure~\ref{fig:trine-routing-overview} summarizes the posterior-routing behavior for the five representative cases. The figure shows that one-step optimization is not merely a scalar comparison between stopping and measuring; it also induces a structured branching rule that depends strongly on the location of the current belief. In particular, different regions of the simplex generate qualitatively different posterior-routing patterns even under the same basic trine state discrimination situation.

Case A represents the nearly symmetric central regime. Here the three branch probabilities are close to balanced, and the three posterior landing points are routed toward different hypothesis-dominant regions in a nearly symmetric manner. This is the dynamic counterpart of the large gain observed near the center in Section~8.2. When the prior belief is highly ambiguous, one optimized measurement can meaningfully separate the three hypotheses by sending the posterior toward distinct sectors of the simplex. Thus the central high-gain region is not merely a static feature of the gain map; it corresponds dynamically to a balanced three-way information split after a single measurement.

Case B lies close to an edge of the simplex, where one hypothesis already has very small weight and the problem behaves almost like a binary discrimination task. In this quasi-binary regime, the routing remains concentrated near the corresponding edge rather than exploring the full two-dimensional interior. The measurement therefore acts primarily to separate the two leading hypotheses, with the third hypothesis remaining only weakly relevant throughout the update. This case makes precise the intuition that the trine problem continuously interpolates between genuinely ternary behavior in the interior and effectively binary behavior near the boundary.

Case C illustrates a near-certainty regime. The initial belief is already strongly concentrated on one hypothesis, and the optimized measurement therefore produces only a small incremental gain. Dynamically, this appears as a routing pattern in which the most likely branches keep the posterior close to the already dominant hypothesis, while only a low-probability branch produces a more substantial corrective shift. This behavior explains why the gain \(G(b)\) becomes small near the vertices even though the absolute one-step value \(J_1^*(b)\) itself remains high there. In other words, high absolute success probability near a certainty vertex does not imply that an additional measurement is especially valuable.

Cases D and E illustrate more delicate interior behavior. Case D is a generic asymmetric interior point, where the prior already favors one hypothesis but not overwhelmingly. In this regime the three outcome-conditioned branches no longer display full symmetry: some branches reinforce the currently favored hypothesis, whereas others provide a sharper shift toward a competing hypothesis. This shows that away from the symmetric center, the routing geometry begins to encode not only uncertainty level but also directional asymmetry in the current posterior.

Case E, by contrast, lies off-center but close to a switching region, so that the optimized orientation is more sensitive to small perturbations of the belief. The routing remains structured, but the associated measurement choice is closer to a boundary. This case is therefore useful for interpreting the transition regions already suggested by the optimal-orientation map in Section~8.2: near such regions, the posterior dynamics remain stable in the view of Bayesian updating, while the identity of the maximizing measurement can become comparatively delicate.

Taken together, these routing patterns complement the value and gain maps from Section~8.2 in an essential way. The one-step gain map indicates where measurement is worth taking, while the routing plot shows what that measurement actually does to the posterior once an outcome is observed. In particular, the central region is seen to support balanced three-way information splitting, edge regions collapse toward quasi-binary behavior, and near-certainty regions exhibit mostly self-reinforcing updates with only occasional corrective branches. These distinctions are already sufficient to reveal the sequential meaning of the trine POMDP: measurement is useful not only because it improves the expected success probability, but because it selectively transports posterior mass into different decision-relevant regions of the simplex.

This routing viewpoint also prepares the transition to the finite-horizon analysis below. Once one understands how a single optimized measurement redistributes belief mass over the simplex, the Bellman recursion may be interpreted as repeatedly composing such routed posterior updates with a stopping comparison at later stages. In this sense, the posterior-routing picture provides the most concrete bridge between the static one-step geometry of Section~8.2 and the genuinely sequential finite-horizon structure considered next.

\subsection{Finite-Horizon Bellman Structure and Compact Robustness Check}
To conclude the trine example, we briefly examine how the one-step geometry from Sections~8.2--8.3 is reorganized by a finite-horizon Bellman recursion, and we also add a compact numerical robustness check. Additional value maps, action maps, and supplementary robustness comparisons are collected in Appendix~\ref{sec:appendixA_trine_finite_horizon}.

For the horizon-two computation, let
\[
S(b):=\max\{b_1,b_2,b_3\}
\]
denote the immediate stopping value on the trine simplex, and write
\[
D_1(b):=V_1(b)-S(b),
\qquad
D_0(b):=V_0(b)-V_1(b).
\]
Here \(D_1(b)\) measures the continuation advantage at the last measurement stage, whereas \(D_0(b)\) measures the additional value of acting one stage before while still taking the next-stage optimization into account. Thus \(D_1\) is the natural finite-horizon analogue of the one-step gain map, while \(D_0\) reveals the genuinely sequential effect of the Bellman recursion.

\begin{figure}[t]
    \centering
    \includegraphics[width=0.49\textwidth]{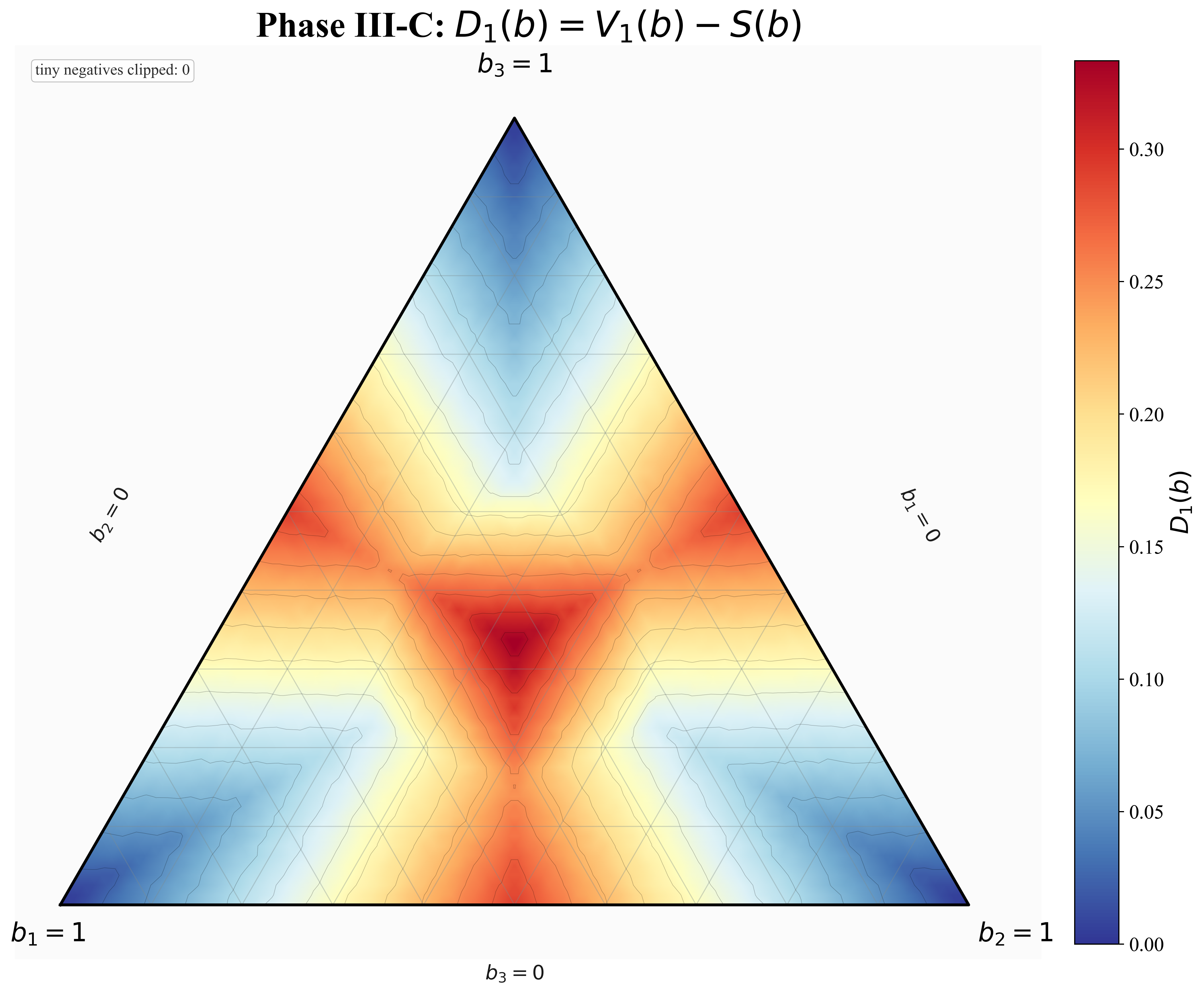}
    \hfill
    \includegraphics[width=0.49\textwidth]{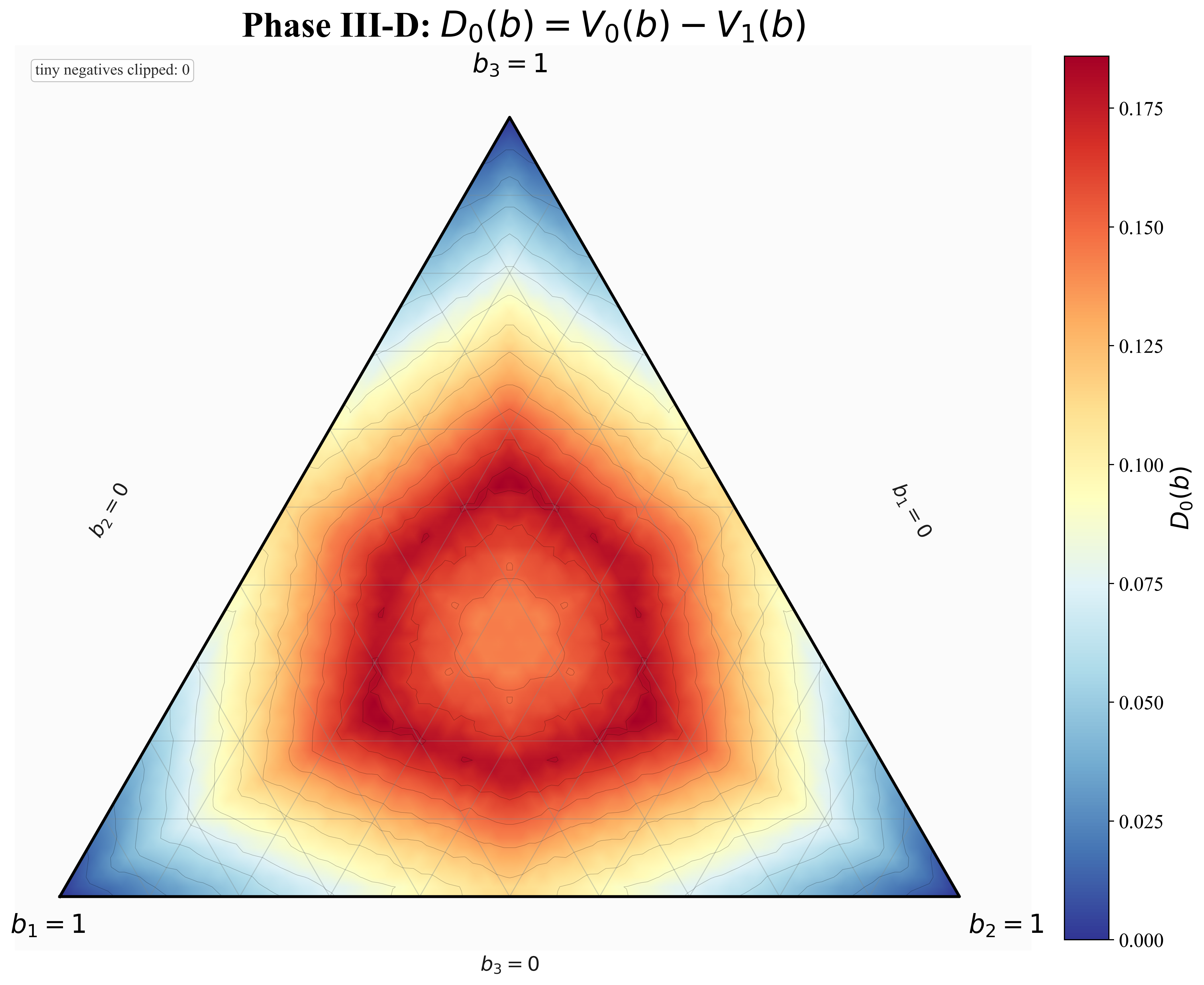}
    \caption{
    Finite-horizon continuation-value difference maps on the trine simplex. Left: \(D_1(b)=V_1(b)-S(b)\), the continuation advantage at the final measurement stage. Right: \(D_0(b)=V_0(b)-V_1(b)\), the additional value of acting one stage earlier while accounting for the next-step Bellman optimization.
    }
    \label{fig:trine-finite-horizon-D}
\end{figure}

Figure~\ref{fig:trine-finite-horizon-D} summarizes the finite-horizon Bellman structure in the most compact form. The left panel shows \(D_1\), which indicates where one further measurement is preferable to immediate stopping at the final nonterminal stage. Its geometry is closely related to the one-step gain structure seen in Section~8.2: continuation is most valuable in the more ambiguous interior region and becomes less useful as the belief approaches a certainty vertex. The right panel shows \(D_0\), which incorporates the effect of one additional Bellman step. This map is not merely a copy of \(D_1\); rather, the continuation-advantage landscape is reorganized by the possibility of future optimized action. In this sense, the comparison between \(D_1\) and \(D_0\) gives a compact visual summary of how the finite-horizon recursion reshapes the stop--measure tradeoff over time.

This distinction is the main sequential message of the horizon-two experiment. A one-step gain map tells us whether a measurement is worthwhile when only one measurement opportunity remains. By contrast, \(D_0\) reflects the fact that the current measurement changes not only the immediate posterior but also the set of future decision options available after that posterior update. The Bellman recursion therefore does not simply repeat the one-step geometry; it reweights and redistributes the continuation incentive according to the remaining horizon.

\begin{figure}[t]
    \centering
    \includegraphics[width=0.82\textwidth]{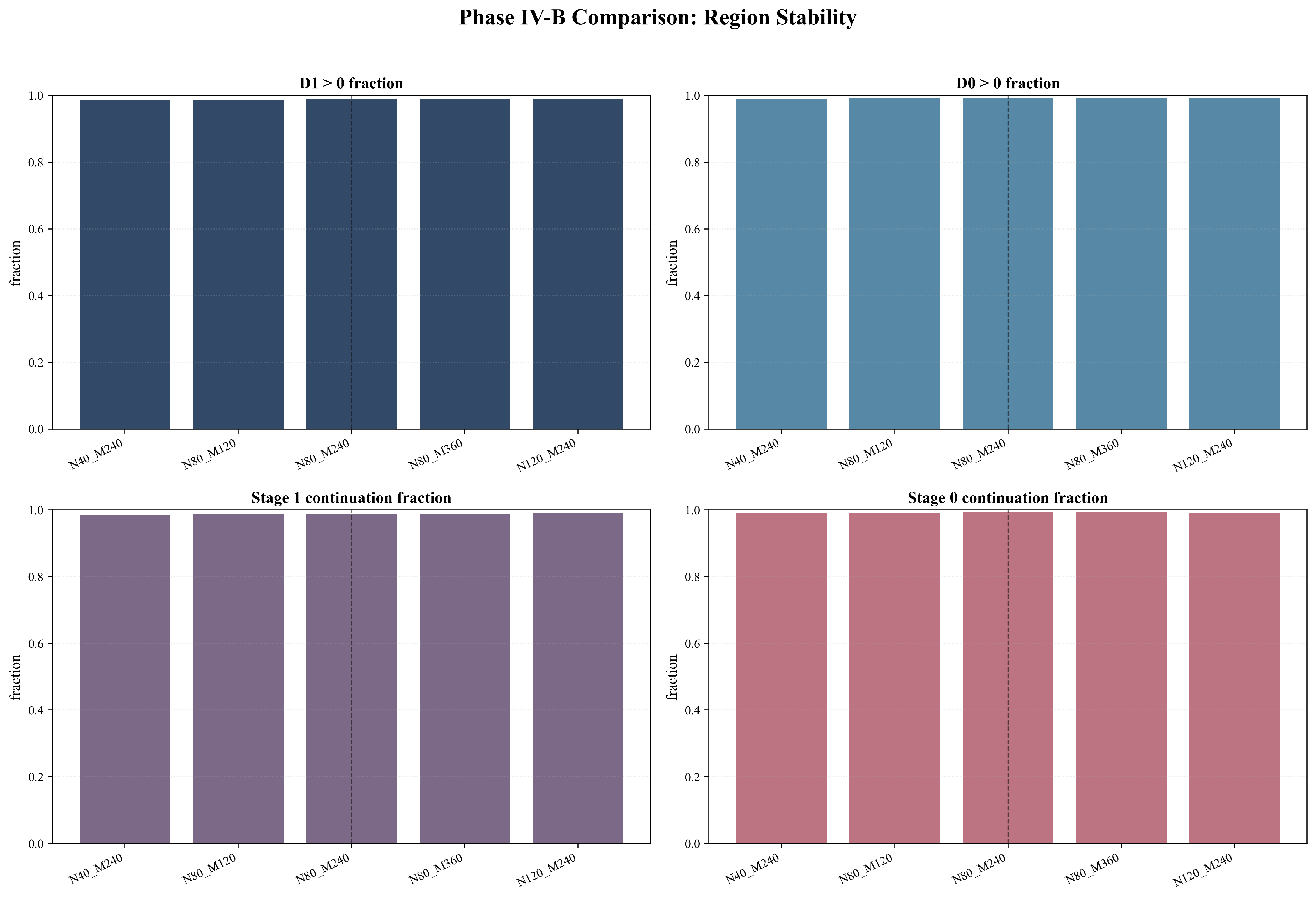}
    \caption{
    Compact robustness check for the finite-horizon continuation regions under discretization refinement. The overall continuation-region fractions remain qualitatively stable across the compared numerical settings.
    }
    \label{fig:trine-robust-region}
\end{figure}

Figure~\ref{fig:trine-robust-region} provides a compact robustness check for this finite-horizon picture. Under the compared discretization settings, the broad continuation pattern  remains qualitatively stable at the level of region fractions. For the purposes of the present numerical section, this supports the conclusion that the horizon-two Bellman structure observed above is not merely an artifact of a single plotting resolution, but a numerically stable qualitative feature of the trine example.

Taken together, these figures complete the numerical narrative of Section~8. The trine simplex first provided the geometric belief state space, then the one-step gain and posterior-routing pictures clarified the local value and dynamics of measurement, and finally the finite-horizon Bellman maps showed how that local structure is reorganized into a genuinely sequential decision architecture when horizon $H$ becomes larger than one.

\section{Conclusion}

In this paper, we formulated sequential quantum state discrimination as a
finite-horizon POMDP with static-hidden-state. This viewpoint makes the latent
hypothesis, the observation process induced by the Born rule, the posterior
belief dynamics, and the stopping decision part of a single dynamic decision
architecture much more explicit. In particular, the belief state emerges as a sufficient statistic
for the full action--observation history \cite{SmallwoodSondik1973}, allowing the sequential-QSD problem to
be written as a Bellman recursion on the belief simplex.

A central conceptual result is that this formulation remains fully consistent
with conventional minimum-error QSD. In the one-step special case, the
posterior-based stopping objective reduces exactly to the standard one-shot objective of measurement plus 
classical post-processing, and also to the guess-labeled POVM formulation. The present framework should
therefore be understood not as a change of discrimination criterion, but as a
sequential generalization in which measurement, Bayesian updating, and stopping
are all treated explicitly.

On the algorithmic and analytical side, we studied a projected dynamic programming
architecture based on a finite belief grid and a finite measurement library. We
derived separate error bounds for belief-space discretization and action-space
discretization, established recursive Lipschitz regularity for the value functions, obtained explicit admissible bounds for the
Lipschitz constants, and combined the two approximation mechanisms
into a total approximation error bound for the fully projected finite-grid
finite-library planner. These results show that the approximation architecture
has a transparent structure: local geometric discretization errors are amplified
by analytic sensitivity constants and then accumulated through the finite horizon.

We also analyzed the computational implications of the projected method. The
offline phase, which performs backward induction over the belief grid and
repeated projection of posterior beliefs, admits specific symbolic and
asymptotic complexity laws. Under simple projection implementations, this leads
to a clear curse-of-dimensionality interpretation through the geometry of the
belief simplex. By contrast, the online phase follows only one realized
trajectory and is naturally controlled by the stopping time rather than by the
full branching structure of the dynamic program. In this way, the distinction
between offline planning and online execution becomes mathematically explicit.

Finally, we illustrated the framework on concrete low-dimensional examples. In the binary case, the formulation recovers the familiar one-dimensional belief geometry and provides a direct consistency check with the conventional Helstrom picture. In the trine case, the numerical study makes the sequential structure visually explicit through the geometry of the belief simplex, the one-step gain landscape, posterior-routing patterns, and compact finite-horizon Bellman continuation regions. These examples are not intended as large-scale empirical benchmarks, but rather as interpretable demonstrations of how our general analytical framework manifests itself in concrete sequential-QSD instances.

Several natural directions remain for future work. On the numerical side, it would be valuable to extend the present illustrative binary and trine studies to broader computational experiments and empirical benchmarks for more diverse sequential-QSD instances. It is also natural to extend the framework toward more robust formulations in which the actually prepared state deviates from the intended candidate states due to state preparation errors. We hope that the POMDP perspective developed here provides a coherent analytical and computational framework for further work on sequential and adaptive quantum state discrimination.

\section*{Acknowledgements}

The authors gratefully acknowledge helpful discussions and intellectual support
during the development of this work. This work was supported by the
National Research Foundation of Korea (NRF) through a
grant funded by the Ministry of Science and ICT (Grants No.
RS-2025-00515537), the Institute
for Information \& Communications Technology Promotion
(IITP) grant funded by the Korean government (MSIP)
(Grants Nos. RS-2025-02304540 and RS-2019-II190003),
the National Research Council of Science \& Technology
(NST) (Grant No. GTL25011-401), and the Korea Institute
of Science and Technology Information (Grant No. P26028).
We acknowledge the Yonsei University Quantum Computing Project Group for providing support and access to the
Quantum System One (Eagle Processor), which is operated at
Yonsei University.

\section*{Code Availability}

The code used to conduct numerical experiments and generate figures reported in this work is publicly available in the project repository, \url{https://github.com/absolute-injury/POMDP_QSD}.

\section*{Competing Interests}

The authors declare that there are no competing interests.

\section*{Author Contributions}

J.J. contributed to this work, undertaking the primary responsibilities, including the development of the main ideas,
mathematical proofs, initial drafting, and revisions of the paper. D.J., H.J., and K.J. provided valuable feedback in shaping the
core ideas and overall direction of the work, while K.J. also
supervised the whole research. All authors discussed the results and contributed to the final paper.

% 참고문헌 스타일 설정 (plain, unsrt, abbrv 등 선택)
\bibliographystyle{plain} 

% 아까 만든 .bib 파일 이름 입력 (확장자 제외)
\bibliography{references}

\appendix

\section{Supplementary Figures for the Trine Example}
\label{sec:appendix-trine}

This appendix collects supplementary figures for the trine-state numerical example of Section~8. We separate the materials into posterior-routing figures corresponding to Section~8.3 and finite-horizon Bellman / robustness figures corresponding to Section~8.4.

\subsection{Posterior-Routing Diagnostics and Full-Size Case Figures}
\label{sec:appendix-routing}

For completeness, we include the numerical diagnostic plot used to validate the posterior-routing computation, together with full-size versions of the five representative routing cases discussed in Section~8.3.

\begin{figure}[h]
    \centering
    \includegraphics[width=0.96\textwidth,height=0.90\textheight,keepaspectratio]{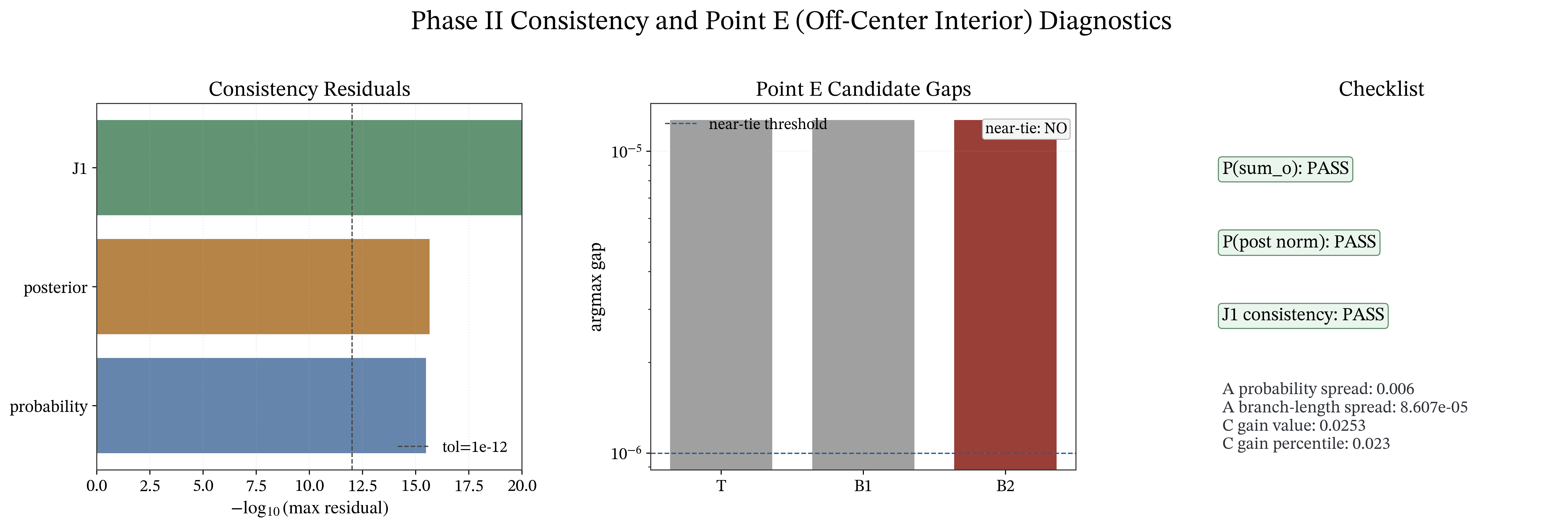}
    \caption{
    Numerical diagnostics for the posterior-routing computation. The figure checks normalization consistency, branch-probability consistency, and the numerical stability of the representative near-switching case.
    }
    \label{fig:trine-routing-diagnostics}
\end{figure}

\begin{figure}[h]
    \centering
    \includegraphics[width=0.96\textwidth,height=0.90\textheight,keepaspectratio]{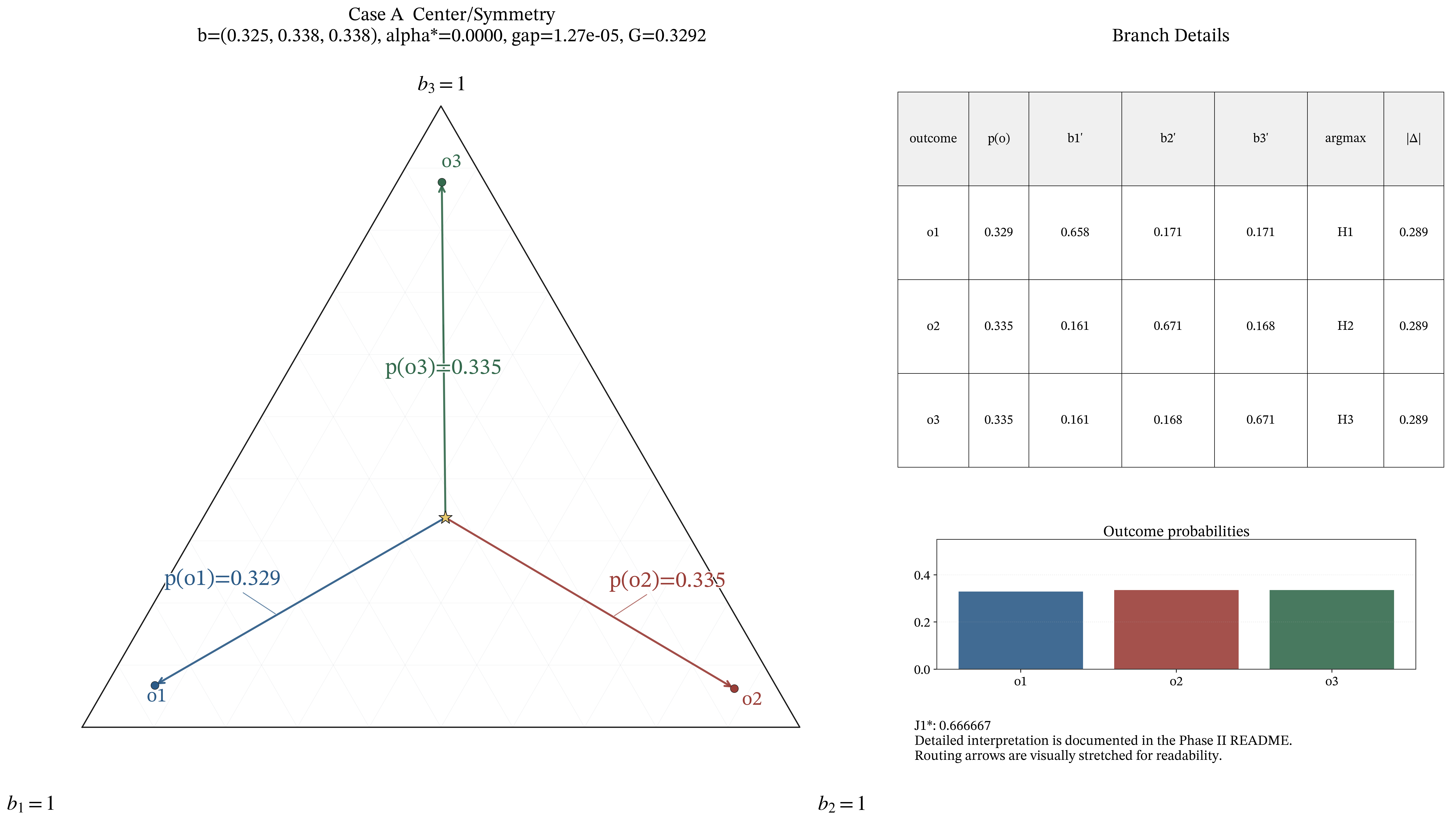}
    \caption{
    Full-size posterior-routing detail for Case A (nearly symmetric central regime).
    }
    \label{fig:appendix-routing-a}
\end{figure}

\begin{figure}[h]
    \centering
    \includegraphics[width=0.96\textwidth,height=0.90\textheight,keepaspectratio]{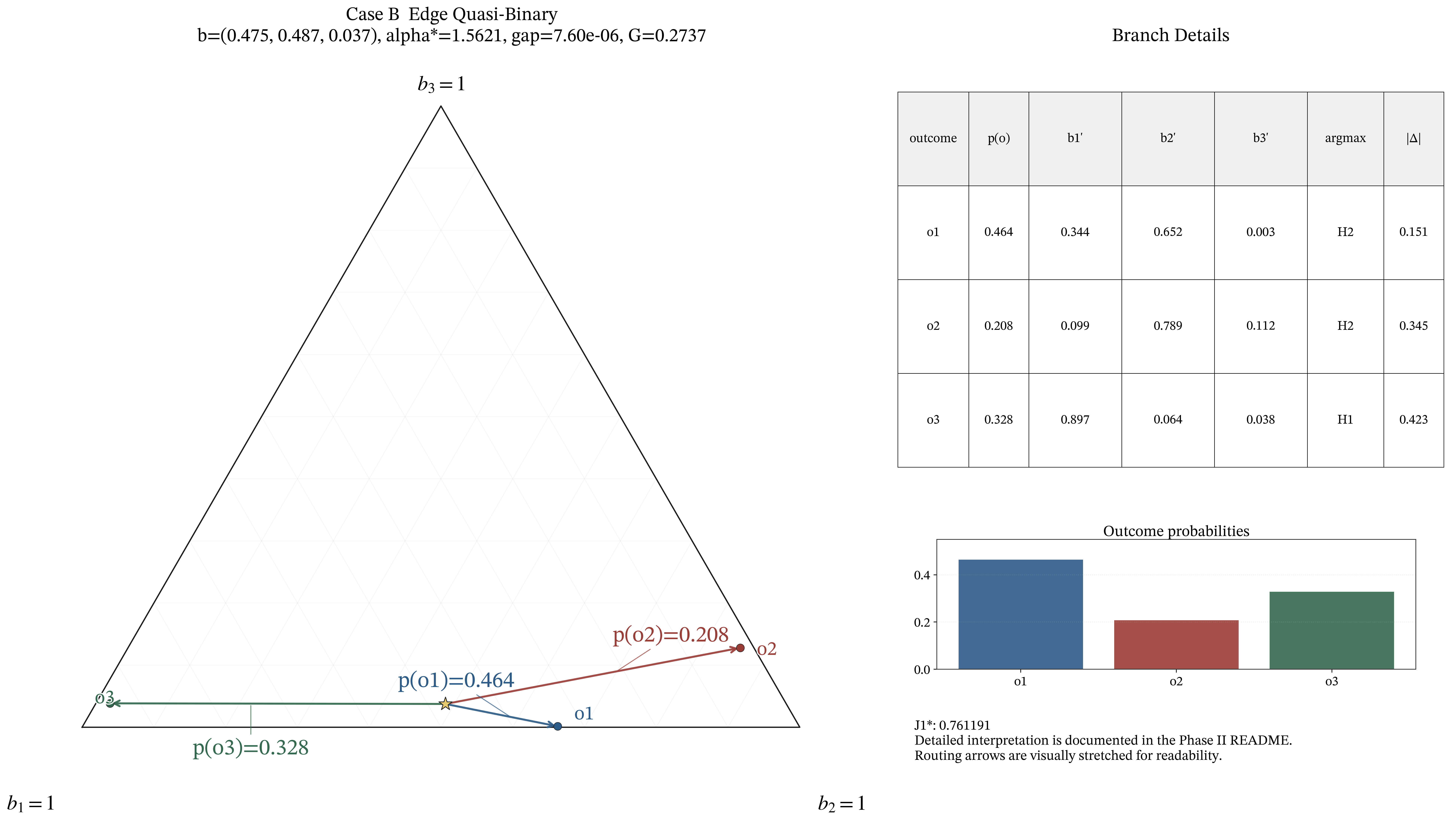}
    \caption{
    Full-size posterior-routing detail for Case B (quasi-binary edge regime).
    }
    \label{fig:appendix-routing-b}
\end{figure}

\begin{figure}[h]
    \centering
    \includegraphics[width=0.96\textwidth,height=0.90\textheight,keepaspectratio]{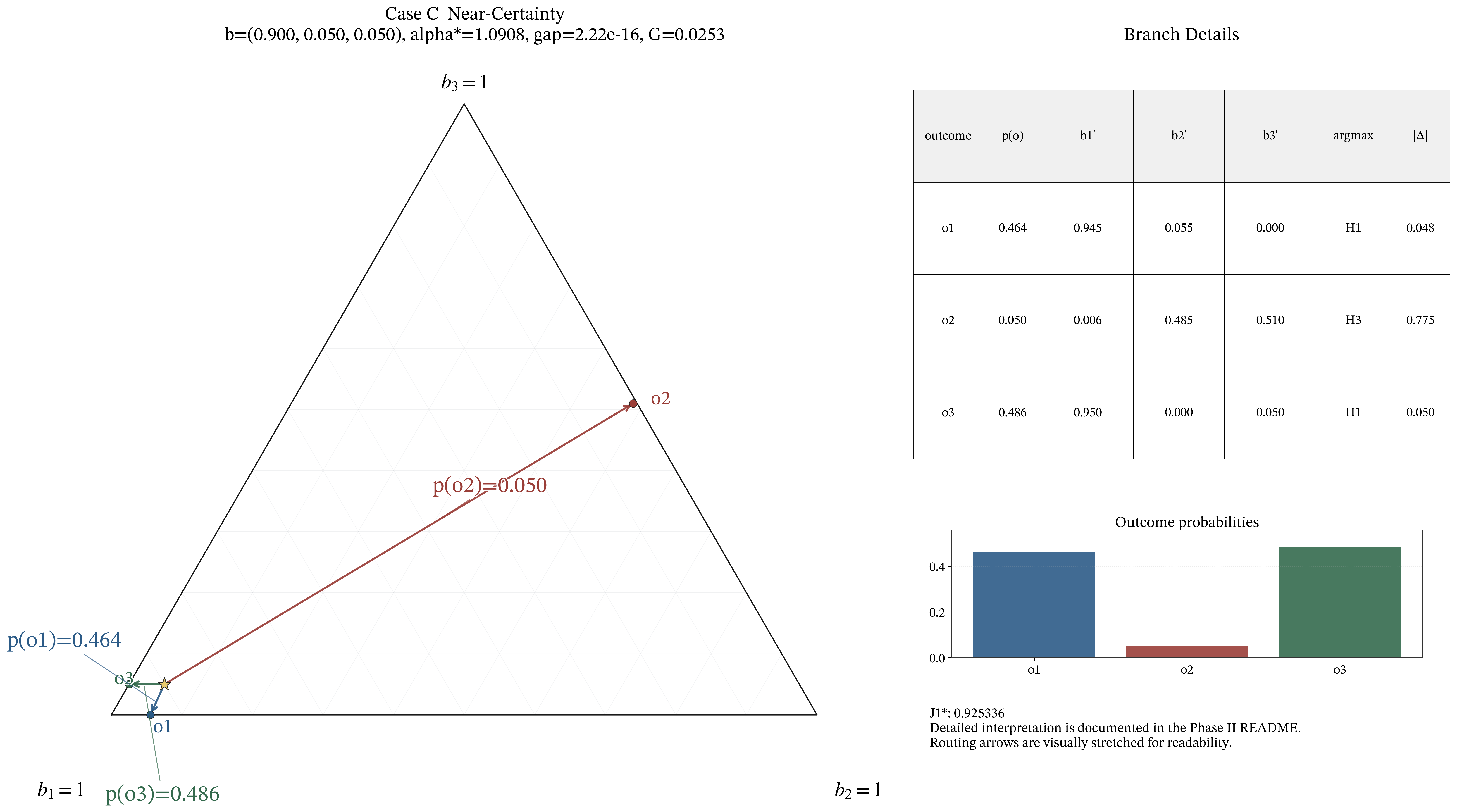}
    \caption{
    Full-size posterior-routing detail for Case C (near-certainty regime).
    }
    \label{fig:appendix-routing-c}
\end{figure}

\begin{figure}[h]
    \centering
    \includegraphics[width=0.96\textwidth,height=0.90\textheight,keepaspectratio]{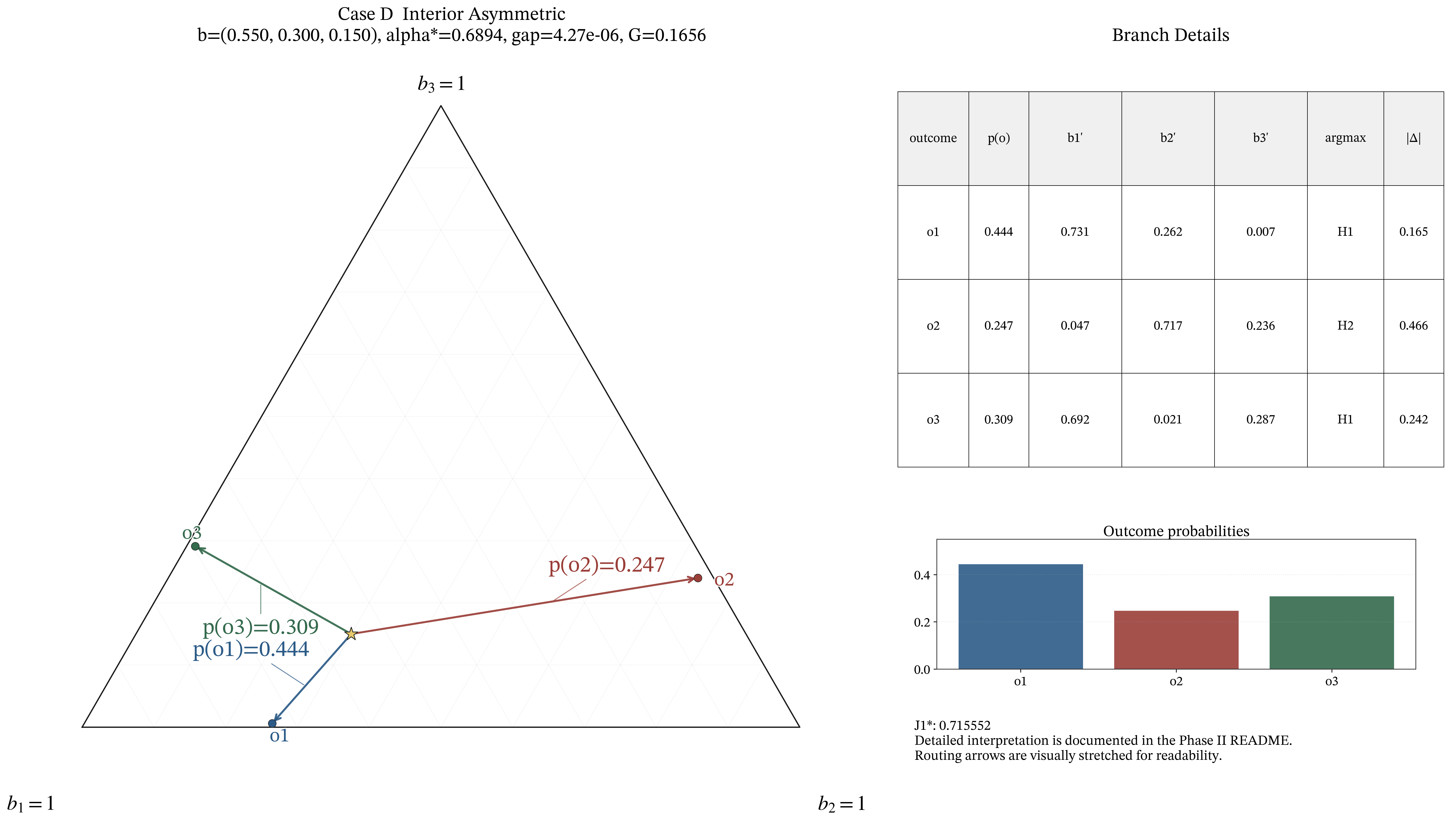}
    \caption{
    Full-size posterior-routing detail for Case D (generic asymmetric interior regime).
    }
    \label{fig:appendix-routing-d}
\end{figure}

\begin{figure}[h]
    \centering
    \includegraphics[width=0.96\textwidth,height=0.90\textheight,keepaspectratio]{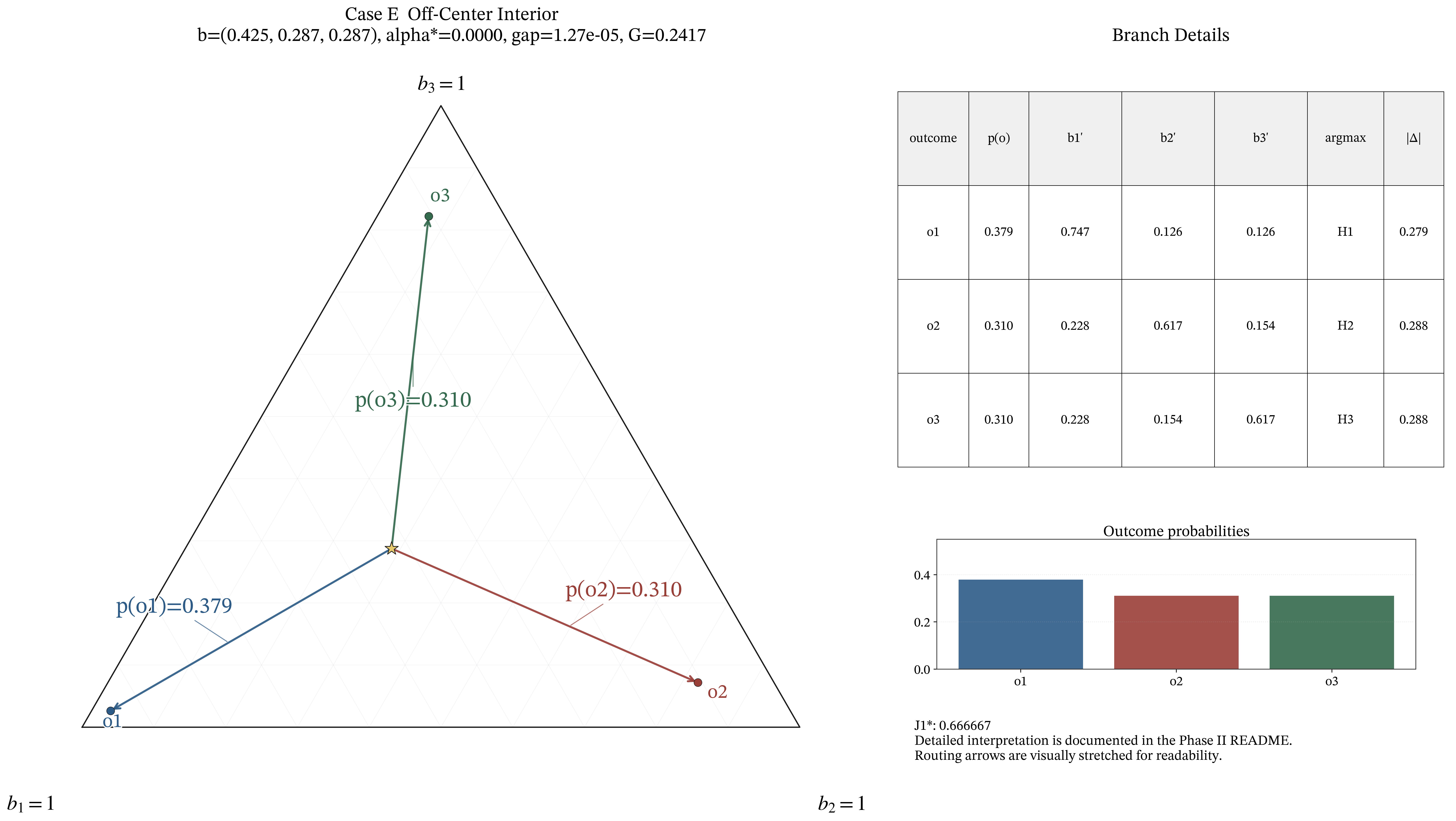}
    \caption{
    Full-size posterior-routing detail for Case E (off-center interior regime near a switching boundary).
    }
    \label{fig:appendix-routing-e}
\end{figure}

\FloatBarrier

\subsection{Additional Finite-Horizon Maps and Robustness Comparisons}
\label{sec:appendixA_trine_finite_horizon}

This subsection collects supplementary value maps, policy maps, and additional robustness comparisons for the finite-horizon trine example discussed in Section~8.4. These figures are not needed for the main numerical narrative, but they provide a fuller record of the horizon-two Bellman computation and of its stability under discretization refinement.

\begin{figure}[h]
    \centering
    \includegraphics[width=0.49\textwidth]{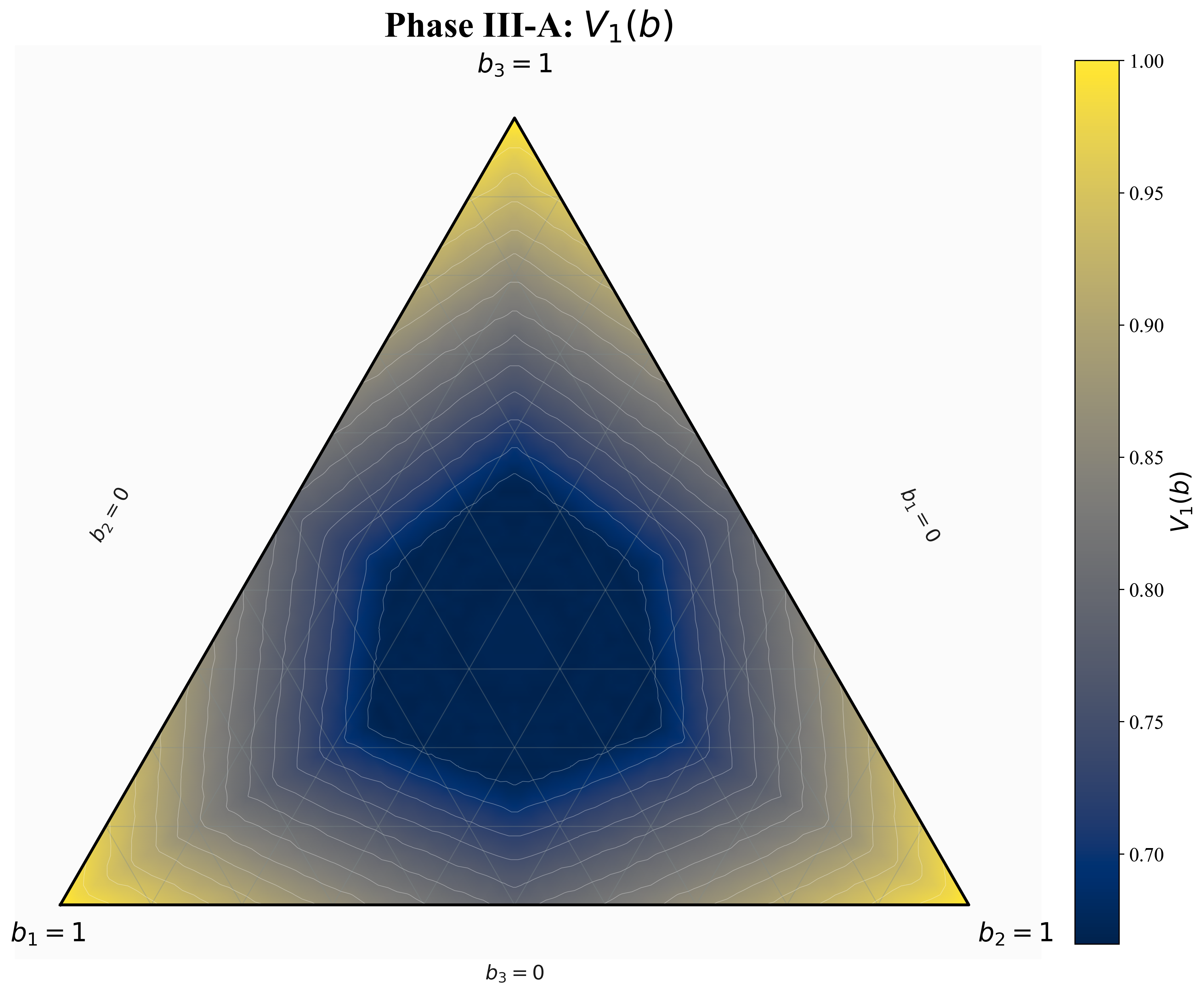}
    \hfill
    \includegraphics[width=0.49\textwidth]{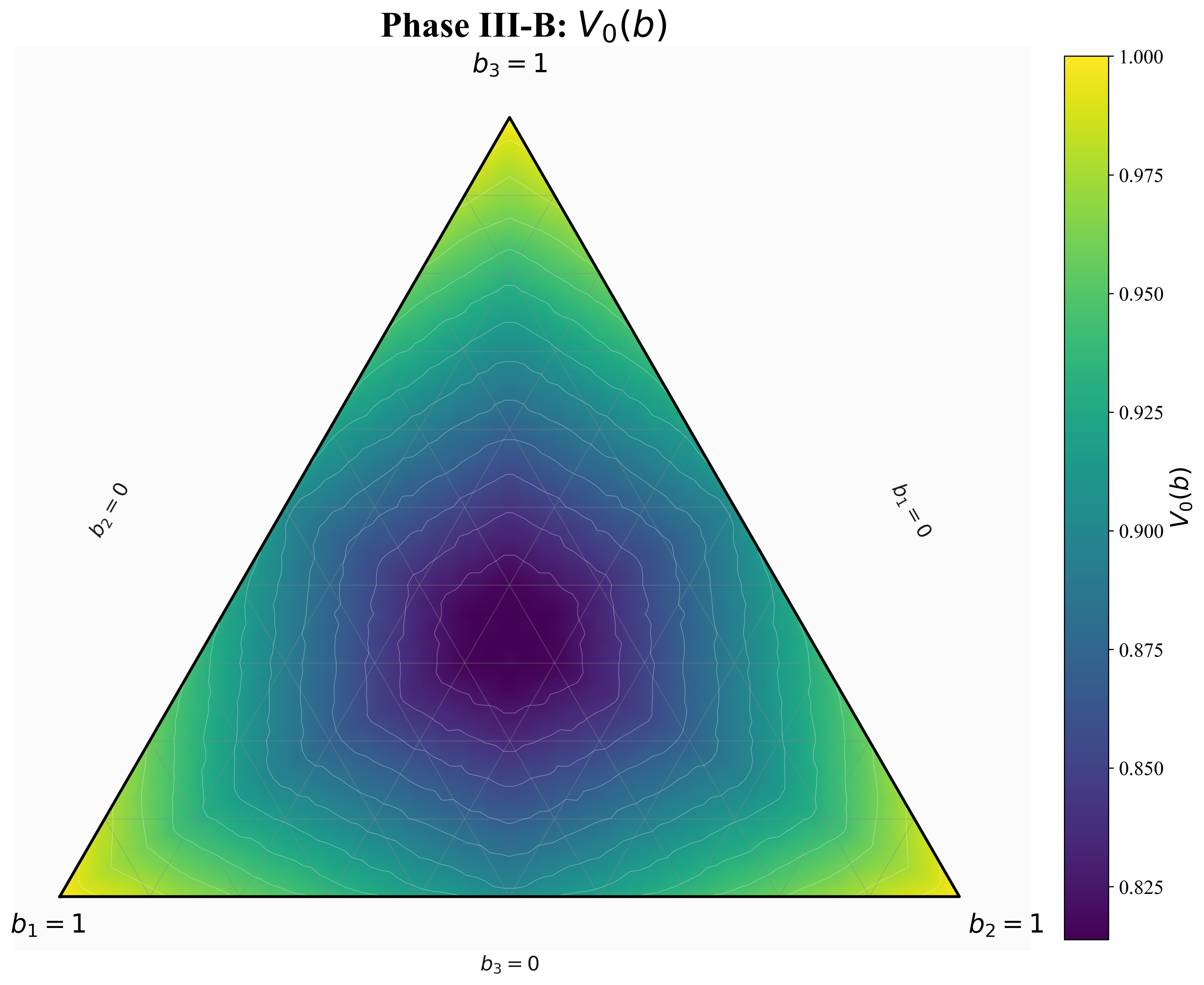}
    \caption{
    Supplementary finite-horizon value maps for the trine example. Left: \(V_1(b)\). Right: \(V_0(b)\). These absolute value maps complement the continuation-value difference plots shown in the main text.
    }
    \label{fig:appendix-trine-Vmaps}
\end{figure}

\begin{figure}[h]
    \centering
    \includegraphics[width=0.49\textwidth]{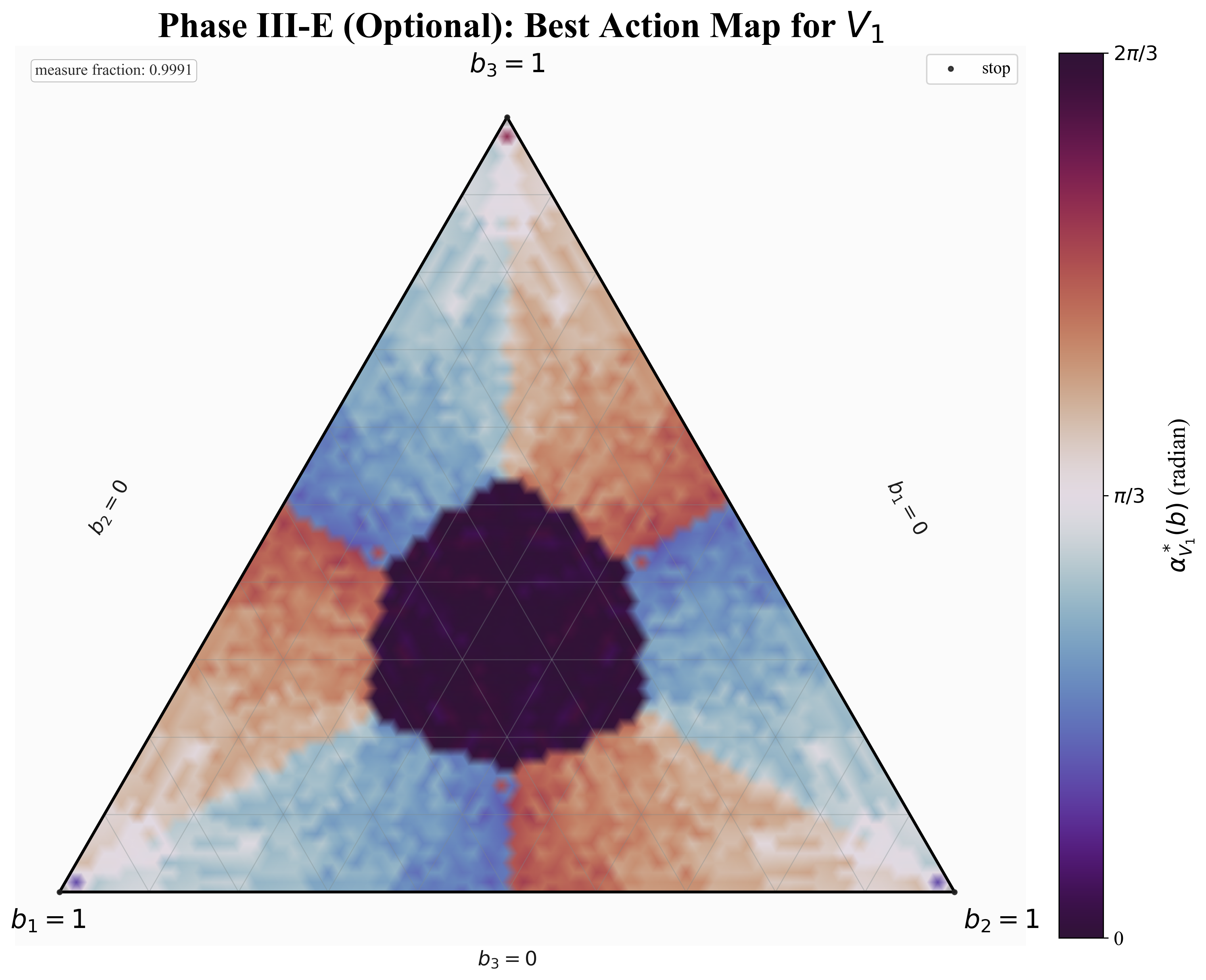}
    \hfill
    \includegraphics[width=0.49\textwidth]{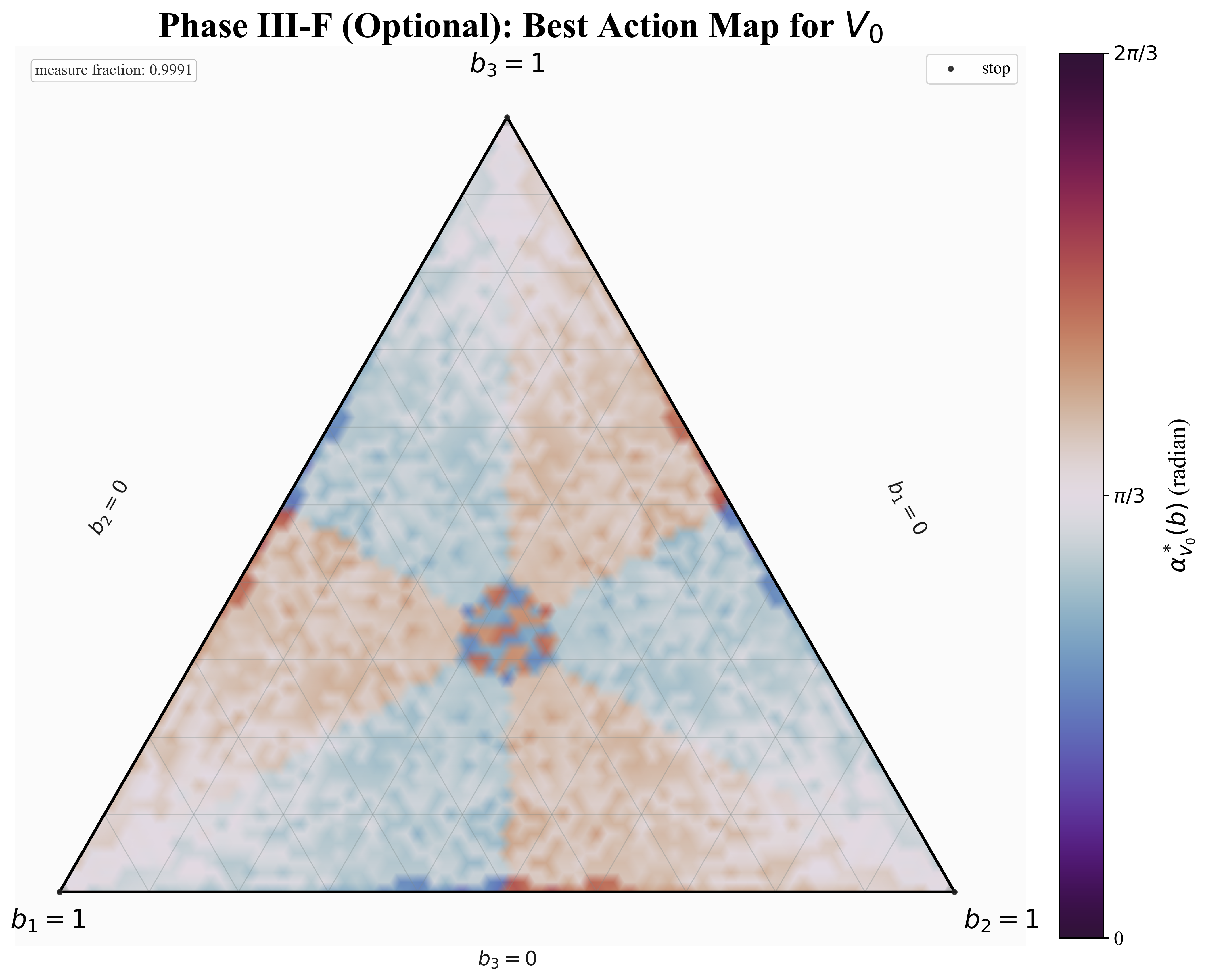}
    \caption{
    Supplementary finite-horizon policy maps for the trine example. Left: optimal stage-1 action map. Right: optimal stage-0 action map. These plots visualize how the stop--measure policy partition changes across the simplex as the remaining timestamp changes.
    }
    \label{fig:appendix-trine-actionmaps}
\end{figure}

\begin{figure}[h]
    \centering
    \includegraphics[width=0.72\textwidth]{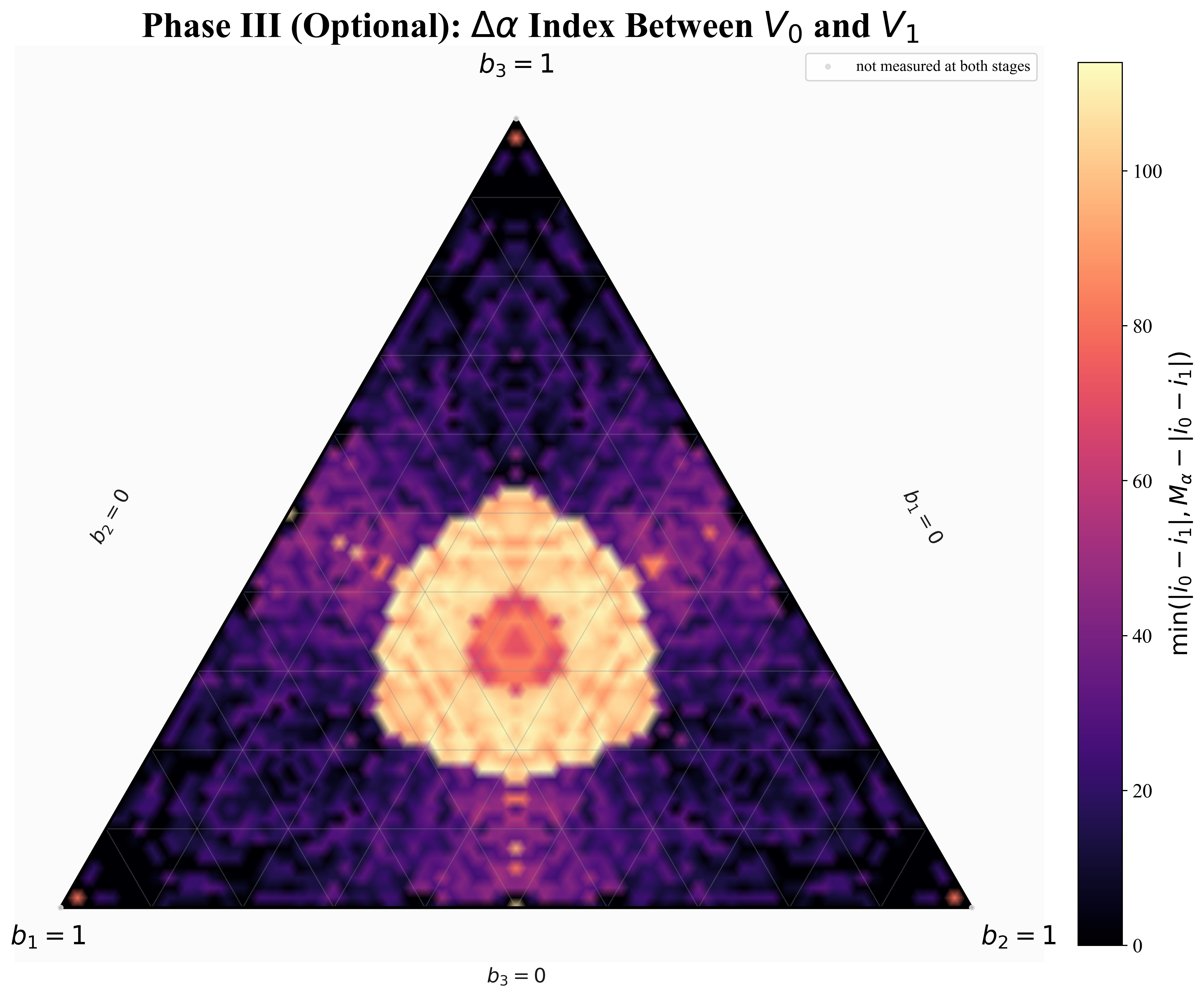}
    \caption{
    Supplementary measurement-index difference map for the finite-horizon trine example. The figure records how the selected measurement orientation index varies across the belief simplex as the remaining timestamp changes.
    }
    \label{fig:appendix-trine-alphaindex}
\end{figure}

\begin{figure}[h]
    \centering
    \includegraphics[width=0.49\textwidth]{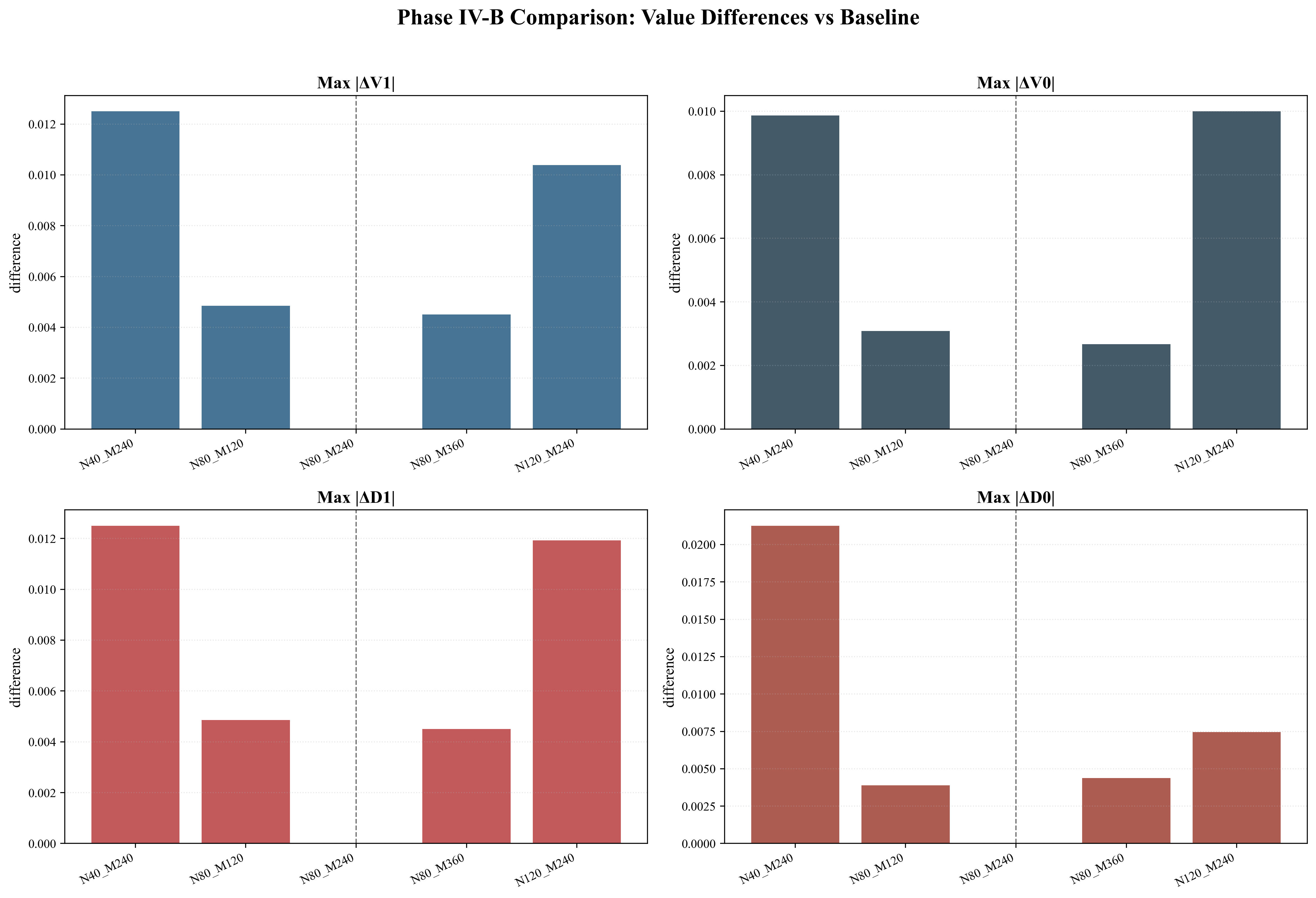}
    \hfill
    \includegraphics[width=0.49\textwidth]{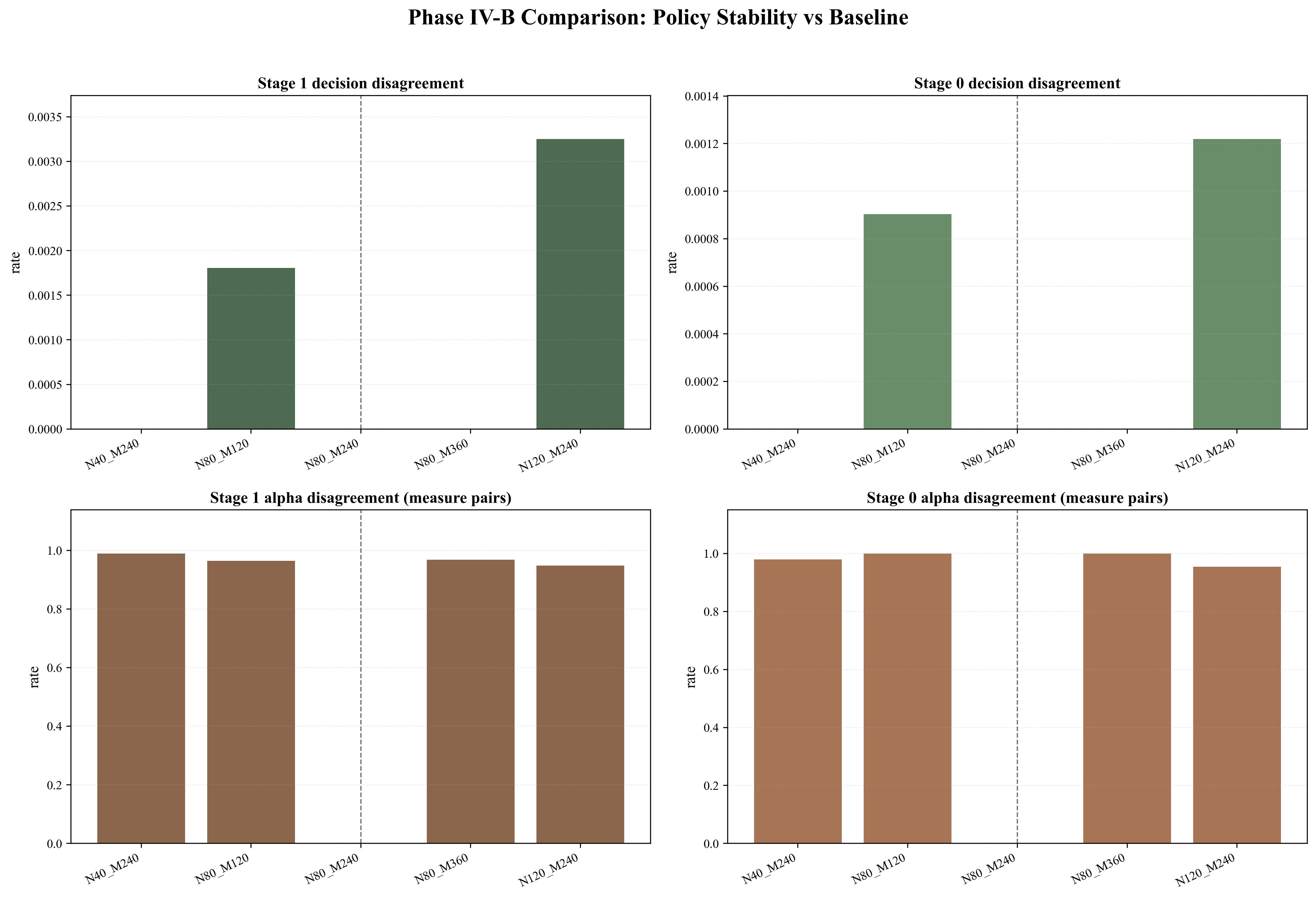}
    \caption{
    Supplementary robustness comparisons for the finite-horizon trine example. Left: value comparison under discretization refinement. Right: policy comparison under discretization refinement. These plots complement the continuation-region comparison shown in the main text.
    }
    \label{fig:appendix-trine-robustness}
\end{figure}

\end{document}